\documentclass[11pt]{article}

\usepackage{etoolbox}

\newbool{SODAfinal}
\setbool{SODAfinal}{false}

\ifbool{SODAfinal}{
\usepackage{ltexpprt}
\usepackage[keeplastbox]{flushend}
}{}

\usepackage[letterpaper,margin=1in]{geometry}
\usepackage{amsmath}

\ifbool{SODAfinal}{}{
\usepackage{amsthm}
}

\usepackage{amssymb}
\usepackage{booktabs}
\usepackage[skip=2pt,font=scriptsize]{caption}
\usepackage{microtype}

\ifbool{SODAfinal}{}{
\usepackage{times}
}

\usepackage[utf8]{inputenc}
\usepackage{url}
\usepackage{hyperref}
\hypersetup{colorlinks, linkcolor=darkblue, citecolor=darkgreen, urlcolor=darkblue}
\usepackage{graphicx}
\graphicspath{{graphics/}}
\usepackage[backgroundcolor=lightgray,linecolor=gray,textsize=scriptsize]{todonotes}
\usepackage{xfrac}
\usepackage{mathtools}
\usepackage[inline]{enumitem}
\usepackage{changepage}
\usepackage{pgf,tikz}
\usetikzlibrary{arrows}
\usetikzlibrary{decorations.pathreplacing}
\usepackage{tikz-3dplot}
\usepackage{mathrsfs}

\definecolor{darkblue}{rgb}{0,0,0.38}
\definecolor{darkgreen}{rgb}{0.1,0.35,0}

\ifbool{SODAfinal}{
\newcommand\qedhere{\hspace*{0mm}\hfill\ensuremath{\square}}
\newcommand\qed\qedhere
\AtEndEnvironment{proof}{\qedhere}

\newtheorem{definition}{Definition}[section]

\newenvironment{myproof}[1]{\medbreak\noindent\textit{#1.}}{\qedhere \medbreak}
}
{
\newtheorem{theorem}{Theorem}
\newtheorem{lemma}[theorem]{Lemma}

\newtheorem{proposition}[theorem]{Proposition}
\newtheorem{definition}[theorem]{Definition}

\newtheorem{corollary}[theorem]{Corollary}

}

\DeclareMathOperator{\conv}{conv}
\DeclareMathOperator{\cl}{cl}
\DeclareMathOperator{\diam}{diam}
\DeclareMathOperator{\flt}{Flt}
\DeclareMathOperator{\width}{width}
\DeclareMathOperator{\xc}{xc}

\DeclareMathOperator{\LPgapMax}{gap^+}
\DeclareMathOperator{\LPgapMin}{gap^-}

\DeclareMathOperator{\rdist}{rdist}
\DeclareMathOperator{\rd}{rdist}

\DeclareMathOperator{\af}{aff}

\DeclareMathOperator{\poly}{poly}
\DeclareMathOperator{\diag}{diag}

\def\R{\mathbb{R}}
\def\Z{\mathbb{Z}}
\def\calH{\mathcal{H}}
\def\calL{\mathcal{L}}
\def\calP{\mathcal{P}}
\def\zerovec{\mathbf{0}}

\def\eps{\varepsilon}
\def\t{^\intercal}
\def\onesvec{\mathbf{1}}

\def\match{\mathrm{P}_\mathrm{match}}
\def\pmatch{\mathrm{P}_\mathrm{pmatch}}
\def\cut{\mathrm{P}_\mathrm{cut}}
\def\tsp{\mathrm{P}_\mathrm{tsp}}
\def\stab{\mathrm{P}_\mathrm{stab}}
\def\knap{\mathrm{P}_\mathrm{knap}}
\def\vjoin{\mathrm{P}_\mathrm{vjoin}}
\def\vjoindom{\mathrm{P}_\mathrm{vjoin}\dom}
\def\oddcutdom{\mathrm{P}_\mathrm{ocut}\dom}
\def\oddcut{\mathrm{P}_\mathrm{ocut}}
\def\dom{^\uparrow}

\newcommand{\NPcomp}{\textsc{NP}}

\makeatletter
\newcommand*\bigcdot{\mathpalette\bigcdot@{.7}}
\newcommand*\bigcdot@[2]{\mathbin{\vcenter{\hbox{\scalebox{#2}{$\m@th#1\bullet$}}}}}
\makeatother

\title{Lifting Linear Extension Complexity Bounds\\ to the Mixed-Integer Setting\ifbool{SODAfinal}{\thanks{The third author was supported by the Swiss National Science Foundation grant 200021\_165866.}}{}} 
\author{
Alfonso Cevallos\thanks{
Department of Mathematics, ETH Zurich, Zurich, Switzerland.
Email: \href{mailto:alfonso.cevallos@ifor.math.ethz.ch}
{alfonso.cevallos@ifor.math.ethz.ch}.
}
\and 
Stefan Weltge\thanks{
Department of Mathematics, ETH Zurich, Zurich, Switzerland.
Email: \href{mailto:stefan.weltge@ifor.math.ethz.ch}
{stefan.weltge@ifor.math.ethz.ch}.
}
\and
Rico Zenklusen\thanks{
Department of Mathematics, ETH Zurich, Zurich, Switzerland.
Email: \href{mailto:ricoz@math.ethz.ch}
{ricoz@math.ethz.ch}.
\ifbool{SODAfinal}{}{Supported by the Swiss National Science Foundation grant
200021\_165866.}
}
}

\ifbool{SODAfinal}{
\date{}
}{}

\begin{document}

\maketitle

\ifbool{SODAfinal}{
\fancyfoot[R]{\footnotesize{\textbf{Copyright \textcopyright\ 2018 by SIAM\\
Unauthorized reproduction of this article is prohibited}}}
}{}

\begin{abstract}
\ifbool{SODAfinal}{\small\baselineskip=9pt}{}
Mixed-integer mathematical programs are among the most commonly used models for a wide set of problems in Operations Research and related fields. However, there is still very little known about what can be expressed by \emph{small} mixed-integer programs. In particular, prior to this work, it was open whether some classical problems, like the minimum odd-cut problem, can be expressed by a compact mixed-integer program with few (even constantly many) integer variables.
This is in stark contrast to linear formulations, where recent breakthroughs in the field of extended formulations have shown that many polytopes associated to classical combinatorial optimization problems do not even admit approximate extended formulations of sub-exponential size.

We provide a general framework for lifting inapproximability results of extended formulations to the setting of mixed-integer extended formulations, and obtain almost tight lower bounds on the number of integer variables needed to describe a variety of classical combinatorial optimization problems.
Among the implications we obtain, we show that any mixed-integer extended formulation of sub-exponential size for the matching polytope, cut polytope, traveling salesman polytope or dominant of the odd-cut polytope, needs $\Omega(n/\log n)$ many integer variables, where $n$ is the number of vertices of the underlying graph.
Conversely, the above-mentioned polyhedra admit polynomial-size mixed-integer formulations with only $O(n)$ or $O(n \log n)$ (for the traveling salesman polytope) many integer variables.

Our results build upon a new decomposition technique that, for any convex set $C$, allows for approximating any mixed-integer description of $C$ by the intersection of $C$ with the union of a small number of affine subspaces.

\end{abstract}
 
{\small
\textbf{Keywords:} extension complexity, mixed-integer programs, extended formulations
}

\section{Introduction}
Mixed-integer linear extended formulations (MILEFs) are one of the most common models to mathematically describe a wide variety of problems in Operations Research and related fields. 
This is due to their high expressive power, which made them the tool of choice for numerous real-world optimization problems, 
and also led to a large ecosystem of commercial solvers and modeling languages supporting mixed-integer models.
Despite their prevalence, the relation between what can be expressed by mixed-integer formulations and the number of integer variables used remains badly understood.
In particular, it is open how many integer variables are needed to obtain a compact, i.e., polynomial-size, MILEF for classical combinatorial objects including matchings, traveling salesman tours, cuts, stable sets, vertex covers, and odd cuts, just to name a few.
Moreover, there are natural problem classes beyond classical combinatorial optimization problems, for which efficient algorithms are known and yet, prior to this work, it was open whether they could as well be solved efficiently via a MILEF with a very small number of integer variables (maybe even just constantly many), through the use of Lenstra's Algorithm~\cite{Lenstra1983}. Bimodular integer programming is such an example (we expand on this in Section~\ref{secApplicationBimodular}). Whereas MILEFs are mostly used to describe hard problems, the hope to find such efficiently solvable MILEFs is driven by the desire to cast efficiently solvable problems, for which only specialized procedures are known, into a common, well-studied framework, for which there are moreover strong solvers available.
The only prior result on lower bounds for the number of integer variables in MILEFs, shows that any compact formulation of the matching polytope of a complete graph with $n$ vertices needs $\Omega(\sqrt{\sfrac{n}{\log n}})$ integer variables~\cite{HildebrandWZ2017}. Unfortunately, the presented technique is highly specialized to the matching polytope, heavily exploiting its facet structure. Moreover, this lower bound leaves a large gap compared to the canonical description of all matchings using $\Theta(n^2)$ many integer variables, one for each edge.

The goal of this work is to address this lack of understanding of the expressive power of MILEFs as a function of the number of constraints and integer variables used, by presenting a general framework to lift linear extension complexity results for approximate extensions to the mixed-integer setting.
To better put our results into context, we start with a brief summary of some basics on linear extensions, which also allows us to introduce notations and terminology to be used later, and formalize the notion of MILEFs.

\medskip

The typical settings in discrete optimization for which linear or mixed-integer models are developed ask about optimizing a linear function over some set of vectors $\mathcal{X} \subseteq \mathbb{Z}^d$, or often even just $\mathcal{X} \subseteq \{0,1\}^d$. The set $\mathcal{X}$ could for example correspond to the characteristic vectors of all matchings in a graph $G=(V,E)$, in which case $d=|E|$. Clearly, such problems can be restated as optimizing a linear function over the \emph{corresponding polytope} $P\subseteq \mathbb{R}^d$, which is simply the convex hull of the points in $\mathcal{X}$, i.e., $P=\conv(\mathcal{X})$. For $\mathcal{X}$ being all matchings of a graph, $P$ would therefore correspond to the matching polytope.
Hence, a discrete optimization problem gets described by a linear program, which is algorithmically well understood.
Ideally, to solve a linear program over $P$, we would like to have an inequality description of $P$, i.e., $P=\{x\in \mathbb{R}^d \mid Ax \leq b\}$.
Unfortunately, for many polytopes that arise in combinatorial optimization, such descriptions require an exponential number of inequalities. However, some of these polytopes admit much smaller descriptions if we allow the use of ``additional variables'', i.e., we allow describing $ P $ as
\begin{equation}
    \label{eqIntroEf}
    P = \{ x \in \R^d \, | \, \exists y \in \R^\ell : Ax + By \le b \}\enspace,
\end{equation}
where $ Ax + By \le b $ is a system of (preferably few) linear inequalities. The description $Q=\{(x,y) \in \mathbb{R}^d \times \mathbb{R}^\ell \mid Ax + By \leq b\}$ is called an \emph{extended formulation} of $P$. It allows for stating the original problem as a linear program over the solutions of $Q$.
Understanding which polytopes admit small extended formulations is the scope of the field of \emph{extended formulations} and we refer to~\cite{ConfortiCZ2013,Kaibel2011} and~\cite[Chap.~4]{ConfortiCZ2014} for many
examples and background material.
Formally, the \emph{extension complexity} $\xc(P)$ of a polyhedron $P$ is the minimum number of facets of an extended formulation $Q$ of $P$.
Clearly, compact (i.e., polynomial-size) extended formulations are desirable since they allow for rephrasing the original problem as a small linear program.

Whereas the above definition of an extended formulation $Q$ of $P$ requires $P$ to be an axis-parallel projection of $Q$, one can lift this restriction and allow for $P$ to be some affine image of $Q$. This generalization, which we use in this paper for convenience, can easily be seen not to have any impact to the notion of extension complexity.
More formally, we say that a polyhedron $Q$ is a linear extended formulation (LEF) of $P$ if there exists an affine map $\pi$ such that $P=\pi(Q)$.
Moreover, the \emph{size} of a LEF is equal to the number of its facets.

The study of linear extensions has received considerable attention recently due to breakthrough results stating that, for various prominent polytopes that arise in combinatorial optimization, the number of inequalities in every description of
type~\eqref{eqIntroEf} is super-polynomial in $ d $, see,
e.g.,~\cite{FioriniMPTW2015,Rothvoss2014,KaibelW2015,AvisT2013,PokuttaV2013,GoosJW2016}.

\medskip

This situation changes dramatically if we further allow for imposing ``integrality constraints'', which leads to the notion
of mixed-integer linear extended formulations.
In this setting, we describe $ P $ as
\ifbool{SODAfinal}{
  \begin{alignat}{2}
   P = \conv \Big\{ & x \in \R^d \; | \; \exists y \in \R^\ell, \, z \in \Z^k : & \\
		    & Ax + By + Cz \le b & \Big\}\enspace. \nonumber
  \end{alignat}
}{
\begin{equation} \label{eqIntroMilef}
    P = \conv \left(\left\{ x \in \R^d \; \middle\vert \; \exists y \in \R^\ell, \, z \in \Z^k : Ax + By + Cz \le b \right\}\right)\enspace.
\end{equation}
}
In other words, $P$ is described by some polyhedron $Q=\{(x,y,z)\in \mathbb{R}^d \times \mathbb{R}^\ell \times \mathbb{R}^k \mid Ax + By + Cz \leq b\}$ that is intersected with $k$ integrality constraints, projected down to a subset of the coordinates, and the convex hull of the resulting set is finally considered. 
The complexity of such a description of $P$ is now captured by $2$ parameters, namely the number $m$ of facets of $Q$---which is again called the \emph{size} of $Q$---and the number $k$ of integer variables.
As in the case of linear extensions, we can lift the restriction of the projection being axis-parallel, and allow for imposing integrality constraints on affine forms, 
without any impact on the values for $m$ and $k$ that can be achieved. 
Formally, we say that a polyhedron $Q\subseteq \mathbb{R}^\ell$ is a \emph{MILEF of complexity $(m,k)$} of $P\subseteq \mathbb{R}^d$, 
if the number of facets of $Q$ is at most $m$, the number of integrality constraints $k$, 
and there are affine maps $\sigma:\mathbb{R}^\ell \rightarrow \mathbb{R}^k$ and $\pi:\mathbb{R}^\ell \rightarrow \mathbb{R}^d$ such that
\ifbool{SODAfinal}{
\begin{align}
P &= \conv\left( \left\{ \pi(x) \;\middle\vert\; x\in Q, \sigma(x)\in \mathbb{Z}^k \right\} \right) \label{eq:defMILEF} \\
  &= \conv\left( \pi \left(Q\cap \sigma^{-1}(\mathbb{Z}^k)\right) \right)\enspace. \nonumber
\end{align}
}{
\begin{equation}\label{eq:defMILEF}
P = \conv\left( \left\{ \pi(x) \;\middle\vert\; x\in Q, \sigma(x)\in \mathbb{Z}^k \right\} \right)
  = \conv\left( \pi \left(Q\cap \sigma^{-1}(\mathbb{Z}^k)\right) \right)\enspace.
\end{equation}
}
When we need to be specific, we also say that the triple $(Q,\sigma, \pi)$ is a MILEF of $P$.
Again, a description of $P$ in terms of a MILEF (with corresponding maps $\pi$ and $\sigma$), allows for reducing any linear programming problem over $P$ to the problem of maximizing a linear function over the mixed-integer set $\pi(Q\cap \sigma^{-1}(\mathbb{Z}^k))$.

Even very difficult structures, like stable sets, or sets over which we can efficiently optimize but for which no small extended formulation exists, like matchings, can easily be described as small mixed-integer (or even just integer) formulations. 
On the downside, mixed-integer linear programs are a considerably harder problem class than their linear counterparts. The currently fastest algorithms, in terms of dependence on the number of integer variables $k$, have a running time dependence of $k^{O(k)}$~\cite{Lenstra1983}; hence, one needs $k=O(\sfrac{\log d}{\log\log d})$ for this expression to become polynomial, where $d$ corresponds to the input size of the original problem.
Thus, for hard problems, like maximum stable sets, we do not expect that MILEFs exist with both small size and few integer variables.

However, the mentioned achievements in the field of extended formulations did not seem to give rise to general techniques for obtaining lower bounds on the number of integer variables in the mixed-integer setting.
In this work, we demonstrate how recent generalizations of results in extended formulations, namely inapproximability results for extended formulations, can be leveraged in a general way to obtain lower bounds on $m$ and $k$ for any MILEF of complexity $(m,k)$ for the problem in question. Actually, the lower bounds we obtain even hold for any MILEF that is a close approximation (in a sense we will define formally later) of the problem under consideration.
As a consequence, we can show close-to-optimal lower bounds on $k$ for any compact (approximate) MILEF of classical polytopes like the matching polytope, the cut polytope, and the dominant of the odd-cut polytope.

 \subsection{Main results and consequences}
To exemplify the type of results we can derive from our technique, and to highlight its breadth, we first state some hardness results for classical combinatorial problems and for an efficiently solvable class of integer programs (bimodular integer programs, which we will formally define in Section~\ref{secApplicationBimodular}), and later discuss our general framework allowing to derive these and further results as consequences.
To this end, let $ K_n $ denote the complete undirected graph on $ n $ vertices.
\begin{samepage}
\begin{theorem}\label{thm:keyConsequences1}
    There is a constant $c>0$ such that the following holds.
    Let $n\in \mathbb{Z}_{\geq 1}$, and let $m,k >0$ such that there exists a MILEF with complexity $(m,k)$ with $m\leq 2^{c\cdot n}$ for either:
    \begin{itemize}[itemsep=-0.4em,topsep=0.2em]
    \item the matching polytope of $ K_n $,
    \item the dominant of the odd-cut polytope of $ K_n $,
    \item the cut polytope of $ K_n $, or
    \item the traveling salesman polytope of $ K_n $.
    \end{itemize}
    Then $k=\Omega(\sfrac{n}{\log n})$.
\end{theorem}
\end{samepage}
As we will see later, in all of the above cases our lower bound on the number of integer variables is tight up to a
factor that is poly-logarithmic in $ n $.
\begin{theorem}\label{thm:keyConsequences2}
    There is a constant $c>0$ such that the following holds.
    Let $n\in \mathbb{Z}_{\geq 1}$, and let $m,k >0$ such that there exists a MILEF with complexity $(m,k)$ with $m\leq 2^{c\cdot \sqrt{n}}$ for either:
    \begin{itemize}[itemsep=-0.4em,topsep=0.2em]
    \item the stable set polytope of any $ n $-vertex graph,
    \item the knapsack polytope of any $ n $-item instance,
    \item the independence polytope of any matroid on a ground set of cardinality $n$, or
    \item the convex hull of all feasible points of any conic bimodular integer program with $n$ variables.
    \end{itemize}
    Then $k=\Omega(\sfrac{\sqrt{n}}{\log n})$.
\end{theorem}
A summary of stated lower bounds together with upper bounds on the number of integer variables needed in
polynomial-size MILEFs can be found in Figure~\ref{figTableResults}.
Actually, our techniques imply even slightly stronger versions of Theorems~\ref{thm:keyConsequences1}
and~\ref{thm:keyConsequences2} that rule out MILEFs for the said problems that closely approximate $P$.
We will get back to this later.

\begin{figure}
    \begin{center}
        \begin{tabular}{lcc}
            \toprule
            Polytope(s) & Lower bound & Upper bound \\
            \midrule
            Matching         &        $ \Omega(\sfrac{n}{\log n}) $ & $ O(n) $ \\
            Odd-cut dominant &        $ \Omega(\sfrac{n}{\log n}) $ & $ O(n) $ \\
            Cut              &        $ \Omega(\sfrac{n}{\log n}) $ & $ O(n) $ \\
            TSP              &        $ \Omega(\sfrac{n}{\log n}) $ & $ O(n \log n) $ \\
            \midrule
            Stable set       & $ \Omega(\sfrac{\sqrt{n}}{\log n}) $ & $ O(n) $ \\
            Knapsack         & $ \Omega(\sfrac{\sqrt{n}}{\log n}) $ & $ O(n) $ \\
            Matroid          & $ \Omega(\sfrac{\sqrt{n}}{\log n}) $ &    $ - $ \\
            Bimodular        & $ \Omega(\sfrac{\sqrt{n}}{\log n}) $ & $ O(n) $ \\
            \bottomrule
        \end{tabular}
    \end{center}
    \caption{
        Bounds on the number of integer variables in MILEFs of a certain size.
        The first four rows refer to the respective polytopes on the complete undirected graph on $ n $ nodes.
        Here, the lower bounds hold for MILEFs up to size $ 2^{c n} $ for some constant $ c > 0 $ (see Theorem~\ref{thm:keyConsequences1}).
        The lower bounds in the latter four rows refer to stable set polytopes of $ n $-vertex graphs, knapsack
        polytopes of $ n $-item instances, independence polytopes of matroids with cardinality $ n $, and
        integer hulls of conic bimodular integer programs with $ n $ variables, respectively.
        These bounds hold for MILEFs up to size $ 2^{c \sqrt{n}} $, and these bounds are to be interpreted as guaranteeing the existence of polytopes in the respective family for which the lower bound holds (see Theorem~\ref{thm:keyConsequences2}).
        The right column contains upper bounds on the number of integer variables that are sufficient to obtain MILEFs of size polynomial in $ n
        $ for the respective family, and these bounds are valid for all members of the family (see Section~\ref{secUpperBounds}).
        For the case of independence polytopes of matroids, no polynomial-size MILEF is known.
    }
    \label{figTableResults}
\end{figure}

Previously, a non-trivial lower bound on the number of integer variables needed in any sub-exponential MILEF was only known for matchings, where a bound of $\Omega(\sqrt{\sfrac{n}{\log n}})$ was obtained in~\cite{HildebrandWZ2017}. Our results not only apply to a much broader class of problems, but are often also nearly tight in terms of revealing how many integer variables are needed, which is a key parameter in MILEFs.
More precisely, polynomial-size MILEFs for cuts and minimal odd-cuts are well known.
For matchings, the textbook integer formulation uses one integer variable per edge, thus leading to a compact MILEF with $O(n^2)$ integer variables.
As we show in Section~\ref{secApplicationMatching}, also matchings admit a compact MILEF with only $O(n)$ integer variables.
Hence, Theorem~\ref{thm:keyConsequences1} shows that even if we allow exponential-size MILEFs, satisfying $m\leq 2^{cn}$ for a well-chosen constant $c>0$, at least nearly linearly many, i.e., $\Omega(\sfrac{n}{\log n})$, integer variables are needed for describing the matching polytope, the cut polytope, or the dominant of the odd-cut polytope.
Whereas this may be natural to expect for hard problems, it is interesting that polytopes corresponding to efficiently solvable problems, like maximum matchings, minimum odd-cuts, or bimodular integer programming, cannot be described by a MILEF with much fewer than linearly many integer variables.
In particular, this rules out the possibility to efficiently solve the odd-cut problem or bimodular integer programming through a MILEF with a classical algorithm for mixed-integer programs like Lenstra's algorithm~\cite{Lenstra1983}, whose running time dependence on the number of integer variables $k$ is $k^{O(k)}$; we would thus need $k=O(\sfrac{\log n}{\log\log n})$ for Lenstra's algorithm to run efficiently.

Moreover, since the matching polytope of a complete graph with $n$ vertices is a linear projection of the traveling salesman polytope (TSP) on an $O(n)$-vertex graph (see~\cite{yannakakis_1991_expressing}), our results also extend to the TSP polytope. We note that TSP and its variants have been heavily studied in the context of MILEF formulations (see \cite{padberg_1991_analytical,gouveia_1995_classification,orman_2007_survey,oncan_2009_comparative} and references therein) and we provide the first nearly-linear lower bound on the number of integer variables needed in such formulations.

\medskip

Whereas Theorems~\ref{thm:keyConsequences1} and~\ref{thm:keyConsequences2} give a nice overview of the type of results we can obtain, the main contribution of our work, which leads to those results, is a general technique to transform hardness results for approximate LEFs to the mixed-integer setting, in a black-box fashion. 
We first need a formal notion of approximate LEFs and MILEFs, which boils down to defining, for two non-empty convex sets $A,B\subseteq \mathbb{R}^d$ with $A\subseteq B$ (think of $B$ as being a relaxation of $A$), how well $B$ approximates $A$.
Various notions have been used in the literature, depending on the context. In particular, from a viewpoint of optimization, it is natural to consider a notion related to the ratio of the optimal values when optimizing a linear objective over $A$ and $B$, respectively, like the integrality gap. We therefore use the notion of the \emph{maximization gap} $\LPgapMax(A,B)$ and the \emph{minimization gap} $\LPgapMin(A,B)$ between $A$ and $B$ with $A\subseteq B\subseteq \R^d_{\geq 0}$, which are defined as follows:
\ifbool{SODAfinal}{
\begin{align*}
\LPgapMax(A,B)= \inf \Big\{& \eps \geq 0: \ \forall c\in \R^d_{\ge 0}, \\
&(1+\eps)\cdot \sup_{a\in A} c\t a \geq \sup_{b\in B} c\t b \Big\},
\\
\LPgapMin(A,B)=\inf \Big\{& \eps \geq 0: \ \forall c\in \R^d_{\ge 0}, \\ 
& \inf_{a\in A} c\t a \leq  (1+\eps)\cdot \inf_{b\in B} c\t b \Big\}.
\end{align*}
}{
\begin{align*}
\LPgapMax(A,B)&=\inf\{\eps \geq 0: \ (1+\eps)\cdot \sup_{a\in A} c\t a \geq \sup_{b\in B} c\t b \quad \forall c\in \R^d_{\ge 0}\},
\\
\LPgapMin(A,B)&=\inf\{\eps \geq 0: \ \inf_{a\in A} c\t a \leq  (1+\eps)\cdot \inf_{b\in B} c\t b \quad \forall c\in \R^d_{\ge 0}\}.
\end{align*}
}
Clearly, the maximization gap is relevant for maximization problems like maximum matchings, or maximum stable set, and the minimization gap is used for minimization problems like minimum odd-cut.

Many approximation hardness results for LEFs are stated in terms of these linear programming (LP) gap notions. However, for our techniques, a more ``geometric'' notion is more convenient. In particular, one that is invariant under basic operations like bijective affine transformations, which is not the case for $\LPgapMax$ and $\LPgapMin$. We therefore introduce the notion of ``relative distance'' between $A$ and $B$, which can be interpreted as a normalized notion of LP gap and has many helpful properties. In particular, it is invariant with respect to affine bijections and can easily be related to the LP gap notions for $0/1$-polytopes.
For two non-empty convex sets $A\subseteq B \subseteq \mathbb{R}^d$, we define their \emph{relative distance} by
\begin{equation*}
\rdist(A,B) \coloneqq \sup_{\pi \colon \R^d \to \R} \ \frac{d_H (\pi(A), \pi(B))}{\diam(\pi(A))} \enspace,
\end{equation*}
where the supremum is taken over all linear maps $\pi \colon \R^d \to \R$, $d_H(\cdot,\cdot)$ is the Hausdorff distance (i.e. $d_H (\pi(A), \pi(B))= \sup_{b \in B} \inf_{a\in A} |\pi(b)-\pi(a)|$), 
and $\diam(\cdot)$ is the diameter function (i.e. $\diam(\pi(A))=\sup_{a,a'\in A} |\pi(a)-\pi(a')|$).
In the above definition, we interpret as $0$ any fraction with $\infty$ in the denominator as well as the fraction $\sfrac{0}{0}$.
One can easily observe that the definition of relative distance does not change if the supremum is only taken over orthogonal projections onto a line, which goes through the origin.
As an illustration of the notion of relative distance, see Figure~\ref{figRdist}.
\begin{figure}
    \begin{center}
              \begin{tikzpicture}[scale=0.3]

           \path (-3,11) -- (3,2) coordinate[pos=0.07] (a) {};

            \fill[fill=black!20] (-4,10) -- (12,11) -- (2,19) -- cycle;

            \fill[fill=black!70] (-3,11) -- (a) -- (7,11) -- (3,17) -- cycle;

            \draw[dotted] (-4,10) -- (1.846,1.231);
            \draw[dotted] (a) -- (3,2);
            \draw[dotted] (7,11) -- (9.923,6.615);
            \draw[dotted] (12,11) -- (13.385,8.923);

            \draw[dotted] (0,0) -- (15,10);
            \draw[black!20,line width=3] (1.846,1.231) -- (13.385,8.923);
            \draw[black!70,line width=3] (3,2) -- (9.923,6.615);

  \begin{scope}
          \tikzstyle{br}=[decorate,decoration={brace,mirror,raise=4pt,amplitude=4pt}, line width=0.5pt]
          \tikzstyle{lab}=[midway, xshift=8pt, yshift=-12pt]
           \draw[br] (9.923,6.615) -- (13.385,8.923) node[lab] {$\alpha$};
           \draw[br] (3,2) -- (9.923,6.615) node[lab] {$\beta$};
  \end{scope}

        \end{tikzpicture}
     \end{center}
    \caption{The relative distance of the above light shaded and dark shaded polygons is $ \frac{\alpha}{\beta} $.}
    \label{figRdist}
\end{figure}
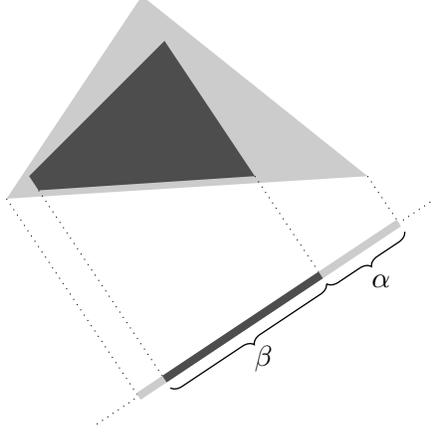
Moreover, we extend the definition to empty sets by setting $\rdist(\emptyset,\emptyset)=0$ and $\rdist(\emptyset,B)=\infty$ for $B\neq\emptyset$.
Even though one could define the relative distance in a broader context without assuming $A\subseteq B$, we restrict ourselves to the above setting since it is the one relevant for our derivations.
Using the notion of relative distance, we now define approximate LEFs and MILEFs in the natural way.

\begin{definition}
For a convex set $C\subseteq \mathbb{R}^d$ and $ \eps \ge 0 $, an $\eps$-LEF is a pair $(Q,\pi)$ where $Q$ is a polyhedron in some space $\R^\ell$, 
and $\pi:\R^\ell\rightarrow \R^d$ is an affine map such that $C\subseteq \pi(Q)$ and $\rdist(C, \pi(Q)) \leq \epsilon$. 

Analogously, an $\eps$-MILEF is a triple $(Q,\sigma,\pi)$ where $Q$ is a polyhedron in $\R^\ell$, 
and $\sigma:\R^\ell\rightarrow \R^k$ and $\pi:\R^\ell\rightarrow \R^d$ are affine maps such that 
$\bar{C}\coloneqq \conv(\pi(Q\cap \sigma^{-1}(\mathbb{Z}^k)))$ satisfies $C\subseteq \bar{C}$ and $\rdist(C,\bar{C})\leq \eps$.
\end{definition}
\noindent Note that classical LEFs and MILEFs are $0$-LEFs and $0$-MILEFs, respectively. 
We also remark that approximate LEFs and MILEFs are well-defined even for non-polyhedral convex sets $C$.

We are now ready to state our main reduction result, which shows that the existence of an approximate MILEF for a convex set $C$ implies the existence of an approximate LEF for $C$. 
Thus, this allows for lifting non-existence results for approximate LEFs to the mixed-integer setting.
\begin{theorem}
    \label{thm:fromMILEFToLEF}
    Let $ C \subseteq \R^d $ be a convex set admitting a $\eps$-MILEF of complexity $(m,k)$, where $ \eps \in [0,1] $.
    Then, for every $ \delta \in (0,1] $, $ C $ admits a $(\eps + \delta)$-LEF of size $ m \cdot \left( 1 +
    \sfrac{k}{\delta} \right)^{O(k)} $.
\end{theorem}

To derive from the above theorem the MILEF-hardness results stated in Theorem~\ref{thm:keyConsequences1}, we proceed as
follows.
First, we observe that existing LEF approximation hardness results imply that for some
$\delta=\Omega(\sfrac{1}{\poly(n)})$ and a constant $c>0$, there is no $\delta$-LEF of size at most $2^{c\cdot n}$
for the polyhedron $P$ we consider.
The result then follows by choosing $\eps=0$ in Theorem~\ref{thm:fromMILEFToLEF}, and observing that if $k < \sfrac{\eta
\cdot n}{\log n}$, for an appropriately chosen constant $\eta >0$, Theorem~\ref{thm:fromMILEFToLEF} implies the
existence of a $\delta$-approximate LEF of size strictly less than $2^{c\cdot n}$, thus leading to a contradiction.
In particular, this proof approach implies that for some $\bar{\eps}=\Omega(\sfrac{1}{\poly(n)})$, there does not even
exist a $\bar{\eps}$-MILEF of complexity $(m,k)$ with $m<2^{c\cdot n}$ and $k\leq \sfrac{\eta\cdot n}{\log n}$, where
$\eta>0$ is some constant.
Due to the relation between $\rdist$ and LP gap that we establish, this result can be rephrased in terms of non-existence of approximate MILEFs with respect to LP gap.
We expand on these connections and the precise statements resulting out of them in later sections.
 \subsection{Organization of the paper}
We start by summarizing some key properties of the relative distance which we exploit later, 
including its relation to LP distance.
This is done in Section~\ref{sec:rdistBasics}.
In Section~\ref{sec:outline}, we provide a thorough overview of our techniques that lead to
Theorem~\ref{thm:fromMILEFToLEF}, our main result to reduce approximation hardness results about LEFs to MILEFs.
A key ingredient of this proof is a new decomposition technique that, for any convex set $D$, allows for approximating
any mixed-integer description of $D$ by the intersection of $D$ with the union of a small number of affine subspaces.
Since this result may be of independent interest, we present it separately in Section~\ref{sec:covering}.
Section~\ref{secApplications} expands on the implications of our techniques to different polytopes, and provides in
particular a formal proof of (strengthened versions of) Theorems~\ref{thm:keyConsequences1}
and~\ref{thm:keyConsequences2}.
In Section~\ref{secUpperBounds}, we provide some MILEFs for the polytopes mentioned in
Theorems~\ref{thm:keyConsequences1} and~\ref{thm:keyConsequences2} to discuss the quality of the bounds derived in these
statements.
\ifbool{SODAfinal}{

\emph{A note on this version.} Due to space constraints, some proofs and constructions have been deferred to a future extended version of this paper. 
In particular, we have deferred the proofs of Lemmas \ref{lem:rd2}, \ref{lemma:rdistMaxGap}, \ref{lemma:rdistMinGap} and~\ref{lemHausdorffBoundedByRdist} on the properties of relative distance, 
and the explicit construction of a polynomial-sized MILEF for the TSP polytope (on the complete undirected graph on $n$ nodes) with only $O(n\log n)$ integer variables, 
whose existence provides an upper bound highlighted in Figure~\ref{figTableResults}.
}{
We close the main part of our paper with a general discussion on the bounds achieved with our techniques, in
Section~\ref{secTowards}.
Appendix~\ref{appendixRdist} contains deferred proofs of properties of the relative distance.
}
 \section{Relative distance: basic properties}
\label{sec:rdistBasics}
Due to the extensive use of the relative distance throughout this paper, we start by stating the key properties that are
used in the main part of the paper.
We remark that our notion of relative distance is closely related to the \emph{difference body metric} for convex bodies,
\footnote{A convex body is a convex set in $ \R^d $ that is bounded, closed and full-dimensional.}
introduced by Shephard~\cite{shephard1960inequalities}.
Namely, for two convex bodies $ A, B \subseteq \R^d $ with $ A \subseteq B $, the difference body metric $ \nu $ is
defined via $ \nu(A, B) \coloneqq \log \left(1 + 2 \rdist(A, B) \right) $.
Thus, most properties stated below are direct consequences of results in~\cite{shephard1960inequalities}. 
\ifbool{SODAfinal}{
A self-contained proof of Lemma~\ref{lem:rd2} (which is in order since we deal with arbitrary convex sets) is deferred to the extended version of this paper. 
}{
For the sake of completeness (and since we deal with arbitrary convex sets), we provide a self-contained proof of
Lemma~\ref{lem:rd2} in Appendix~\ref{appendixRdist}.
}

\begin{lemma}\label{lem:rd2}
    Consider three convex sets $A\subseteq B \subseteq C$ in $\R^d$.
    \begin{enumerate}[label=(\roman*),itemsep=0.0em,topsep=0.2em]
        \item\label{item:rdistSetDef}
        $\rdist(A, B) = \inf \left \{ \lambda \ge 0 : B\subseteq (1+\lambda)A -\lambda A \right \}$.
        \item\label{item:rdistProjSmaller}
        For any affine map $\pi:\R^d\rightarrow \R^m$, we have $\rd(\pi(A),\pi(B))\leq \rd(A,B)$, with equality in the
        case that $m=d$ and $\pi$ is invertible.
        \item\label{item:rdistTriangleIneq}
        \ifbool{SODAfinal}{
        \[\rdist(A,C)\leq \begin{array}{l}
                          \rdist(A,B)+\rdist(B,C)\\
                           + 2 \rdist(A,B) \rdist(B,C)
                         \end{array}.\]
 
        }{
        $\rdist(A,C)\leq \rdist(A,B)+\rdist(B,C) + 2 \rdist(A,B) \rdist(B,C)$.
        }        
        \item\label{item:lemrd2HullUnion}
        For convex sets $A_1,\cdots, A_t, B_1, \cdots, B_t \subseteq \R^d$ with $A_i\subseteq B_i$ for all $i \in [t]$,
        \ifbool{SODAfinal}{
        \begin{align*}
         \rd &\left(\conv(\cup_{i\in [t]} A_i), \conv(\cup_{i\in[t]} B_i) \right) \\
         &\leq \max_{i \in [t]} \ \rd(A_i, B_i)\enspace.
        \end{align*}
        }{
        \begin{equation*}
            \rd\left(\conv(\cup_{i\in [t]} A_i), \conv(\cup_{i\in[t]} B_i) \right) \leq \max_{i \in [t]} \ \rd(A_i, B_i)\enspace.
    \end{equation*}
	}
    \end{enumerate}
\end{lemma}

The next two lemmas highlight the relation between the relative distance and the LP gap notions, $\LPgapMax$ and
$\LPgapMin$.
This allows us to first translate LEF approximation hardness results, which are often stated in terms of LP gap, into a
gap in terms of relative distance.
The same lemmas allow for translating our hardness results, which are stated with respect to the relative distance, back
to the notion of LP gap.

We call a convex set $C\subseteq \mathbb{R}^d_{\geq 0}$ \emph{down-closed} if for every $x\in C$ and $y\in \mathbb{R}^d_{\geq 0}$ such that $y \leq x$, we have $y\in C$. 
Moreover, for a $0/1$-polytope $P\subseteq [0,1]^d$---i.e., a polytope all of whose vertices are within $\{0,1\}^d$---we 
say that $P$ is \emph{up-closed} if $x\in P$ and $y\in [0,1]^d$ with $x\leq y$ implies $y\in P$.

\begin{lemma}
    \label{lemma:rdistMaxGap}
 For two non-empty convex sets $A\subseteq B\subseteq \R^d_{\geq 0}$ where $A$ is down-closed, we have $\rd(A,B)=\LPgapMax(A,B)$.  
\end{lemma}

\begin{lemma}
    \label{lemma:rdistMinGap}
    Let $ A \subseteq [0,1]^d $ be an up-closed $ 0/1 $-polytope and $ B \subseteq [0,1]^d $ be a convex set with $A\subseteq B$ and $ d' := \dim(A) = \dim(B) $.
    If $ d' = 1 $, then $ A = B $.
    Otherwise, we have
    \begin{enumerate}[label=(\roman*)]
        \item \label{itemRdistMinGapLowerBoundRdist} $ \rd(A, B) \ge \frac{1}{d' - 1} \cdot \frac{\LPgapMin(A, B)}{1 + \LPgapMin(A, B)} $, and
        \item \label{itemRdistMinGapLowerBoundMinGap} $ \LPgapMin(A, B) \ge \frac{\rd(A, B)}{d' - 1 - \rd(A, B)} $.
    \end{enumerate}
\end{lemma}
\ifbool{SODAfinal}{
The proofs of Lemmas~\ref{lemma:rdistMaxGap} and~\ref{lemma:rdistMinGap} are deferred to the extended version of this paper.
}{
The proofs of Lemmas~\ref{lemma:rdistMaxGap} and~\ref{lemma:rdistMinGap} are postponed to Appendix~\ref{appendixRdist}.
}

 \section{Outline of our techniques}
\label{sec:outline}

In this section, we explain our approach for proving Theorem~\ref{thm:fromMILEFToLEF} and reduce it to a problem of approximating the mixed-integer hull of a convex set  
by the intersection of the set with few affine subspaces.

\ifbool{SODAfinal}{}{
\begin{figure}[h]
    \begin{center}
       \tdplotsetmaincoords{70}{30}
\begin{tikzpicture}[tdplot_main_coords,scale=0.4]
\small

    \node at (-13, -8, 3) {$ \R^k $};
    \draw (-13, 6, 8) -- (-13, 6, -3) -- (-13, -6, -3) -- (-13, -6, 8) -- cycle;

    \draw[opacity=0.2] (-13,  3,  5) -- (-13, 4,  5);
    \draw[opacity=0.2] (-13,  2,  4) -- (-13, 4,  4);
    \draw[opacity=0.2] (-13,  1,  3) -- (-13, 4,  3);
    \draw[opacity=0.2] (-13,  0,  2) -- (-13, 4,  2);
    \draw[opacity=0.2] (-13, -1,  1) -- (-13, 4,  1);
    \draw[opacity=0.2] (-13, -2,  0) -- (-13, 4,  0);
    \draw[opacity=0.2] (-13, -3, -1) -- (-13, 4, -1);

    \draw[very thick] (-13, 4, 6) -- (-13, 4, -2) -- (-13, -4, -2) -- cycle;

    \tdplotsetrotatedcoords{0}{90}{0}
    \tdplotdrawarc[tdplot_rotated_coords]{( 0,-3, -13)}{0.15}{0}{360}{}{};
    \tdplotdrawarc[fill,tdplot_rotated_coords]{( 0, 0, -13)}{0.15}{0}{360}{}{};
    \tdplotdrawarc[fill,tdplot_rotated_coords]{( 0, 3, -13)}{0.15}{0}{360}{}{};
    \tdplotdrawarc[tdplot_rotated_coords]{(-3,-3, -13)}{0.15}{0}{360}{}{};
    \tdplotdrawarc[tdplot_rotated_coords]{(-3, 0, -13)}{0.15}{0}{360}{}{};
    \tdplotdrawarc[fill,tdplot_rotated_coords]{(-3, 3, -13)}{0.15}{0}{360}{}{};
    \tdplotdrawarc[tdplot_rotated_coords]{(-6,-3, -13)}{0.15}{0}{360}{}{};
    \tdplotdrawarc[tdplot_rotated_coords]{(-6, 0, -13)}{0.15}{0}{360}{}{};
    \tdplotdrawarc[tdplot_rotated_coords]{(-6, 3, -13)}{0.15}{0}{360}{}{};

    \node at (-13, 2, 1) {$ \sigma(Q) $};

    \draw (-10, 6, -10) -- (-10, -6, -10) -- (10, -6, -10) -- (10, 6, -10) -- cycle;
    \node at (-5, -8, -10) {$ \R^d $};

    \draw[red,very thick,dashed] (-2, 0, -10) -- ( 2, 0, -10);
    \draw[red,very thick,dashed] (-5, 3, -10) -- ( 5, 3, -10);

    \fill[red,opacity=0.3] (-2, 0, -10) -- (-5, 3, -10) -- ( 5, 3, -10) -- ( 2, 0, -10);

    \draw[very thick] (-8, 4, -10) -- (0, -4, -10) -- ( 8, 4, -10) -- cycle;
    \node at (-5, -1, -10) {$ \pi(Q) $};
    \node at (5.5, -2.5, -10) {$ \color{red} \conv \left( \pi(Q \cap \sigma^{-1}(\Z^k)) \right) $};

    \node at (1, 4, 9) {$ \R^\ell $};

    \draw[red,very thick,dashed] (-2, 0, 0) -- ( 2, 0, 0);
    \draw[red,very thick,dashed] (-5, 3, 0) -- ( 5, 3, 0);
    \draw[red,very thick,dashed] (-2, 3, 3) -- ( 2, 3, 3);

    \fill[blue,opacity=0.2] ( 0, -4, -2) -- ( 8, 4, -2) -- ( 0, 4, 6);
    \fill[blue,opacity=0.1] (-8, 4, -2) -- ( 0, -4, -2) -- ( 0, 4, 6);

    \draw[opacity=0.1] (-1, 4,  5) -- ( 1, 4,  5);
    \draw[opacity=0.1] (-1, 4,  5) -- ( 1, 4,  5);
    \draw[opacity=0.1] (-2, 4,  4) -- ( 2, 4,  4);
    \draw[opacity=0.1] (-3, 4,  3) -- ( 3, 4,  3);
    \draw[opacity=0.1] (-4, 4,  2) -- ( 4, 4,  2);
    \draw[opacity=0.1] (-5, 4,  1) -- ( 5, 4,  1);
    \draw[opacity=0.1] (-6, 4,  0) -- ( 6, 4,  0);
    \draw[opacity=0.1] (-7, 4, -1) -- ( 7, 4, -1);
    \draw[opacity=0.2] (-1, 4,  5) -- ( 0,  3,  5) -- ( 1, 4,  5);
    \draw[opacity=0.2] (-1, 4,  5) -- ( 0,  3,  5) -- ( 1, 4,  5);
    \draw[opacity=0.2] (-2, 4,  4) -- ( 0,  2,  4) -- ( 2, 4,  4);
    \draw[opacity=0.2] (-3, 4,  3) -- ( 0,  1,  3) -- ( 3, 4,  3);
    \draw[opacity=0.2] (-4, 4,  2) -- ( 0,  0,  2) -- ( 4, 4,  2);
    \draw[opacity=0.2] (-5, 4,  1) -- ( 0, -1,  1) -- ( 5, 4,  1);
    \draw[opacity=0.2] (-6, 4,  0) -- ( 0, -2,  0) -- ( 6, 4,  0);
    \draw[opacity=0.2] (-7, 4, -1) -- ( 0, -3, -1) -- ( 7, 4, -1);

    \draw[opacity=0.5] (-8, 4, -2) -- ( 8, 4, -2);
    \draw[very thick] (-8, 4, -2) -- ( 0, -4, -2) -- ( 8, 4, -2) -- ( 0, 4, 6) -- cycle;
    \draw[very thick] ( 0, -4, -2) -- ( 0, 4, 6);
    \node at (5, 4, 2.5) {$ Q $};
    \node at (9, 4, 1) {$ \color{red} Q \cap \sigma^{-1}(\Z^k) $};

    \draw[dotted] (-13, 0, 0) -- (-2, 0, 0);
    \draw[dotted] (-13, 3, 0) -- (-5, 3, 0);
    \draw[dotted] (-13, 3, 3) -- (-2, 3, 3);    
\end{tikzpicture}
     \end{center}
    \caption{
    Representation of a MILEF $(Q,\sigma, \pi)$ of a convex set $C \subseteq \mathbb{R}^d$, i.e., $C = \conv(\pi(Q)\cap \sigma^{-1}(\mathbb{Z}^{k}))$.  The MILEF has $k$ integrality constraints. 
    In this picture, $Q$ lives in a $3$-dimensional space, i.e., $\ell=3$; the convex set $C$ lives in a $2$-dimensional space, i.e., $d=2$; and the number of integer constraints is $k=2$. 
    The projection $\sigma(Q)$ of $Q$ onto the integer space is highlighted on the left of the picture, and the projection $\pi(Q)$ onto the space of $C$ is highlighted at the bottom of the picture.}
    \label{figMilef}
\end{figure}
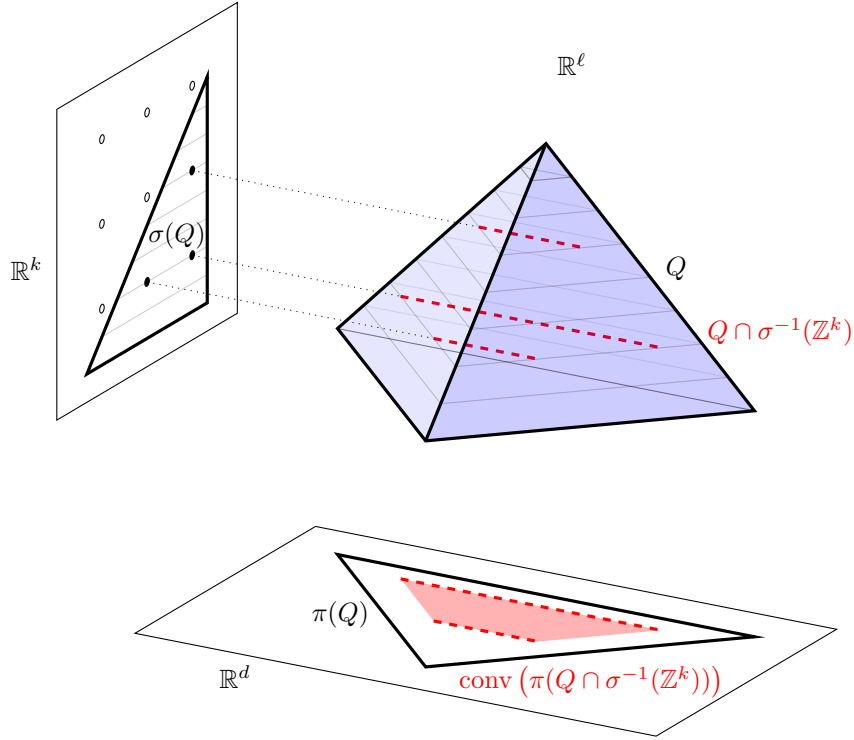
}

To exemplify our approach, consider a convex set $C\subseteq \mathbb{R}^d$, and a MILEF $(Q,\sigma,\pi)$ for $C$ of complexity $(m,k)$
\ifbool{SODAfinal}{.}{; see Figure~\ref{figMilef}.} 
This means that 
\ifbool{SODAfinal}{
\begin{align*}
C &= \conv\left(\pi(Q\cap \sigma^{-1}(\mathbb{Z}^k))\right) \\
  &= \pi\left(\conv\left( Q\cap \sigma^{-1}(\mathbb{Z}^k) \right)\right)
\enspace,
\end{align*}
}{
\begin{equation*}
C = \conv\left(\pi(Q\cap \sigma^{-1}(\mathbb{Z}^k))\right)
  = \pi\left(\conv\left(
 Q\cap \sigma^{-1}(\mathbb{Z}^k)
\right)\right)
\enspace,
\end{equation*}
}
where the second equality follows from the fact that the convex hull commutes with affine maps.
Assume that there are two constants $c,\eta > 0$ such that, for any $\delta < d^{-\eta}$, the set $C$ does not admit a $\delta$-LEF of size smaller than $2^{c\cdot d}$.
Our goal is to transform the MILEF into a $\delta$-LEF without blowing up the size too much.

The integer constraints of the MILEF cut the polyhedron $Q$ into \emph{fibers}, where a fiber is a set $Q\cap \sigma^{-1}(z)$ for some $z\in \mathbb{Z}^k$. 
Consider first a simple special case, where the number of non-empty fibers is not very large. To this end, let
\begin{equation*}
I = \{z\in \mathbb{Z}^k \mid Q\cap \sigma^{-1}(z)\neq\emptyset\}\enspace,
\end{equation*}
be all points in $\mathbb{Z}^k$ that correspond to non-empty fibers, and assume that we have $|I|<2^{c\cdot d/2}$. Notice that we can rewrite $C$ as
\begin{equation*}
C=\pi(Q_\sigma) \enspace,
\end{equation*}
where 
\ifbool{SODAfinal}{
\begin{align*}
Q_\sigma &\coloneqq \conv\Big(
 \bigcup_{z\in I} \left(
 Q\cap \sigma^{-1}(z)
 \right)
 \Big) \\
 &= \conv\left(Q\cap \sigma^{-1}(\mathbb{Z}^k)\right) \enspace.
\end{align*}
}{
\begin{equation*}
Q_\sigma \coloneqq \conv\Big(
 \bigcup_{z\in I} \left(
 Q\cap \sigma^{-1}(z)
 \right)
 \Big) 
 = \conv\left(Q\cap \sigma^{-1}(\mathbb{Z}^k)\right) \enspace.
\end{equation*}
}
We refer to $Q_\sigma$ as the mixed-integer hull of $Q$ \emph{with respect to $\sigma$}, and we remark that $Q_\sigma$ is a LEF of $C$, via the affine map $\pi$.
Moreover, we can bound the extension complexity of $Q_\sigma$ through a technique known as \emph{disjunctive programming}~\cite{Balas1979}, 
which allows for obtaining an inequality description of the convex hull of the union of a family of polyhedra, given an inequality description for each polyhedron in the family. 
In our case, the polyhedra are the fibers $Q\cap \sigma^{-1}(z)$ for $z\in I$, and each of those polyhedra has at most $m$ facets, because it is the intersection of $Q$, whose facets are bounded by $m$, 
and the affine subspace $\sigma^{-1}(z)$. 
The disjunctive programming technique then implies $\xc(Q_\sigma) \leq |I|(m+1)$, which, by assuming $|I|< 2^{c\cdot d/2}$, implies $\xc(Q_\sigma) = O(m \cdot 2^{c\cdot d/2})$.
Moreover, since we assumed that every $\delta$-LEF of $C$ has size at least $2^{c\cdot d}$, and $Q_\sigma$ is a $0$-LEF of $C$, we must have $\xc(Q_\sigma)= \Omega(2^{c\cdot d})$ 
and thus $m= \Omega(2^{c\cdot d/2})$.
In words, if the MILEF $(Q,\sigma,\pi)$ for $C$ has few fibers, then it must have a very large size.

Ideally, if one could show that any MILEF for $C$ of complexity $(m,k)$ with $k = O(\sfrac{d}{\log d})$ has a small number of fibers, then we would be done.
In particular, the number of fibers would be sufficiently small if all integer variables involved in the MILEF had bounded range.
Unfortunately, it does not hold in general that the number of fibers is bounded. A key aspect of our approach is to overcome this hurdle. 
More precisely, given a MILEF $(Q,\sigma,\pi)$ for $C$ as before, instead of describing $Q_\sigma$ in terms of the fibers $Q\cap \sigma^{-1}(z)$, 
we will show that one can \emph{approximate} $Q_\sigma$ by sets of the form $Q\cap H$, where $H$ comes from a family $\mathcal{H}$ of affine subspaces which is of small cardinality whenever $k$ is small. 
In particular, an affine subspace $H\in \mathcal{H}$ will not be of the form $\sigma^{-1}(z)$ as in the case of fibers, but will typically contain many subspaces of the form $\sigma^{-1}(z)$ for $z\in I$.
Moreover, the family $\mathcal{H}$ may contain subspaces of different dimensions. 
The price we pay is that the resulting description is not exact anymore, but only yields an approximation of $Q_\sigma$. 
Concretely, the resulting set will be 
\[
 Q_\mathcal{H} \coloneqq \conv\Big(\bigcup_{H\in \mathcal{H}} ( Q\cap H ) \Big) = \conv\Big( Q\cap \bigcup_{H\in\calH} H  \Big) \enspace,
\]
having the property that $Q_\sigma \subseteq Q_\mathcal{H} \subseteq Q$. 
We show that $Q_\mathcal{H}$ is a very good approximation of $Q_\sigma$, with an error of $\delta=O(d^{-\kappa})$ in terms of relative distance, 
where $\kappa>0$ is a constant that we can choose, and that only impacts other constants in our statements. 
Consequently, $\pi(Q_\mathcal{H})$ will be a good approximation of $C$ as well.

To find such family $\mathcal{H}$ of subspaces, we recursively ``slice'' $Q$ along different directions of small ``width''. 
We rely on a celebrated result in convex geometry to find good directions, which is commonly known as the Flatness Theorem, and which shows that low-dimensional lattice-free convex sets have small width. 
To formally state the Flatness Theorem, we start by defining the flatness constant and lattice width.
\footnote{In the remainder of the paper, we denote general convex sets by $B$, $C$ and $D$, and for the most part we keep the notational convention that $C$ is in the original space $\R^d$, 
$D$ is in the extended space $\R^\ell$, and $B$ is in the auxiliary space $\R^k$ (see Figure \ref{figMilef}).}

\begin{definition}[Flatness constant and lattice width]
Let $k\in \mathbb{Z}_{\geq 1}$. The \emph{flatness constant} $\flt(k)$ in dimension $k$ is the smallest $\lambda \in \mathbb{R}_{\geq 0}$ such that for any convex, closed, and full-dimensional set $ B \subseteq \R^k $ with $ B \cap \Z^k = \emptyset$, there exists a vector $ v \in \Z^k \setminus \{ \zerovec \} $ with
\begin{equation*}
\sup_{x \in B} v\t x - \inf_{x \in B} v\t x \le \lambda.
\end{equation*}
Moreover, the quantity
\begin{equation*}
\width(B) \coloneqq \inf_{v\in \mathbb{Z}^k\setminus \{\zerovec\}}
\left(
\sup_{x\in B} v\t x - \inf_{x\in B} v\t x
\right)
\end{equation*}
is called the \emph{lattice width of $B$}.
\end{definition}

Hence, $\flt(k)$ is the smallest real that upper bounds the lattice width of any convex, closed, full-dimensional and lattice-free set in $\mathbb{R}^k$. 
Notice that the term flatness \emph{constant} may be slightly misleading since $\flt(k)$ does depend on $k$. 
This term is historical, and comes from the fact that lattice width was often studied in settings where $k$ (and thus also $\flt(k)$) is constant.
Finally, the term Flatness Theorem is used for theorems that bound the quantity $\flt(k)$ in terms of $k$. 
There are many versions of it; we state one coming from \cite[Sec.~VII.8]{Barvinok2002} that is convenient for us in what follows.

\begin{theorem}[Flatness Theorem {(see~\cite{Barvinok2002})}]
\label{thm:flatness}
The flatness constant $\flt(k)$ is always finite and can moreover be bounded by a polynomial in $k$.
\end{theorem}

The following theorem is a key technical ingredient in our approach, and guarantees the existence of a good family $\mathcal{H}$ of subspaces.

\begin{theorem}
    \label{thm:Divide}
    Let $ \ell, k \in \Z_{> 0} $, $ D \subseteq \R^\ell $ be a convex set, and $ \sigma : \R^\ell \to \R^k $ be an affine map.
    Then for every $ \delta > 0 $ there exists a family $ \calH $ of at most $ \left(1+ \frac{1+\delta}{\delta} \flt(k) \right)^k $ affine subspaces of $ \R^\ell $ such that the sets 
    $D_\sigma = \conv \left( D\cap \sigma^{-1} (\Z^k) \right)$ and $D_\calH \coloneqq \conv \left( D \cap \bigcup_{H \in \calH} H \right) $ satisfy
    \begin{enumerate}[label=(\alph*),itemsep=0.0em]
        \item\label{item:thmDivIsRelax} $ D_\sigma \subseteq D_\calH $, and
        \item\label{item:thmDivRdist} $ \rdist(D_\sigma, D_\calH) \le \delta $.
    \end{enumerate}
\end{theorem}
Since the above statement is independent of the notion of (mixed-integer) extended formulations and might be of
independent interest, we discuss its proof in the next section. 

Let us demonstrate how Theorem~\ref{thm:Divide} indeed implies Theorem~\ref{thm:fromMILEFToLEF}.
To this end, we show that Theorem~\ref{thm:Divide} implies the following, slightly stronger version.
\begin{theorem}
    \label{thm:fromMILEFToLEFstrong}
    Let $ C \subseteq \R^d $ be a convex set that has an $\eps$-MILEF of complexity $(m,k)$, for some $\eps\geq 0$.
    Then for every $\delta > 0$, $C$ has an $(\eps + \delta + 2 \eps \delta)$-LEF of size at most $(m + 1) \left(1+ \frac{1+\delta}{\delta} \flt(k) \right)^k$.
\end{theorem}
We first argue that Theorem~\ref{thm:fromMILEFToLEFstrong} indeed implies Theorem~\ref{thm:fromMILEFToLEF}.
\ifbool{SODAfinal}{\begin{myproof}{Proof of Theorem~\ref{thm:fromMILEFToLEF}}}{\begin{proof}[Proof of Theorem~\ref{thm:fromMILEFToLEF}]}
    Let $ C \subseteq \R^d $ be a convex set admitting an $\eps$-MILEF of complexity $(m,k)$, where $ \eps \in [0,1] $,
    and let $ \delta \in (0,1] $.
    Applying Theorem~\ref{thm:fromMILEFToLEFstrong}, with $\delta$ replaced by $ \sfrac{\delta}{3} $, we obtain that $
    C $ has a $ \mu $-LEF of size at most $ s $, where $ \mu \coloneqq \eps + \frac{\delta}{3} + \frac{2}{3} \eps \delta $ and
    $ s \coloneqq (m + 1) \left(1+ \frac{3+\delta}{\delta} \flt(k) \right)^k$.
    Notice that $ \eps \le 1 $ implies $ \mu \le \eps + \delta $; hence it only remains to prove that $s=m\cdot(1+\sfrac{k}{\delta})^{O(k)}.$
    Since $\delta\leq 1$, $m\geq 1$, and there are constants $\beta>0$ and $\gamma\geq 1$ such that $\flt(k)\leq \beta k^\gamma$, we have
    \ifbool{SODAfinal}{
    \begin{align*}
     s & \le 2m \cdot \Big( 1+ \frac{4}{\delta} \beta k^\gamma \Big)^k \leq 2m\cdot\left[ (1+4\beta)(1+\sfrac{k^\gamma}{\delta})  \right]^k \\
     &\leq 2m\cdot\left[ (1+4\beta)(1+\sfrac{k}{\delta})^\gamma  \right]^k.
    \end{align*}
    }{
    \[
        s \le 2m \cdot \Big( 1+ \frac{4}{\delta} \beta k^\gamma \Big)^k \leq 2m\cdot\left[ (1+4\beta)(1+\sfrac{k^\gamma}{\delta})  \right]^k \leq 2m\cdot\left[ (1+4\beta)(1+\sfrac{k}{\delta})^\gamma  \right]^k.
    \]
    }
    And finally, since $(1+\sfrac{k}{\delta})\geq 2$, there must be a constant $c$ such that $1+4\beta \leq (1+\sfrac{k}{\delta})^c$. Thus,
  \[
   s\leq 2m\cdot \left(1+\sfrac{k}{\delta}\right)^{(\gamma+c)k}\leq m\cdot (1+\sfrac{k}{\delta})^{(\gamma+c)k+1}\enspace,
  \]
    which completes the proof of the claim.
\ifbool{SODAfinal}{\end{myproof}}{\end{proof}}
It remains to show that Theorem~\ref{thm:Divide} implies Theorem~\ref{thm:fromMILEFToLEFstrong}.

\ifbool{SODAfinal}{\begin{myproof}{Proof of Theorem~\ref{thm:fromMILEFToLEFstrong}}}{\begin{proof}[Proof of Theorem~\ref{thm:fromMILEFToLEFstrong}]}
    By the assumption, there exists a polyhedron $ Q \subseteq \R^\ell $ with at most $ m $ facets, and affine maps $ \pi:\R^\ell \to \R^d $ and $ \sigma : \R^\ell \to \R^k $
    such that $Q_\sigma$ satisfies $C\subseteq \pi(Q_\sigma)$ and $ \rdist(C, \pi(Q_\sigma)) \le \eps $. 
    Applying Theorem~\ref{thm:Divide}, we obtain a set $ \calH $ of affine subspaces of $ \R^\ell$ with $ |\calH| \le \left(1+ \frac{1+\delta}{\delta} \flt(k) \right)^k $, 
    such that $ Q_\sigma \subseteq Q_\calH $ and $\rdist(Q_\sigma, Q_\calH) \le \delta $.

    Now, let $\cl(Q_\calH) $ be the closure of $ Q_\calH $.
    By Balas' Theorem~\cite{Balas1979}, we have 
    \ifbool{SODAfinal}{
    \begin{align*}
     \xc(\cl(Q_\calH)) &\le \sum_{H \in \calH} (\xc(Q \cap H) + 1) \\
     &\le \sum_{H \in \calH} (\xc(Q) + 1) \le |\calH| (m + 1).
    \end{align*}
    }{
        \[
        \xc(\cl(Q_\calH)) \le \sum_{H \in \calH} (\xc(Q \cap H) + 1) \le \sum_{H \in \calH} (\xc(Q) + 1) \le |\calH| (m + 1).
    \]
    }
    Thus, there exists a polyhedron $ Q' \subseteq \R^{\ell'} $ with at most $ (m + 1)|\calH| $ facets and an affine map
    $ \tau : \R^{\ell'} \to \R^\ell $ with $ \tau(Q') = \cl(Q_\calH) $.
    Let us define $ \pi' \coloneqq \pi \circ \tau $. We show Theorem~\ref{thm:fromMILEFToLEFstrong} by proving that $(Q',\pi')$ is a $(\epsilon+\delta+2\epsilon\delta)$-LEF of $C$; the above already shows that its size is bounded by $(m+1)(1+\frac{1+\delta}{\delta}\flt(k))^k$, as desired.
Clearly, we have 
    \ifbool{SODAfinal}{
    \begin{align*}
     \pi'(Q') &= \pi(\tau(Q')) = \pi(\cl(Q_{\calH})) \\
     &\supseteq \pi(Q_\calH) \supseteq \pi(Q_\sigma) \supseteq C,
    \end{align*}
    }{
    \[
        \pi'(Q') = \pi(\tau(Q')) = \pi(\cl(Q_{\calH})) \supseteq \pi(Q_\calH) \supseteq \pi(Q_\sigma) \supseteq C,
    \]
    }
    and hence it only remains to show that $ \rdist(C, \pi'(Q')) \le \eps + \delta + 2 \eps \delta $.
    To this end, first observe that
    \begin{align*}
        \rdist(\pi(Q_\sigma), \pi'(Q'))
        & = \rdist(\pi(Q_\sigma), \pi(\cl(Q_\calH))) \\
        & \le \rdist(Q_\sigma, \cl(Q_\calH)) \\
        & = \rdist(Q_\sigma, Q_\calH) \\
        & \le \delta\enspace,
    \end{align*}
where the first inequality follows from Lemma~\ref{lem:rd2}~\ref{item:rdistProjSmaller}, 
    the second equality from the fact that $ \rdist(A, B) = \rdist(A, \cl(B)) $ always holds, 
    and the last inequality from the definition of $ Q_\calH$.
    Finally, recall that $ \rdist(C, \pi(Q_\sigma)) \le \eps $, and hence using Lemma~\ref{lem:rd2}~\ref{item:rdistTriangleIneq} we obtain
    \ifbool{SODAfinal}{
        \begin{align*}
        \rdist &(C, \pi'(Q')) \\
        \le & \rdist(C,\pi(Q_\sigma)) +\rdist(\pi(Q_\sigma),\pi'(Q')) \\
        &+ 2 \rdist(C,\pi(Q_\sigma)) \rdist(\pi(Q_\sigma),\pi'(Q')) \\
        \le & \eps + \delta + 2 \eps \delta,     
    \end{align*}
    }{
        \begin{align*}
        \rdist(C, \pi'(Q')) &\le \rdist(C,\pi(Q_\sigma))+\rdist(\pi(Q_\sigma),\pi'(Q')) + 2 \rdist(C,\pi(Q_\sigma)) \rdist(\pi(Q_\sigma),\pi'(Q')) \\
        &\le \eps + \delta + 2 \eps \delta,     
    \end{align*}
    }
    as claimed.
\ifbool{SODAfinal}{\end{myproof}}{\end{proof}}
 \section{Approximating mixed-integer hulls by unions of slices}
\label{sec:covering}
As demonstrated in the previous section, the key technical ingredient for the proof of Theorem~\ref{thm:fromMILEFToLEF}
is the statement of Theorem~\ref{thm:Divide}.
In this section we prove the latter. 
Throughout this section, we make extensive use of the sets $D_\sigma$ and $D_\calH$ as defined in Theorem~\ref{thm:Divide}.
Before we give an idea of its proof, let us argue that it can easily be derived from the following statement.
\begin{proposition}
    \label{propDivide}
    Let $ \ell, k \in \Z_{> 0} $, $ D \subseteq \R^\ell $ be a convex set, and $ \sigma : \R^\ell \to \R^k $
    be an affine map.
    Then, for every $ \delta > 0 $ there exists a family $ \calH $ of affine subspaces of $ \R^\ell $ satisfying:
    \begin{enumerate}[itemsep=0.0em,label=(\roman*)]
        \item\label{item:propDivFamSize} $ |\calH| \le \prod_{i = 1}^k \left( 1+ \frac{(1 + \delta)}{\delta} \flt(i) \right) $,

        \item\label{item:propDivIntConstr} $ D \cap \sigma^{-1}(\Z^k) \subseteq  \bigcup_{H \in \calH} H $, and

        \item\label{item:propDivRdist} $ \rdist\left(  (D\cap H)_\sigma, D\cap H  \right) \le \delta $ for each $ H \in \calH $.
    \end{enumerate}
\end{proposition}
\ifbool{SODAfinal}{\begin{myproof}{Proof of Theorem~\ref{thm:Divide}}}{\begin{proof}[Proof of Theorem~\ref{thm:Divide}]}
    Let $\mathcal{H}$ be the family of affine subspaces as described in Proposition~\ref{propDivide}. 
    By Proposition~\ref{propDivide}~\ref{item:propDivFamSize} and the monotonicity of $\flt(k)$, we immediately obtain the bound $|\calH|\leq \left(1+\frac{1+\delta}{\delta}\flt(k)\right)^k$. 
    Next, property~\ref{item:thmDivIsRelax} of Theorem~\ref{thm:Divide} follows from Proposition~\ref{propDivide}~\ref{item:propDivIntConstr} because
    \ifbool{SODAfinal}{
    \begin{align*}
    D_\calH &= \conv\left( D\cap \bigcup_{H\in \mathcal{H}} H\right)\\
           &\supseteq \conv\left( D\cap \sigma^{-1}(\Z^k) \cap \bigcup_{H\in \mathcal{H}} H \right)\\
           &= \conv\left(D\cap \sigma^{-1}(\mathbb{Z}^k)\right) = D_\sigma \enspace.
    \end{align*}
    }{
    \begin{equation*}
    D_\calH = \conv\left( D\cap \bigcup_{H\in \mathcal{H}} H\right)
           \supseteq \conv\left( D\cap \sigma^{-1}(\Z^k) \cap \bigcup_{H\in \mathcal{H}} H \right)
           = \conv\left(D\cap \sigma^{-1}(\mathbb{Z}^k)\right) = D_\sigma \enspace.
    \end{equation*}
    }
Hence, it only remains to show property~\ref{item:thmDivRdist} of Theorem~\ref{thm:Divide}. To this end, for each $ H\in \calH$, we define 
$A_H \coloneqq (D\cap H)_\sigma$ and $B_H \coloneqq D\cap H\enspace$,
which leads to
\ifbool{SODAfinal}{
\begin{align*}
A_H &\subseteq B_H\enspace, \\
D_\sigma & = \conv\left(\bigcup_{H\in \calH} A_H \right)\enspace, \quad \text{ and }\\
D_\calH &= \conv\left(\bigcup_{H\in\calH} B_H \right)\enspace,
\end{align*}
}{
\begin{align*}
A_H\subseteq B_H\enspace, \qquad D_\sigma = \conv\left(\bigcup_{H\in \calH} A_H \right)\enspace,\qquad \text{ and } \qquad D_\calH = \conv\left(\bigcup_{H\in\calH} B_H \right)\enspace,
\end{align*}
}
where the second relation follows from Proposition~\ref{propDivide}~\ref{item:propDivIntConstr}. We finally obtain the bound
    \begin{align*}
        \rdist( D_\sigma, D_\calH) \leq \max_{H \in \calH} \rdist(A_H, B_H) \le \delta \enspace,
    \end{align*}
where the first inequality follows by Lemma~\ref{lem:rd2}~\ref{item:lemrd2HullUnion} and the second one by Proposition~\ref{propDivide}~\ref{item:propDivRdist}.
\ifbool{SODAfinal}{\end{myproof}}{\end{proof}}

Proposition~\ref{propDivide} states that there is a small family of affine subspaces of $\R^\ell$ that cover all fibers $D\cap\sigma^{-1}(\Z^k)$ of the set $D$, 
with the property that each slice $D\cap H$ approximates well its own mixed-integer hull $(D\cap H)_\sigma$. 
The idea behind the proof can be sketched as follows. 
If the set $ D $ is already a good approximation of its mixed-integer hull $D_\sigma$, 
then there is no need to intersect it with proper affine subspaces, i.e., we can simply choose $ \calH = \{\R^\ell\} $.
Otherwise, the next statement claims that there is a small family of affine subspaces covering all the fibers, 
such that the mixed-integer hull $(D\cap H)_\sigma$ of each slice can be described with fewer integer constraints; and we can recurse. 

Intuitively, the idea of how we exploit the Flatness Theorem is the following. 
If $D$ is not a good approximation of $D_\sigma$, then there is a point in $D\setminus D_\sigma$ certifying that $\rdist(D_\sigma,D)$ is large.
However, for this to be possible, the set of all fibers $\sigma^{-1}(\mathbb{Z}^k)$ cannot be extremely dense with respect to every direction. 
We exploit this through the Flatness Theorem to find a good direction with respect to which we can slice $D$ into polynomially (in $k$) many slices.

As mentioned before, Proposition~\ref{propDivide} is obtained through recursive slicing. 
The following lemma shows that we can find a family of affine subspaces for slicing $D$ once, thus reducing the number of integer constraints $k$ by one. 
With the recursive use of this lemma to eliminate all integer variables, the proof of Proposition~\ref{propDivide} becomes straightforward and is postponed to the end of this section.

\begin{lemma}
    \label{lemElimination}
    Let $ \ell, k \in \Z_{> 0} $, $ D \subseteq \R^\ell $ be a convex set, and $ \sigma : \R^\ell \to \R^k $ be an affine map.
    If
    \[
        \rdist(D_\sigma, D) = \delta > 0\enspace,
    \]
    then there exists a set $ \calH $ of affine subspaces of $ \R^\ell $ and an affine map $ \tau : \R^\ell \to \R^{k - 1} $
    satisfying:
    \begin{enumerate}[label=(\roman*),itemsep=0em]
        \item\label{item:lemElimSmallSize} $ |\calH| \le 1+ \frac{1 + \delta}{\delta}\flt(k)$,
        \item\label{item:lemElimCover} $ D \cap \sigma^{-1}(\Z^k) \subseteq  \bigcup_{H \in \calH} H $, and
        \item\label{item:lemElimKMinusOne} $ H \cap \sigma^{-1}(\Z^k) = H \cap \tau^{-1}(\Z^{k - 1}) $ for each $ H \in \calH $.
    \end{enumerate}
\end{lemma}
The affine subspaces (and the map $ \tau $) needed for Lemma~\ref{lemElimination} are implicitly given by the next
lemma, which shows that all fibers in $ D \cap \sigma^{-1}(\Z^k) $ can be covered by a small number of parallel lattice
hyperplanes.
\begin{lemma}
    \label{lemLatticeWidth}
    Let $ \ell, k \in \Z_{> 0} $, $ D \subseteq \R^\ell $ be a convex set, and $ \sigma : \R^d \to \R^k $ be an affine map.
    If
        $\rdist(D_\sigma, D) = \delta > 0$,  then
$\width\left(\sigma(D_\sigma) \right) \leq \frac{1 + \delta}{\delta} \flt(k)$.
\end{lemma}
\begin{proof}
    We assume that the set $\sigma(D_\sigma)$ is full-dimensional in $\R^k$, for otherwise its lattice width is $ 0 $ and the statement holds trivially.
    Notice that it is enough to show that, for any value $\lambda>0$ with $\lambda < \rd(D_\sigma, D)$, 
    the lattice width of $\sigma(D_\sigma)$ is bounded by $\frac{1+\lambda}{\lambda} \flt(k)$.
    
    For such a value $0<\lambda<\delta$, it follows from Lemma~\ref{lem:rd2}~\ref{item:rdistSetDef} that there exists a point $ y \in D $ 
    such that $ y \notin (1 + \lambda) D_\sigma - \lambda D_\sigma $. 
    In turn, this means that the sets $ y + \lambda D_\sigma $ and $ (1 + \lambda) D_\sigma $ are disjoint.
    Scaling both sets by $ \frac{1}{1 + \lambda} $ and setting $ \mu \coloneqq \frac{\lambda}{1 + \lambda} \in (0,1) $, we
    obtain
    \begin{equation}
        \label{eqProoflemLatticeWidth1}
        \left( (1 - \mu) y + \mu D_\sigma \right) \cap D_\sigma = \emptyset.
    \end{equation}
    Now consider the (convex, closed, and full-dimensional) set
    \begin{equation}\label{eq:defOfB}
    B \coloneqq (1 - \mu)\sigma(y) +\mu \sigma(D_\sigma) \subseteq \R^k\enspace,
    \end{equation}
    and notice that $\sigma(D_\sigma)$ can be obtained from $B$ via a scaling with factor $\sfrac{1}{\mu}$ followed by a translation. Hence,
    \begin{equation*}
        \width(D_\sigma) = \frac{1}{\mu} \width(B) = \frac{1 + \lambda}{\lambda} \width(B)\enspace.
    \end{equation*}
    Thus, to prove the lemma, it suffices to show that $B$ is lattice free, and consequently $\width(B) \leq \flt(k)$.
    
    Assume, for the sake of deriving a contradiction, that there is a point $z\in B\cap \Z^k$. 
    As $z$ is in $B$, and by the definition of $B$ given in~\eqref{eq:defOfB}, there exists some $ x \in D_\sigma $ such that $ z = (1 - \mu) \sigma(y) + \mu \sigma(x)$. 
    Now, consider the point 
    \[
     w\coloneqq (1-\mu)y+\mu x \enspace.
    \]
    It is clear from its definition that $w\in (1-\mu)y+\mu D_\sigma$. 
    Moreover, $w\in D$, because it is a convex combination of points $x$ and $y$ in $D$. 
    And finally, $w\in \sigma^{-1}(\Z^k)$, because $\sigma(w)=(1-\mu)\sigma(y)+\mu\sigma(x) = z \in\Z^k$. 
    Therefore, we obtain $w\in \left( (1 - \mu) y + \mu D_\sigma \right) \cap D_\sigma $, a contradiction to~\eqref{eqProoflemLatticeWidth1}.
\end{proof}

We are now ready to prove Lemma~\ref{lemElimination}.
\ifbool{SODAfinal}{\begin{myproof}{Proof of Lemma~\ref{lemElimination}}}{\begin{proof}[Proof of Lemma~\ref{lemElimination}]}
    By the hypothesis and by Lemma~\ref{lemLatticeWidth}, we have $\width(\sigma(D_\sigma))\leq \frac{1 + \delta }{\delta} \flt(k)$.
    Hence, there exists a vector $ v \in \Z^k \setminus \{\zerovec\} $ such that the set
    \ifbool{SODAfinal}{
    \begin{align*}
     I &\coloneqq \{ v\t z : z \in \sigma(D_\sigma) \cap \Z^k \} \\
       &= \{ v\t z : z \in \sigma(D) \cap \Z^k \} \subseteq \Z
    \end{align*}
    }{
    \begin{equation*}
    I \coloneqq \{ v\t z : z \in \sigma(D_\sigma) \cap \Z^k \} = \{ v\t z : z \in \sigma(D) \cap \Z^k \} \subseteq \Z
    \end{equation*}
    }
    has cardinality $|I|\leq 1+ \frac{1 + \delta }{\delta} \flt(k) $.
    We may assume $ \gcd(v) = 1 $, for otherwise replacing $ v $ by $ \sfrac{v}{\gcd(v)}$ can only decrease the cardinality of $I$.
  By the definition of $I$ we have
\begin{equation}\label{eq:coveredBySlices}
  \sigma(D) \cap \mathbb{Z}^k \subseteq \bigcup_{i\in I} \{x\in \mathbb{R}^k : v\t x = i\}\enspace,
\end{equation}
i.e., the set $\sigma(D)\cap \mathbb{Z}^k$ can be covered by only $|I|$ many hyperplanes. 
We now take the pre-images of these hyperplanes under $\sigma$ to define our family $\mathcal{H}$. Hence, let $\mathcal{H} \coloneqq \{H_i : i \in I\}$, where
\ifbool{SODAfinal}{
\begin{align*}
 H_i \coloneqq & \{y\in \mathbb{R}^\ell : v\t \sigma(y)=i\} \\
     =& \sigma^{-1}\left(\{x\in \mathbb{R}^k : v\t x = i\}\right) \quad \forall i\in I.
\end{align*}

}{
\[
 H_i \coloneqq \{y\in \mathbb{R}^\ell : v\t \sigma(y)=i\}
    = \sigma^{-1}\left(\{x\in \mathbb{R}^k : v\t x = i\}\right) \quad \forall i\in I.
\]
}
Clearly, $ \calH $ satisfies property~\ref{item:lemElimSmallSize} of Lemma~\ref{lemElimination}. 
Moreover, property~\ref{item:lemElimCover} follows immediately from~\eqref{eq:coveredBySlices} and the fact that 
the $H_i$'s are the pre-images of the hyperplanes in~\eqref{eq:coveredBySlices}.
It remains to show that $\mathcal{H}$ fulfills property~\ref{item:lemElimKMinusOne} of Lemma~\ref{lemElimination}.
Since $ \gcd(v) = 1 $, it is well known that there exists a unimodular matrix $ U \in \Z^{k \times k} $, i.e., $\det(U)\in \{-1,1\}$, with the first row being $ v\t $.\footnote{
The existence of such a unimodular matrix $U$ with $v\t$ as its first row easily follows from the fact that the Hermite Normal Form of any vector with $\gcd=1$ is $e_1=(1,0,\ldots, 0)$. 
Since the (column) Hermite Normal Form can be obtained by integer column operations, these operations can be described by a unimodular matrix $A\in \mathbb{Z}^{k\times k}$. Hence, there is a unimodular matrix $A\in \mathbb{Z}^{k\times k}$ such that $v\t A = e_1\t$, and one can choose $U=A^{-1}$.
}
We have 
\begin{equation}\label{eq:unimodIntegral}
\left(z\in \mathbb{Z}^k \iff Uz\in \mathbb{Z}^k\right)
\qquad \forall z\in\mathbb{R}^k\enspace,
\end{equation}
because $U^{-1}$ is integral.
Let $ U' \in \mathbb{Z}^{(k - 1) \times k} $ be the matrix that arises from $ U $ by removing the first row.
Clearly, by defining $ \phi : \R^k \to \R^{k - 1} $ via $ \phi(x) \coloneqq U' x $, we can rephrase~\eqref{eq:unimodIntegral} as follows:
\ifbool{SODAfinal}{
\begin{equation}\label{eq:unimodIntegralRephrased}
        \left( z \in \Z^k \iff 
        \begin{array}{c} 
        v\t z \in \Z \, \text{ and } \\
        \phi(z) \in \Z^{k - 1}
        \end{array}
        \right)
\quad \forall z\in \mathbb{R}^k\enspace.
\end{equation}
}{
\begin{equation}\label{eq:unimodIntegralRephrased}
        \left( z \in \Z^k \iff v\t z \in \Z \, \text{ and } \, \phi(z) \in \Z^{k - 1}\right)
\qquad \forall z\in \mathbb{R}^k\enspace.
\end{equation}
}
    Since $ I \subseteq \Z $, we thus obtain for all $ i \in I $
    \ifbool{SODAfinal}{
    \begin{align*}
        H_i &\cap \sigma^{-1}(\Z^k) \\
        & = \{ y \in \R^\ell : v\t \sigma(y) = i, \, \sigma(y) \in \Z^k \} \\
        & = \{ y \in \R^\ell : v\t \sigma(y) = i, \, \phi(\sigma(y)) \in \Z^{k-1} \} \\
        & = H_i \cap (\phi \circ \sigma)^{-1}(\Z^{k-1}),
    \end{align*}
    }{
    \begin{align*}
        H_i \cap \sigma^{-1}(\Z^k)
        &= \{ y \in \R^\ell : v\t \sigma(y) = i, \, \sigma(y) \in \Z^k \}
        = \{ y \in \R^\ell : v\t \sigma(y) = i, \, \phi(\sigma(y)) \in \Z^{k-1} \} \\
        & = H_i \cap (\phi \circ \sigma)^{-1}(\Z^{k-1}),
    \end{align*}
    }
where the second equality follows from~\eqref{eq:unimodIntegralRephrased}. 
Hence, by setting $\tau \coloneqq \phi \circ \sigma$, we have that $\mathcal{H}$ fulfills property~\ref{item:lemElimKMinusOne} of Lemma~\ref{lemElimination}, as desired.
\ifbool{SODAfinal}{\end{myproof}}{\end{proof}}

Finally, we provide the proof of Proposition~\ref{propDivide}.
\ifbool{SODAfinal}{\begin{myproof}{Proof of Proposition~\ref{propDivide}}}{\begin{proof}[Proof of Proposition~\ref{propDivide}]}
    We proceed by induction over $ k \ge 0 $ and note that the claim is trivial for $ k = 0 $ by choosing $ \calH = \{\R^\ell \} $.
    Now let $ k \ge 1 $ and observe that we may assume that $ \rdist(D_\sigma, D ) > \delta $, as otherwise we can again choose $ \calH = \{ \R^d \} $.
    By Lemma~\ref{lemElimination} there exists a family $ \calL $ of affine subspaces of $ \R^\ell $ and an affine map $ \tau: \R^\ell \to \R^{k - 1} $ such that
    \begin{align}
        \label{eqProofPropDivide1}
        |\calL| & \le 1+ \frac{1 + \delta }{\delta} \flt(k) \enspace, \\
        \label{eqProofPropDivide2}
        D \cap \sigma^{-1}(\Z^k) & \subseteq \bigcup_{L \in \calL} L \enspace, \text{ and} \\
        \label{eqProofPropDivide3}
        L \cap \sigma^{-1}(\Z^k) & = L \cap \tau^{-1}(\Z^{k - 1}) \quad \forall L \in \calL \enspace.
    \end{align}
    For each such $ L \in \calL $, by the induction hypothesis applied to $ D \cap L $ and $ \tau $, there exists a family $\calH_L $ of affine subspaces in $ \R^\ell $ such that
    \ifbool{SODAfinal}{
    \begin{align}
        \label{eqProofPropDivide4}
        & |\calH_L| \le \prod_{i = 1}^{k-1} \left( 1+ \frac{1 + \delta }{\delta} \flt(i) \right), \\
        \label{eqProofPropDivide5}
        & D \cap L \cap \tau^{-1}(\Z^{k-1}) \subseteq \bigcup_{H \in \calH_L} H \enspace, \text{ and} \\
        \label{eqProofPropDivide6}
        & \rdist ((D\cap L\cap H)_\tau , D\cap L \cap H) \le \delta,
    \end{align}
    where \eqref{eqProofPropDivide6} holds for every $H \in \calH_L.$ 
    }{
    \begin{align}
        \label{eqProofPropDivide4}
        |\calH_L| & \le \prod_{i = 1}^{k-1} \left( 1+ \frac{1 + \delta }{\delta} \flt(i) \right), \\
        \label{eqProofPropDivide5}
        D \cap L \cap \tau^{-1}(\Z^{k-1}) & \subseteq \bigcup_{H \in \calH_L} H \enspace, \text{ and} \\
        \label{eqProofPropDivide6}
        \rdist ((D\cap L\cap H)_\tau , D\cap L \cap H) & \le \delta \quad \forall H \in \calH_L.
    \end{align}
    }
    Defining the set $ \calH \coloneqq \{ L \cap H : L \in \calL, \, H \in \calH_L \} $, we clearly satisfy~\ref{item:propDivFamSize} due to~\eqref{eqProofPropDivide1} and~\eqref{eqProofPropDivide4}.
    Furthermore, we have
    \ifbool{SODAfinal}{
    \begin{align*}
        & D \cap \sigma^{-1}(\Z^k) \\
        & = \bigcup_{L \in \calL} \left(D \cap L \cap \sigma^{-1}(\Z^k)\right) & \text{ (by \eqref{eqProofPropDivide2})} \\
        & = \bigcup_{L \in \calL} \left(D \cap L \cap \tau^{-1}(\Z^{k - 1})\right) & \text{ (by \eqref{eqProofPropDivide3})} \\
        & \subseteq \bigcup_{L \in \calL} \bigcup_{H \in \calH_L} L\cap H\enspace,   & \text{ (by \eqref{eqProofPropDivide5})}
    \end{align*}
    }{
    \begin{align*}
        D \cap \sigma^{-1}(\Z^k)
        & = \bigcup_{L \in \calL} \left(D \cap L \cap \sigma^{-1}(\Z^k)\right) & \text{ (by \eqref{eqProofPropDivide2})} \\
        & = \bigcup_{L \in \calL} \left(D \cap L \cap \tau^{-1}(\Z^{k - 1})\right) & \text{ (by \eqref{eqProofPropDivide3})} \\
        & \subseteq \bigcup_{L \in \calL} \bigcup_{H \in \calH_L} L\cap H\enspace,   & \text{ (by \eqref{eqProofPropDivide5})}
    \end{align*}
    }
    which shows~\ref{item:propDivIntConstr}.
    Finally,~\ref{item:propDivRdist} is a direct consequence of \eqref{eqProofPropDivide6}, 
    and the fact $(D\cap L\cap H)_\sigma = (D\cap L\cap H)_\tau$ for each $L\in\calL$ and $H\in\calH_L$, by \eqref{eqProofPropDivide3}.
\ifbool{SODAfinal}{\end{myproof}}{\end{proof}}
 \section{Applications}
\label{secApplications}
In this section, we demonstrate how our framework can be applied to obtain strong lower bounds on the number of integer
variables in MILEFs in several relevant settings.
Among other results, we will obtain the statements mentioned in Theorems~\ref{thm:keyConsequences1}
and~\ref{thm:keyConsequences2} using Theorem~\ref{thm:fromMILEFToLEF} and existing inapproximability results on LEFs.
In fact, we prove stronger versions of all these statements, as we also rule out the existence of \emph{approximate}
MILEFs.
To this end, we first derive the following direct consequence of Theorem~\ref{thm:fromMILEFToLEF}, which is suited for
the applications we consider.
\begin{corollary}
    \label{corollary:milefHardness}
    Let $ \alpha, \beta, \gamma, \eps > 0 $ be constants with $ \eps < \alpha, \gamma $.
    Let $ C $ be any non-empty convex set that does not admit an $ \frac{\alpha}{n^\beta} $-LEF of size at most $
    2^{\gamma n} $.
    Then any $ \frac{\alpha - \varepsilon}{n^\beta} $-MILEF of $ C $ of size at most $ 2^{(\gamma - \varepsilon) n} $
    has $ \Omega \left( \sfrac{n}{\log n} \right) $ integer variables.
\end{corollary}
\begin{proof}
    Let $ C $ be any convex set as in the hypothesis and suppose it admits a $ \frac{\alpha - \eps}{n^\beta} $-MILEF of
    complexity $ (m, k) $ where $ m \le 2^{(\gamma - \eps) n} $.
    We may assume that $ n \ge 2 $ and $ k \le n $.
    By Theorem~\ref{thm:fromMILEFToLEF}, $ C $ admits a $ \varrho $-LEF of size $ s $, where
    \[
        \varrho \coloneqq \frac{\alpha - \eps}{n^\beta} + \frac{\eps}{n^\beta} = \frac{\alpha}{n^\beta}
    \]
    and
    \[
        s \coloneqq m \left(1 + \frac{k}{\sfrac{\epsilon}{n^\beta}} \right)^{c k},
    \]
    for some constant $ c > 0 $.
    By the assumption, we must have
    \ifbool{SODAfinal}{
    \begin{align*}
        2^{\gamma n}
        & \le s
        \le m \left( k n^{\beta'} \right)^{c k} \\
        & \le m n^{(\beta' + 1) ck}
        \le 2^{(\gamma - \eps)n} n^{(\beta' + 1) ck}
    \end{align*}
    }{
    \begin{align*}
        2^{\gamma n}
        \le s
        \le m \left( k n^{\beta'} \right)^{c k}
        \le m n^{(\beta' + 1) ck}
        \le 2^{(\gamma - \eps)n} n^{(\beta' + 1) ck}
    \end{align*}
    }
    for some constant $ \beta' > 0 $, where the second inequality follows from $ n \ge 2 $ and by choosing $ \beta' $
    sufficiently large, the third inequality follows from $ k \le n $, and the last inequality is due to $ m \le
    2^{(\gamma - \eps)n} $.
    This implies $ k \ge \frac{\eps n}{c (\beta' + 1) \log n} $, which yields the claim.
\end{proof}
Note that the above statement allows for quickly translating an inapproximability result on LEFs into a certain
inapproximability result on MILEFs.
Besides the proofs of Theorems~\ref{thm:keyConsequences1} and~\ref{thm:keyConsequences2}, the main purpose of this
section is to demonstrate how several existing inapproximability results on LEFs in the literature, which are usually
stated using different notions of approximations, can be transferred into inapproximability results on LEFs as required
in the statement of Corollary~\ref{corollary:milefHardness}.
 \subsection{Matching polytope}
\label{secApplicationMatching}
We start by applying our framework to the matching polytope of the complete graph, to which we simply refer to as the \emph{matching polytope}, 
and which is defined as the convex hull of the characteristic vectors of all matchings in $ K_n = (V, E) $ (the complete graph on $ n $ nodes). 
We denote this polytope by $\match(n) \subseteq \R^E$.
Recall that a \emph{matching} is an edge subset $ M \subseteq E $ such that every vertex in $ (V, M) $ has degree at
most one.
A well-known result by Edmonds~\cite{edmonds1965maximum} is that this polytope has an exponential number of
facets, even though any linear function can be optimized over it in strongly polynomial time.
The question of whether the matching polytope admits an extended formulation of size polynomial in $ n $ was open for a
long time, until Rothvo{\ss}~\cite{Rothvoss2014} proved that its extension complexity is exponential in $ n $.
More recently, it was even proved that this polytope cannot be well approximated by a polytope of low extension
complexity:
\begin{theorem}[{\cite{BraunP2015}, see also~\cite[Thm.~16]{Rothvoss2014}}]
    \label{thmBraun}
    There exist constants $ \alpha, \gamma > 0 $ such that every polytope $ K \subseteq \R^E $ with
    $ \match(n) \subseteq K \subseteq (1 + \frac{\alpha}{n}) \match(n) $ satisfies $\xc(K) > 2^{\gamma n} $.
\end{theorem}
Let us translate this result using the notion of relative distance.
\begin{corollary}
    \label{corBraun}
    There exist constants $ \alpha, \gamma > 0 $ such that $ \match(n) $ admits no $ \frac{\alpha}{n} $-LEF of
    size at most $ 2^{\gamma n} $.
\end{corollary}
\begin{proof}
    Let $ \alpha $ and $ \gamma' $ be the constants defined in Theorem~\ref{thmBraun}.
    We may assume that $ 2^{\gamma' n} > \binom{n}{2} $.
    We have to show that every polyhedron $ K \subseteq \R^E $ with $ \match(n) \subseteq K $ and $ \rdist(\match(n), K)
    \le \frac{\alpha}{n} $ satisfies $ \xc(K) \ge 2^{\gamma n} $ for some constant $ \gamma > 0 $.
    To this end, first note that since $ \rdist(\match(n), K) $ is finite and $ \match(n) $ is bounded, $ K $ must also
    be bounded.
    Defining the polytope $ K' \coloneqq K \cap \R^E_{\ge 0} $, we clearly have $ \rdist(\match(n), K') \le \frac{\alpha}{n} $.
    Since $ \match(n) $ is down-closed, by Lemma~\ref{lemma:rdistMaxGap} we have $ \LPgapMax(\match(n), K') \le
    \frac{\alpha}{n} $, which implies $ \match(n) \subseteq K' \subseteq (1 + \frac{\alpha}{n}) \match(n) $.
    Thus, by Theorem~\ref{thmBraun}, we obtain $ \xc(K') \ge 2^{\gamma' n} $ and hence $ \xc(K) \ge \xc(K') - |E| \ge
    2^{\gamma' n} - \binom{n}{2} \ge 2^{\gamma n} $ for some universal constant $ \gamma > 0 $.
\end{proof}
By Corollary~\ref{corollary:milefHardness} we directly obtain:
\begin{corollary}
    \label{corLowerBoundMatchingRdist}
    There exist constants $ \alpha, \gamma > 0 $ such that any $ \frac{\alpha}{n} $-MILEF of $ \match(n) $ of size at
    most $ 2^{\gamma n} $ has $ \Omega \left( \sfrac{n}{\log n} \right) $ integer variables. \hfill $ \qed $
\end{corollary}
Using Lemma~\ref{lemma:rdistMaxGap}, the above statement can be phrased similarly to Theorem~\ref{thmBraun}.
\begin{corollary}
    \label{corLowerBoundMatchingGap}
    There exist constants $ \alpha, \gamma > 0 $ such that the following holds.
    Let $ K \subseteq \R^E_{\ge 0} $ be a polytope with $ \match(n) \subseteq K \subseteq (1 + \frac{\alpha}{n})
    \match(n) $.
    Then any MILEF of $ K $ of size at most $ 2^{\gamma n} $ has $ \Omega \left( \sfrac{n}{\log n} \right) $ integer
    variables. \hfill $ \qed $
\end{corollary}
While polynomial-size textbook MILEFs for $ \match(n) $ usually require $ \Omega(n^2) $ integer variables, in
Section~\ref{secUpperBoundMatching} we give a simple polynomial-size MILEF for $ \match(n) $ that only uses $ O(n) $
integer variables.
Thus, the lower bound on the number of integer variables in Corollaries~\ref{corLowerBoundMatchingRdist}
and~\ref{corLowerBoundMatchingGap} is tight up to a factor of $ O(\log n) $.
 \subsection{Cut polytope}
\label{secApplicationCut}
Let $ K_n = (V, E) $ be the complete undirected graph on $ n $ vertices, and define a \emph{cut} in $ K_n $ to be a
subset $ F \subseteq E $ that can be written as $ F = \left\{ \{v, w\} \in E : v \in S, \, w \notin S \right\} $ for
some $ S \subseteq V $.
\footnote{
    We highlight that $S$ is allowed to be equal to the empty set or $V$.
    Sometimes, to define cuts, one requires $S\not\in \{\emptyset, V\}$.
    Our discussion can easily be transferred to this case, but is a bit simpler when also allowing the trivial sets
    $S=\emptyset$ and $S=V$.
}
The convex hull $ \cut(n) \subseteq [0,1]^E $ of all characteristic vectors of cuts in $ K_n $ is called the \emph{cut
polytope}.
Recall that optimizing a linear function over $ \cut(n) $ is at least as hard as solving the maximum cut problem, which
is NP-hard.
The cut polytope was the first specific $ 0/1 $-polytope shown to have a super-polynomial (in its dimension) extension
complexity; see~\cite{FioriniMPTW2015}.
More specifically, every LEF for $ \cut(n) $ has size at least exponential in $ n $; see also~\cite{KaibelW2015}.
In what follows, we lift this bound to MILEFs with $k\leq  \kappa n / \log n $ integer variables, for some constant
$\kappa$.
To this end, we make use of the following inapproximability result in~\cite{BraunFPS2015} that refers to the
\emph{correlation polytope}
\[
    P_n \coloneqq \conv \left \{ bb\t : b \in \{0,1\}^n \right \} \subseteq \R^{n \times n},
\]
which is affinely isomorphic to $ \cut(n) $, i.e., there exists an affine bijection $ \pi : \R^E \to \af(P_n) $ with $
\pi(\cut(n)) = P_n $; see~\cite{DeSimone1990}.
\ifbool{SODAfinal}{
\begin{theorem}[{\cite[Thm.~6]{BraunFPS2015}}]
    \label{thm:BraunEtAl}
    There is a constant $ \gamma > 0 $ such that every polyhedron $ K \subseteq \R^{n \times n} $ with
    \[
        P_n \subseteq K \subseteq Q_n
    \]
    satisfies $ \xc(K) \ge 2^{\gamma n} $, where 
    \begin{alignat*}{2}
     Q_n \coloneqq \Big\{ & x \in \R^{n \times n} : \quad \forall a\in \{0,1\}^n, & \\ 
		      & \left( 2 \diag(a) - aa\t \right)\t x \le 2 &  \Big\}. 
    \end{alignat*}

\end{theorem}
}{
    \begin{theorem}[{\cite[Thm.~6]{BraunFPS2015}}]
    \label{thm:BraunEtAl}
    There is a constant $ \gamma > 0 $ such that every polyhedron $ K \subseteq \R^{n \times n} $ with
    \[
        P_n \subseteq K \subseteq Q_n \coloneqq \left\{ x \in \R^{n \times n} : \left( 2 \diag(a) - aa\t \right)\bigcdot x
        \le 2 \quad\forall a \in \{0,1\}^n \right\} 
    \]
    satisfies $ \xc(K) \ge 2^{\gamma n} $.\footnote{For two matrices $A,B\in \mathbb{R}^{n\times n}$, we denote by $A\bigcdot B \coloneqq \sum_{i=1}^n\sum_{j=1}^n A_{ij} B_{ij}$ the Frobenius inner product of $A$ and $B$.}
\end{theorem}
}
Again, let us translate this result using the notion of relative distance.
\begin{corollary}
    \label{corBraunCut}
    There exist constants $ \alpha, \gamma > 0 $ such that $ \cut(n) $ admits no $ \frac{\alpha}{n^2} $-LEF of
    size at most $ 2^{\gamma n} $.
\end{corollary}
\begin{proof}
    Define $ \alpha \coloneqq \frac{1}{3} $ and let $ \gamma > 0 $ be the constant in Theorem~\ref{thm:BraunEtAl}.
    Let $ K \subseteq \R^E $ be any polyhedron with $ \cut(n) \subseteq K $ and $ \rdist(\cut(n), K) \le
    \frac{\alpha}{n^2} $.
    It remains to show that $ \xc(K) \ge 2^{\gamma n} $ holds.
    To this end, let $ \pi : \R^E \to \af(P_n) $ be the affine map that satisfies $ \pi(\cut(n)) = P_n $.
    Clearly, we have $ P_n \subseteq \pi(K) $, as well as $ \rdist(P_n, \pi(C)) \le \frac{\alpha}{n^2} $.

    We claim that $ \pi(K) $ is contained in the set $ Q_n $ (as defined in the statement of
    Theorem~\ref{thm:BraunEtAl}).
    Otherwise, there is some $ a \in \{0,1\}^n $ such that the matrix $ c \coloneqq 2 \diag(a) - aa\t $ satisfies $
    \sup_{x \in \pi(K)} c\bigcdot x > 2 $.
    On the other hand, one has $ \max_{x \in P_n} c\bigcdot x \le 1 $ (see, e.g.,~\cite[Sec. IV]{BraunFPS2015}) as
    well as $ \min_{x \in P_n} c\bigcdot x \ge -n^2 $ (because $ c \in [-1,1]^{n \times n} $ and $ P_n \subseteq [0,1]^{n
    \times n} $).
    By the definition of the relative distance, this would imply $ \rdist(P_n, \pi(K)) \ge \frac{2 - 1}{n^2 + 1} >
    \frac{\alpha}{n^2} $, a contradiction.

    Thus, we have $ P_n \subseteq \pi(K) \subseteq Q_n $ and hence by Theorem~\ref{thm:BraunEtAl} we obtain $ \xc(\pi(K)) \ge
    2^{\gamma n} $.
    The claim follows since we have $ \xc(K) \ge \xc(\pi(K)) $.
\end{proof}
By Corollary~\ref{corollary:milefHardness} we directly obtain:
\begin{corollary}
    \label{corLowerBoundCut}
    There exist constants $ \alpha, \gamma > 0 $ such that any $ \frac{\alpha}{n^2} $-MILEF of $ \cut(n) $ of size at
    most $ 2^{\gamma n} $ has $ \Omega \left( \sfrac{n}{\log n} \right) $ integer variables. \hfill $ \qed $
\end{corollary}
Several known polynomial-size MILEFs for $ \cut(n) $ use $ \Omega(n^2) $ integer variables.
However, similar to the case for the matching polytope, there are simple polynomial-size MILEFs for $ \cut(n) $ that
only use $ O(n) $ integer variables; see Section~\ref{secUpperBoundCut}.
Again, the bound on the number of integer variables given in Corollary~\ref{corLowerBoundCut} is tight up to a factor of
$ O(\log n) $.
 \subsection{Traveling salesman polytope}
\label{secApplicationTSP}
In this section, we use our result on the matching polytope to obtain a lower bound on the number of integer variables
in MILEFs for the \emph{traveling salesman polytope} $ \tsp(n) \subseteq \R^E $, which is defined as the convex hull of the
characteristic vectors of all Hamiltonian cycles in $ K_n = (V, E) $.
It is known that there is a constant $ c > 0 $ such that for every $ n $, there exists a face of $ \tsp(n) $ that can be
affinely projected onto $ \match(n') $, where $ n' \ge c n $, see, e.g.,~\cite[Proof of Thm.~2]{yannakakis_1991_expressing}.
By the following lemma, this implies that whenever $ \tsp(n) $ admits a MILEF of complexity $ (m, k) $, then also $
\match(n') $ admits a MILEF of the same complexity.
\begin{lemma}
    \label{lemProjectionOfFace}
    Let $ P \subseteq \R^d $ and $ P' \subseteq \R^{d'} $ be non-empty polyhedra such that $ P' $ is an affine projection of a face of $ P $.
    If $ P $ admits a MILEF of complexity $ (m, k) $, then also $ P' $ admits a MILEF of complexity $ (m, k) $.
\end{lemma}
\begin{proof}
    By the hypotheses, there is a face $F$ of $P$ and an affine map $\tau:\R^d \to \R^{d'}$ such that $P'=\tau(F)$. 
    Additionally, there exists a polyhedron $Q\subseteq \R^\ell$ with at most $m$ facets, and affine maps $ \sigma:\R^\ell \to \R^k $ and $ \pi :\R^\ell \to \R^{d} $ 
    such that $P=\pi(Q_\sigma)$, where $Q_{\sigma}\coloneqq \conv(Q\cap \sigma^{-1}(\mathbb{Z}^k))$.
    
    We remark that $F$ is a face of $P$ if and only if there is an affine map $\phi:\R^d\to \R$ such that $\phi(P)\subseteq \R_{\geq 0}$ and $F=P\cap \phi^{-1}(0)$. 
    Now, define the affine subspace $H\coloneqq (\phi \circ \pi)^{-1} (0)$ in $\R^\ell$, and notice that $Q_\sigma \cap H$ is a face of $Q_\sigma$ 
    by the above-mentioned characterization of a face and the fact that $(\phi\circ \pi) (Q_\sigma) = \phi(P)\subseteq \R_{\geq 0}$.
    This implies that $Q_\sigma \cap H = (Q\cap H)_\sigma \coloneqq \conv(Q \cap H \cap \sigma^{-1}(\Z^k)) $.
    \footnote{
    Explicitly, the claim is $\conv(Q\cap \sigma^{-1}(\Z^k))\cap H=\conv(Q\cap \sigma^{-1}(\Z^k) \cap H)$. 
    The inclusion ``$\supseteq$" follows immediately from the fact that $\conv(Q\cap \sigma^{-1}(\mathbb{Z}^k))\cap H$ is a convex set containing $Q\cap \sigma^{-1}(\mathbb{Z}^k) \cap H$, and must thus contain $\conv(Q\cap \sigma^{-1}(\mathbb{Z}^k) \cap H)$, which is the smallest convex set containing $Q\cap \sigma^{-1}(\mathbb{Z}^k) \cap H$.
For the opposite inclusion, consider a point $x\in \conv(Q\cap \sigma^{-1}(\Z^k))\cap H$. 
    It must be a convex combination of some points $(x_i)_{i\in I}$ in $Q\cap \sigma^{-1}(\Z^k)$. 
    As each point $x_i$ is in $Q_\sigma$, we must have $(\phi\circ \pi)(x_i)\geq 0$; however, $x$ being in $H$ implies that $(\phi\circ \pi)(x)= 0$, 
    which forces all these inequalities to be tight, and thus $x_i\in H$  for each $i\in I$. This proves that $x\in \conv(Q\cap \sigma^{-1}(\Z^k) \cap H)$, as desired.
    }
    Thus, we obtain
    \begin{align*}
        P' = \tau(F) &= \tau(P\cap \phi^{-1}(0))\\
        & = \tau(\pi(Q_\sigma) \cap \pi(H))\\
        & = (\tau \circ \pi) (Q_\sigma \cap H)\\
        & = (\tau \circ \pi) ((Q \cap H)_\sigma)\enspace,
    \end{align*}
    and hence $ P' $ admits the MILEF $(Q\cap H, \sigma, \tau\circ \pi)$, which is of complexity $(m', k)$, 
    where $ m' $ is the number of facets of $ Q \cap H $.
    Since the number of facets of $ Q \cap H $ is at most the number of facets of $ Q $, we obtain $ m' \le m $, which yields the claim.
\end{proof}
By Corollary~\ref{corLowerBoundMatchingRdist} we directly obtain:
\begin{corollary}
    \label{corLowerBoundTSP}
    There exists a constant $ \gamma > 0 $ such that any MILEF for $ \tsp(n) $ of size at most $ 2^{\gamma n} $ has $
    \Omega \left( \sfrac{n}{\log n} \right) $ integer variables. \hfill $ \qed $
\end{corollary}
\ifbool{SODAfinal}{
While most polynomial-size textbook MILEFs for $ \tsp(n) $ require $ \Omega(n^2) $ integer variables, 
there exists a polynomial-size MILEF for $ \tsp(n) $ that uses only $ O(n \log n) $ integer variables. 
The explicit construction of such MILEF is deferred to the extended version of this paper, due to space constraints.
}{
While most polynomial-size textbook MILEFs for $ \tsp(n) $ require $ \Omega(n^2) $ integer variables, in
Section~\ref{secUpperBoundTSP} we give a polynomial-size MILEF for $ \tsp(n) $ that only uses $ O(n \log n) $ integer
variables.
}
Thus, the above bound on the number of integer variables is tight up to a factor of $ O(\log^2 n) $.
 \subsection{Stable set polytope}
\label{secApplicationStableSet}
The \emph{stable set polytope} $ \stab(G) \subseteq \R^V $ of an undirected graph $ G = (V, E) $ is defined as the
convex hull of characteristic vectors of all stable sets in $ G $.
While for some graphs $ G $ the polytope $ \stab(G) $ can be easily described, it is an arguably complicated polytope
in general.
As an example, in~\cite[Lem. 8]{FioriniMPTW2015} it was shown that for every $ n $, there exists a graph $ H_n $ on $
n^2 $ vertices such that a face of $ \stab(H_n) $ can be affinely projected onto $ \cut(n + 1) $.
Thus, using Lemma~\ref{lemProjectionOfFace} and Corollary~\ref{corLowerBoundCut} we conclude:
\begin{corollary}
    \label{corLowerBoundStableSet}
    There exists a constant $ \gamma > 0 $ such that the following holds.
    For every $ n $, there exists an $ n $-vertex graph $ G $ such that any MILEF of $ \stab(G) $ of size at most $
    2^{\gamma \sqrt{n}} $ has $ \Omega \left( \sfrac{\sqrt{n}}{\log n} \right) $ integer variables.
\end{corollary}
We highlight that the above result can also be reduced from our MILEF extension complexity result for matchings, i.e., Corollary~\ref{corLowerBoundMatchingRdist}. This follows from the fact that the matching polytope of any graph $G$ is the stable set polytope of the corresponding line graph, whose number of vertices is equal to the number of edges in $G$. Hence, if $G=K_n$, then the matching polytope of $K_n$ is the stable set polytope of a graph with $O(n^2)$ many vertices.

Note that $ \stab(G) = \conv \left \{ x \in \{0,1\}^V : x_v + x_w \le 1 \ \forall \{v,w\} \in E \right \} $, and hence, $
\stab(G) $ admits a polynomial-size MILEF with $ n $ integer variables, for every $ n $-vertex graph $ G $.
However, we are not aware of polynomial-size MILEFs with $ o(n) $ integer variables.
In particular, we believe that the bound in Corollary~\ref{corLowerBoundStableSet} can be significantly improved.
\ifbool{SODAfinal}{}{We comment on this issue in Section~\ref{secTowards}.}

 \subsection{Knapsack polytope}
\label{secApplicationKnapsack}
Given item sizes $ a = (a_1,\dotsc,a_n) \in \R^n_{\ge 0} $ and a capacity $ B \ge 0 $, the corresponding \emph{knapsack
polytope} is defined as $ \knap(a, B) \coloneqq \conv \left \{ x \in \{0,1\}^n : a\t x \le B \right\} $.
Similar to the case of stable set polytopes, for certain item sizes and capacities the corresponding knapsack polytopes
have a simple structure.
In general, however, knapsack polytopes turn out to be complicated polytopes.
Indeed, in~\cite{PokuttaV2013} it is shown that for every $ n $, there exist item sizes $ a \in \R^{O(n^2)}_{\ge 0} $
and a capacity $ B \ge 0 $ such that $ \cut(n) $ is an affine projection of a face of $ \knap(a, B) $.
Analogous to the previous section, using Lemma~\ref{lemProjectionOfFace} and Corollary~\ref{corLowerBoundCut} we
conclude:
\begin{corollary}
    \label{corLowerBoundKnapsack}
    There exists a constant $ \gamma > 0 $ such that the following holds.
    For every $ n $, there exist item sizes $ a \in \R^n_{\ge 0} $ and a capacity $ B \ge 0 $ such that any MILEF of $
    \knap(a, B) $ of size at most $ 2^{\gamma \sqrt{n}} $ has $ \Omega \left( \sfrac{\sqrt{n}}{\log n} \right) $ integer
    variables. \hfill $ \qed $
\end{corollary}
Clearly, by its definition, $ \knap(a, B) $ admits a linear-size MILEF with $ n $ integer variables.
While we are not aware of any other polynomial-size MILEF for general knapsack polytopes that uses $ o(n) $ integer
variables, it is not clear to us whether the bound in Corollary~\ref{corLowerBoundKnapsack} can be significantly
improved.
 \subsection{Dominant of the $ V $-join polytope}
\label{secApplicationVJoin}
In this section, we consider the $ V $-join polytope of $ K_n $ and in particular its dominant.
Since both polyhedra contain the perfect matching polytope as a face, it is not surprising that we obtain lower bounds
on the complexity of MILEFs of these polyhedra.
However, the main purpose of this section is to obtain lower bounds for \emph{approximate} (MI)LEFs, which will be
essential for establishing lower bounds for the dominant of the odd cut polytope in the next section.

Let $ n $ be even and $ K_n = (V, E) $ be the complete graph on $ n $ vertices.
Recall that an edge subset $ F \subseteq E $ is a called a \emph{$V$-join} in $ K_n $ if every vertex in $ (V, F) $ has
odd degree.
The \emph{$V$-join polytope} of $ K_n $ is defined as the convex hull of the characteristic vectors of all $V$-joins in
$ K_n $ and is denoted by $ \vjoin(n) \subseteq \R^E $.
The \emph{dominant} of the $ V $-join polytope is defined as $ \vjoindom(n) \coloneqq \vjoin(n) + \R^E_{\ge 0} $.

In the next statement, we derive a lower bound on approximate LEFs for $ \vjoindom(n) $ by exploiting the following
relation between $ \vjoindom(n) $ and $ \match(n) $.
First, note that every $ V $-join has cardinality at least $ n/2 $, and hence the set $ F \coloneqq \{ x \in \vjoindom(n)
: \onesvec\t x = n/2 \} $ is a face of $ \vjoindom(n) $.
Furthermore, a subset of edges of $ E $ is a $ V $-join of cardinality $ n/2 $ if and only if it is a \emph{perfect
matching} in $ K_n $, i.e., a matching of cardinality $ n/2 $.
Since every matching consists of at most $ n/2 $ edges, we have that
\[
    \pmatch(n) \coloneqq \left \{ x \in \match(n) : \onesvec\t x = n/2 \right \} = F
\]
is also a face of $ \match(n) $.
The polytope $ \pmatch(n) $ is called the \emph{perfect matching polytope}.
Furthermore, it is easy to see that $ \match(n) = \{ y \in \R^E_{\ge 0} : y \le x \text{ for some } x \in \pmatch(n) \}
$ holds.
In what follows, we use all these relations together with Theorem~\ref{thmBraun} to obtain a similar statement for $
\vjoindom(n) $.
\begin{theorem}
    \label{thm:LEFlowboundVjoin}
    There are constants $ \alpha, \gamma > 0 $ such that, for every $ n $ even, any polyhedron $ K $ with
    $ \vjoindom(n) \subseteq K \subseteq (1 - \frac{\alpha}{n^4}) \vjoindom(n) $ satisfies $ \xc(K) > 2^{\gamma n} $.
\end{theorem}
\begin{proof}
    For brevity, we write $ \vjoindom $, $ \match $, and $ \pmatch $ as shorthands for $ \vjoindom(n) $, $ \match(n) $,
    and $ \pmatch(n) $, respectively.
    Consider any polyhedron $ K $ such that $ \vjoindom \subseteq K \subseteq (1 - \eps) \vjoindom $, where $ \eps $
    will be fixed later, and consider the hyperplanes $ H = \{ x \in \R^E : \onesvec\t x = n/2 \} $ and $ H_0 = \{ x \in
    \R^E : \onesvec\t x = 0 \} $.
    To better structure the proof, we divide it into three claims.

    \medskip

    \noindent
    \textit{Claim: For any $ c \in \R^E $ with $ \|c\|_2 \leq 1 $, we have $ \min_{x \in \vjoindom} (c + 2n \cdot
    \onesvec)\t x = \min_{x \in \pmatch} (c + 2n \cdot \onesvec)\t x $.}

    \smallskip

    \begin{adjustwidth}{0.5cm}{}
        \noindent
        Since $ \pmatch \subseteq \vjoindom $, it suffices to show that $ \alpha \coloneqq \min_{x \in \vjoindom} (c + 2n \cdot
        \onesvec)\t x $ is attained by a point in $ \pmatch $.
        To this end, let $ y $ be any vertex of $ \pmatch $ and observe that we have
        \ifbool{SODAfinal}{
        \begin{align*}
         \alpha & \le (c + 2n \cdot \onesvec)\t y = c\t y + 2n \|y\|_1 \\
		& \le \|y\|_2 + 2n \|y\|_1 \\
		& \le (2n + 1) \|y\|_1 = (2n + 1) \frac{n}{2}.
        \end{align*}
        }{
        \[
            \alpha
            \le (c + 2n \cdot \onesvec)\t y
            = c\t y + 2n \|y\|_1
            \le \|y\|_2 + 2n \|y\|_1
            \le (2n + 1) \|y\|_1
            = (2n + 1) \frac{n}{2}.
        \]
        }
        Next, since $ c + 2n \cdot \onesvec $ is nonnegative, $ \alpha $ is finite and hence attained by a vertex $ x $
        of $ \vjoindom $.
        We claim that $ x $ must be contained in $ \pmatch $.
        Indeed, otherwise, $ x $ would satisfy $ \|x\|_1 = \onesvec\t x \ge \frac{n}{2} + 1 $, and hence
        \ifbool{SODAfinal}{
        \begin{align*}
         \alpha & = (c + 2n \cdot \onesvec)\t x = c\t x + 2n \|x\|_1 \\
		& \ge - \|x\|_2 + 2n \|x\|_1 \\
		& \ge (2n - 1) \|x\|_1 \ge (2n - 1) \left(\frac{n}{2}+1\right),
        \end{align*}
        }{
         \[
            \alpha
            = (c + 2n \cdot \onesvec)\t x
            = c\t x + 2n \|x\|_1
            \ge - \|x\|_2 + 2n \|x\|_1
            \ge (2n - 1) \|x\|_1
            \ge (2n - 1) \left(\frac{n}{2}+1\right),
        \]
        }
        which contradicts the previous inequality whenever $n>1$. \hfill $ \diamond $
    \end{adjustwidth}

    \medskip

    \noindent
    Next, we show that $ K \cap H $ approximates $ \pmatch $ well.

    \medskip

    \noindent
    \textit{Claim: for any $ c \in H_0 $ with $ \|c\|_2 \leq 1 $, we have $ \max_{x \in K \cap H} c\t x \leq \max_{x \in
    \pmatch} c\t x + \eps n^2 $.}

    \smallskip

    \begin{adjustwidth}{0.5cm}{}
        \noindent
        First, let $ c \in \R^E $ with $ \|c\|_2 \le 1 $ be arbitrary.
        Since $ c + 2n \cdot \onesvec $ is nonnegative and since $ K \subseteq (1 - \eps) \vjoindom $, we obtain
        \ifbool{SODAfinal}{
        \begin{align*}
         & \min_{x \in K} (c + 2n \cdot \onesvec)\t x \\
         & \geq (1 - \eps) \min_{x \in \vjoindom} (c + 2n \cdot \onesvec)\t x \\
         & = (1 - \eps ) \min_{x \in \pmatch} (c + 2n \cdot \onesvec)\t x,
        \end{align*}
        }{
        \[
            \min_{x \in K} (c + 2n \cdot \onesvec)\t x
            \geq (1 - \eps) \min_{x \in \vjoindom} (c + 2n \cdot \onesvec)\t x
            = (1 - \eps ) \min_{x \in \pmatch} (c + 2n \cdot \onesvec)\t x,
        \]
        }
        where the equality follows from the previous claim.
        This clearly implies that
        \ifbool{SODAfinal}{
        \begin{align*}
            & \min_{x \in K \cap H} c\t x \\
            &= \min_{x \in K \cap H} (c + 2n \cdot \onesvec) \t x - n^2 \\
            & \ge \min_{x \in K} (c + 2n \cdot \onesvec) \t x - n^2 \\
            & \ge (1 - \eps ) \min_{x \in \pmatch} (c + 2n \cdot \onesvec)\t x - n^2 \\
            & = (1 - \eps ) \min_{x \in \pmatch} c\t x + (1 - \eps) n^2 - n^2 \\
            & = (1 - \eps ) \min_{x \in \pmatch} c\t x - \eps n^2
        \end{align*}
        }{
        \begin{align*}
            \min_{x \in K \cap H} c\t x
            = \min_{x \in K \cap H} (c + 2n \cdot \onesvec) \t x - n^2
            & \ge \min_{x \in K} (c + 2n \cdot \onesvec) \t x - n^2 \\
            & \ge (1 - \eps ) \min_{x \in \pmatch} (c + 2n \cdot \onesvec)\t x - n^2 \\
            & = (1 - \eps ) \min_{x \in \pmatch} c\t x + (1 - \eps) n^2 - n^2 \\
            & = (1 - \eps ) \min_{x \in \pmatch} c\t x - \eps n^2
        \end{align*}
        }
        holds for every $ c \in \R^E $ with $ \|c\|_2 \le 1 $.
        Equivalently, we obtain that
        \[
            \max_{x \in K \cap H} c\t x \leq (1 - \eps) \max_{x \in \pmatch} c\t x + \varepsilon n^2
        \]
        holds for every $ c \in \R^E $ with $ \|c\|_2 \le 1 $.
        Now let $ c \in H_0 $ with $ \|c\|_2 \le 1 $.
        Since $ c $ satisfies $ c\t \onesvec = 0 $ and since $ \frac{1}{n - 1} \onesvec $ is contained in $ \pmatch $,
        we clearly have $ \max_{x \in \pmatch} c\t x \geq 0 $ and hence we obtain the claimed inequality. \hfill $
        \diamond $
    \end{adjustwidth}

    \medskip

    \noindent
    Define $ \bar{K} \coloneqq \{ x \in \R^E_{\geq 0} : x \leq y \text{ for some } y \in K \cap H \} $ and observe that we have
    \[
        \xc(\bar{K}) \le 2 |E| + \xc(K \cap H) \le n^2 + \xc(K).
    \]
    Thus, it is enough to prove that $ \xc(\bar{K}) \ge 2^{\gamma' n} $ holds for some universal constant $ \gamma' > 0$.

    Since $ \pmatch \subseteq K \cap H $ and $ \match = \{ x \in \R^E_{\ge 0} : x \le y \text{ for some } y \in \pmatch \} $, we have $ \match \subseteq \bar{K} $.
    Together with the following claim we finally obtain $(1-\epsilon n^3)\bar{K} \subseteq \match \subseteq \bar{K}$. This implies $\xc(\bar{K}) = 2^{\Omega(n)}$, as desired, by setting $\epsilon \coloneqq \frac{\alpha}{n^4}$, where $\alpha$ is the constant from Theorem~\ref{thmBraun}, and using Theorem~\ref{thmBraun}.

    \medskip

    \noindent
    \textit{Claim: We have $ (1 - \eps n^3) \bar{K} \subseteq \match $.}

    \smallskip

    \begin{adjustwidth}{0.5cm}{}
        \noindent
        As $ \bar{K} $ and $ \match $ are down-closed, it suffices to show that $ (1 - \eps n^3) \max_{x \in
        \bar{K}} \bar{c}\t x \le \max_{x \in \match} \bar{c}\t x $ holds for every $ \bar{c} \in \R^E_{\ge 0} $ with $
        \|\bar{c}\|_2 = 1 $.
        To this end, fix such a $ \bar{c} $ and write it as $ \bar{c} = c + \lambda \cdot \onesvec $, where $ c \in H_0
        $ with $ \|c\|_2 \le 1 $ and $ \lambda \ge 0 $.
        From the previous claim, we obtain
        \ifbool{SODAfinal}{
        \begin{align*}
            & \max_{x \in \bar{K}}  \bar{c}\t x - \max_{x \in \match} \bar{c}\t x \\
            & = \max_{x\in K \cap H} \bar{c}\t x - \max_{x \in \pmatch} \bar{c}\t x \\
            & = \max_{x\in K \cap H} (c + \lambda \cdot \onesvec)\t x - \max_{x \in \pmatch} (c + \lambda \cdot \onesvec)\t x \\
            & = \max_{x\in K \cap H} c\t x - \max_{x \in \pmatch} c\t x \\
            & \le \eps n^2
            \le \eps n^3 \cdot \max_{x \in \bar{K}} \bar{c}\t x,
        \end{align*}
        }{
        \begin{align*}
            \max_{x \in \bar{K}} \bar{c}\t x - \max_{x \in \match} \bar{c}\t x
            & = \max_{x\in K \cap H} \bar{c}\t x - \max_{x \in \pmatch} \bar{c}\t x \\
            & = \max_{x\in K \cap H} (c + \lambda \cdot \onesvec)\t x - \max_{x \in \pmatch} (c + \lambda \cdot \onesvec)\t x \\
            & = \max_{x\in K \cap H} c\t x - \max_{x \in \pmatch} c\t x \\
            & \le \eps n^2
            \le \eps n^3 \cdot \max_{x \in \bar{K}} \bar{c}\t x,
        \end{align*}
        }
        where the first equality follows from the fact that $ \bar{c} $ is nonnegative, and the last inequality
        is implied by $ \max_{x \in \bar{K}} \bar{c}\t x \ge \frac{1}{n}$, which holds due to the following.
        As $ \frac{1}{n-1} \onesvec \in \pmatch \subseteq \bar{K} $, we have $ \max_{x \in \bar{K}} \bar{c}\t x \ge
        \bar{c} (\frac{1}{n-1} \onesvec) \ge \frac{1}{n} \|\bar{c}\|_1 \ge \frac{1}{n} \|\bar{c} \|_2 = \frac{1}{n} $.
        \hfill $ \diamond $
    \end{adjustwidth}
\end{proof}
Next, we demonstrate that a statement as in Theorem~\ref{thm:LEFlowboundVjoin} implies a particular inapproximability result in terms of relative distance.
To this end, for a set $ P \subseteq \R^d $ we define $ P\dom := P + \R^d_{\ge 0} $.
\begin{lemma}
    \label{lemInapproxDomToInapproxRdist}
    Let $ P \subseteq [0,1]^d $ be a $ 0/1 $-polytope and let $ \eps \in (0,1) $, $ M > 0 $ such that every polyhedron $ K \subseteq \R^d $ with $ P\dom \subseteq K \subseteq (1 - \eps) P\dom $ satisfies $ \xc(K) \ge M $.
    Then $ P\dom \cap [0,1]^d $ does not admit an $ \frac{\eps}{d} $-LEF of size less than $ M - 3d $.
\end{lemma}
\begin{proof}
    We have to show that every polyhedron $ K'' \subseteq \R^d $ with $ P\dom \cap [0,1]^d \subseteq K'' $ and $ \rd(P\dom \cap [0,1]^d, K'') \le \frac{\eps}{d} $ satisfies $ \xc(K'') \ge M - 3d $.
    To this end, define $ K' := K'' \cap [0,1]^d $ and observe that we have $ P\dom \cap [0,1]^d \subseteq K' $ and $ \rd(P\dom \cap [0,1]^d, K') \le \rd(P\dom \cap [0,1]^d, K'') \leq \frac{\eps}{d} $.
    Note that the latter implies $ \dim(P\dom \cap [0,1]^d) = \dim(K') $.
    Since $ P\dom \cap [0,1]^d $ is an up-closed $ 0/1 $-polytope, Lemma~\ref{lemma:rdistMinGap}~\ref{itemRdistMinGapLowerBoundRdist} implies
    \[
        \frac{\LPgapMin(P\dom \cap [0,1]^d, K')}{1 + \LPgapMin(P\dom \cap [0,1]^d, K')} \le d \cdot \rd(P\dom \cap [0,1]^d, K') \le \eps,
    \]
    which is equivalent to $ \LPgapMin(P\dom \cap [0,1]^d, K') \le \frac{\eps}{1 - \eps} $.
    Note that this implies
    \[
        \LPgapMin(P\dom, (K')\dom)
        = \LPgapMin((P\dom \cap [0,1]^d)\dom, (K')\dom)
        \le \frac{\eps}{1 - \eps},
    \]
    where the equality follows from $ P \subseteq [0,1]^d $ implying $ (P\dom \cap [0,1]^d)\dom = P\dom $.
    Thus, for the polyhedron $ K := (K')\dom \supseteq P\dom $ we also obtain $ \LPgapMin(P\dom, K) \le \frac{\eps}{1 - \eps} $.
    Since $ P\dom, K \subseteq \R^d_{\ge 0} $ are equal to their dominants, this implies $ (1 + \frac{\eps}{1 - \eps}) K \subseteq P\dom $, which is equivalent to $ K \subseteq (1 - \eps) P\dom $.
    By the assumption, we conclude that $ \xc(K) \ge M $ holds.
    Recall that $ K $ is defined via
    \[
        K = \{ x + y : x \in K'', \, x \in [0,1]^d, \, y \in \R^d_{\ge 0} \},
    \]
    and hence
    \[
        \xc(K) \le \xc(K'') + \xc([0,1]^d) + \xc(\R^d_{\ge 0}) \le \xc(K'') + 2d + d = \xc(K'') + 3d,
    \]
    which shows $ \xc(K'') \ge \xc(K) - 3d \ge M - 3d $, as claimed.    
\end{proof}
The following inapproximability result for $ \vjoindom(n) \cap [0,1]^E $ is a direct consequence of Theorem~\ref{thm:LEFlowboundVjoin} and Lemma~\ref{lemInapproxDomToInapproxRdist}.
\begin{corollary}
    \label{corVjoinNoLEF}
    There exist constants $ \alpha, \gamma > 0 $ such that, for every $ n $ even, $ \vjoindom(n) \cap [0,1]^E $ does not
    admit an $ \frac{\alpha}{n^6} $-LEF of size at most $ 2^{\gamma n} $.
\end{corollary}
Using Corollary~\ref{corollary:milefHardness}, this immediately implies:
\begin{corollary}
    \label{corLowerBoundMILEFVjoinDomIntersectedCube}
    There exist constants $ \alpha, \gamma > 0 $ such that, for every $ n $ even, any $ \frac{\alpha}{n^6} $-MILEF of $
    \vjoindom(n) \cap [0,1]^E $ of size at most $ 2^{\gamma n} $ has $ \Omega \left( \sfrac{n}{\log n} \right) $ integer
    variables. 
\end{corollary}
Finally, we use the following lemma to deduce an inapproximability result for MILEFs of $ \vjoindom(n) $.
\begin{lemma}
    \label{lemmaAnotherAuxLemmaForVjoins}
    Let $ P \subseteq [0,1]^d $ be a $ 0/1 $-polytope and let $ \eps, M, k > 0 $ such that every $ \eps $-MILEF of $ P\dom \cap [0,1]^d $ of size at most $ M $ has at least $ k $ integer variables.
    Furthermore, let $ K $ be any polyhedron with $ P\dom \subseteq K \subseteq (1 - \frac{\eps}{d + \eps}) P\dom $.
    Then every MILEF of $ K $ of size at most $ M - 2d $ has at least $ k $ integer variables.
\end{lemma}
\begin{proof}
    We may assume that $ \dim(K) = \dim(P\dom) $, otherwise intersect $ K $ with the affine hull $ H $ of $ P\dom $ and observe that if $ K $ has a MILEF of a certain complexity, then $ K \cap H $ has a MILEF of the same complexity.

    By the assumption, we have $ \frac{\eps + d}{d} K \subseteq P\dom $.
    First, we claim that this implies
    \begin{equation}
        \label{eqAnotherAuxLemmaForVjoins1}
        \LPgapMin \left( P\dom \cap [0,1]^d, K \cap [0,1]^d \right) \le \frac{\eps + d}{d} - 1.
    \end{equation}
    To see this, let $ c \in \R^d_{\ge 0} $.
    We have to show that
    \begin{equation}
        \label{eqAnotherAuxLemmaForVjoins2}
        \min_{x \in P\dom \cap [0,1]^d} c\t x \le \frac{\eps + d}{d} \min_{y \in K \cap [0,1]^d} c\t y
    \end{equation}
    holds.
    Note that since $ P \subseteq [0,1]^d $ we have
    \[
        \min_{x \in P\dom \cap [0,1]^d} c\t x = \min_{x \in P\dom} c\t x.
    \]
    Let $ y^\star \in K \cap [0,1]^d $ such that $ c\t y^\star = \min_{y \in K \cap [0,1]^d} c\t y $.
    Since $ \frac{\eps + d}{d} K \subseteq P\dom $, we obtain $ \frac{\eps + d}{d} y^\star \in P\dom $ and hence
    \[
        \min_{x \in P\dom \cap [0,1]^d} c\t x = \min_{x \in P\dom} c\t x \le \frac{\eps + d}{d} c\t y^\star = \frac{\eps + d}{d} \min_{y \in K \cap [0,1]^d} c\t y,
    \]
    which shows~\eqref{eqAnotherAuxLemmaForVjoins2} and hence we have established~\eqref{eqAnotherAuxLemmaForVjoins1}.

    Thus, using the facts that $ P\dom \cap [0,1]^d $ is a $ 0/1 $-polytope, $ \dim(K \cap [0,1]^d) = \dim(P\dom \cap [0,1]^d) $, and $ P\dom \cap [0,1]^d \subseteq K \cap [0,1]^d $, we can invoke Lemma~\ref{lemma:rdistMinGap}~\ref{itemRdistMinGapLowerBoundMinGap}, which, together with inequality~\eqref{eqAnotherAuxLemmaForVjoins1} implies
    \begin{align*}
        \rd \left(P\dom \cap [0,1]^d, K \cap [0,1]^d\right)
        & \le d \cdot \LPgapMin\left(P\dom \cap [0,1]^d, K \cap [0,1]^d \right) \\
        & \le d \cdot \left( \frac{\eps + d}{d} - 1 \right) \\
        & = \eps.
    \end{align*}
    Suppose now that $ K $ has a MILEF of size at most $ M - 2d $ with $ k' $ integer variables.
    Then $ K \cap [0,1]^d $ has a MILEF of size at most $ M $ with $ k' $ integer variables.
    This means that $ P\dom \cap [0,1]^d $ has an $ \eps $-MILEF of size at most $ M $ with $ k' $ integer variables.
    By the assumption we must have $ k' \ge k $, which yields the claim.
\end{proof}
Finally, we are able to prove the following lower bound on the complexity of MILEFs approximating the dominant of the $ V $-join polytope, which is a direct consequence of Corollary~\ref{corLowerBoundMILEFVjoinDomIntersectedCube} and Lemma~\ref{lemmaAnotherAuxLemmaForVjoins}.
\begin{corollary}
    \label{corLowerBoundMILEFVjoin}
    There are constants $ \alpha, \gamma > 0 $ such that the following holds.
    Let $ n $ be even and $ K $ be any polyhedron with $ \vjoindom(n) \subseteq K \subseteq (1 - \frac{\alpha}{n^8})
    \vjoindom(n) $.
    Then every MILEF of $ K $ of size at most $ 2^{\gamma n} $ has $ \Omega( \sfrac{n}{\log n} ) $ integer variables.
\end{corollary}
We remark that the $ V $-join polytope (and hence also its dominant) has a polynomial-size (exact) MILEF with $ O(n) $
integer variables, see Section~\ref{secUpperBoundMatching}.
 \subsection{Dominant of the odd-cut polytope}
\label{secApplicationOddCut}
Using the bounds obtained in the previous section, we are ready to provide lower bounds on the complexity of MILEFs for
the dominant of the odd-cut polytope.
Let $ n $ be even and let $ K_n = (V, E) $ be the complete undirected graph on $ n $ vertices.
An \emph{odd cut} in $ K_n $ is defined as a subset $ F \subseteq E $ that can be written as $ F = \left\{ \{v, w\} \in
E : v \in S, \, w \notin S \right\} $ for some set $ S \subseteq V $ that has odd cardinality.
The \emph{odd-cut polytope} $ \oddcut(n) \subseteq \R^E $ is defined as the convex hull of the characteristic vectors of all
odd cuts in $ K_n $, and its dominant is defined as $ \oddcutdom(n) \coloneqq \oddcut(n) + \R^E_{\ge 0} $.

It is easy to check that every odd cut intersects every $ V $-join.
However, an even stronger and well-known link between $ \oddcutdom(n) $ and $ \vjoindom(n) $ is that these polyhedra are so-called \emph{blockers} of each other.
For a convex set $ K \subseteq \R^d $, its \emph{blocker} $B(K)$ is defined as $ B(K) \coloneqq \{ x \in \R^d_{\ge 0} : y\t x \ge 1 \ \forall y \in K \} $ (see, e.g.,~\cite[Sec.~9]{Schrijver86} for more information on blocking polyhedra).
Using this notation, the mentioned relation reads
\ifbool{SODAfinal}{
\begin{align*}
  B \left( \oddcutdom(n) \right) & = \vjoindom(n) \quad \text{ and } \\ 
  B \left( \vjoindom(n) \right) &= \oddcutdom(n). 
\end{align*}

}{
\[
    B \left( \oddcutdom(n) \right) = \vjoindom(n) \quad \text{ and } \quad B \left( \vjoindom(n) \right) = \oddcutdom(n).
\]
}
Another important fact that we will use in what follows is the observation that every linear extended formulation for a
polyhedron $ P \subseteq \R^d $ can be turned into one for $ B(P) $ by adding at most $ d + 1 $ additional inequalities.
More precisely, we use the following well-known fact (see, e.g., \cite[Prop.~1]{ConfortiKWW2015}):
\begin{equation}\label{eq:xcOfBlocker}
\xc(B(P)) \leq \xc(P) + d + 1 \qquad \forall P\subseteq \mathbb{R}^d\enspace.
\end{equation}
We are ready to transfer Theorem~\ref{thm:LEFlowboundVjoin} to the dominant of the odd-cut polytope:
\begin{corollary}
    \label{corLEFlowboundOddCut}
    There are constants $ \alpha, \gamma > 0 $ such that, for every $ n $ even, any polyhedron $ K $ with
    $ \oddcutdom(n) \subseteq K \subseteq (1 - \frac{\alpha}{n^4}) \oddcutdom(n) $ satisfies $ \xc(K) > 2^{\gamma n} $.
\end{corollary}
\begin{proof}
    Let $ \alpha, \gamma $ denote the constants in the statement of Theorem~\ref{thm:LEFlowboundVjoin}.
    For brevity, let us use the notation $ P \coloneqq \oddcutdom(n) $ and $ \eps \coloneqq \frac{\alpha}{n^4} $, where we may assume
    that $ \eps \in [0,1] $.

    Let $ K $ be a polyhedron with $ P \subseteq K \subseteq (1 - \eps) P $.
    Note that we have $ B((1 - \eps)P) \subseteq B(K) \subseteq B(P) $.
    By $ B((1 - \eps) P) = \frac{1}{1 - \eps} B(P) $ this yields $ B(P) \subseteq (1 - \eps) B(K) \subseteq (1 - \eps)
    B(P) $.
    Since $ B(P) = \vjoindom(n) $, using Theorem~\ref{thm:LEFlowboundVjoin} we obtain $ \xc(B(K)) \ge 2^{\gamma n} $.
    By~\eqref{eq:xcOfBlocker}, we have $ \xc(K) \ge \xc(B(K)) - |E| - 1 \ge 2^{\gamma' n} $ for some
    universal constant $ \gamma' > 0 $, and the claim follows.
\end{proof}
Analogously to the case of the $ V $-join polytope, we obtain from Corollary~\ref{corLEFlowboundOddCut} and Lemma~\ref{lemInapproxDomToInapproxRdist}.
\begin{corollary}
    There exist constants $ \alpha, \gamma > 0 $ such that, for every $ n $ even, $ \oddcutdom(n) \cap [0,1]^E $ does
    not admit an $ \frac{\alpha}{n^6} $-LEF of size at most $ 2^{\gamma n} $. \hfill $ \qed $
\end{corollary}
Using Corollary~\ref{corollary:milefHardness}, this immediately implies:
\begin{corollary}
    \label{corLowerBoundOddCutDom}
    There exist constants $ \alpha, \gamma > 0 $ such that, for every $ n $ even, any $ \frac{\alpha}{n^6} $-MILEF of $
    \oddcutdom(n) \cap [0,1]^E $ of size at most $ 2^{\gamma n} $ has $ \Omega \left( \sfrac{n}{\log n} \right) $
    integer variables.
\end{corollary}
Finally, Corollary~\ref{corLowerBoundOddCutDom} and Lemma~\ref{lemmaAnotherAuxLemmaForVjoins} yield:
\begin{corollary}\label{cor:lowerBoundOddCutDomMILEF}
    There are constants $ \alpha, \gamma > 0 $ such that the following holds.
    Let $ n $ be even and $ K $ be any polyhedron with $ \oddcutdom(n) \subseteq K \subseteq (1 - \frac{\alpha}{n^8})
    \oddcutdom(n) $.
    Then every MILEF of $ K $ of size at most $ 2^{\gamma n} $ has $ \Omega( \sfrac{n}{\log n} ) $ integer variables.
\end{corollary}
We remark that the odd-cut polytope (and hence also its dominant) has a polynomial-size (exact) MILEF with $ O(n) $
integer variables; see Section~\ref{secUpperBoundCut}.
 \subsection{Conic bimodular integer programming}
\label{secApplicationBimodular}
Given $ A \in \Z^{m \times d} $ and $ b \in \Z^m $, consider the problem of optimizing a given linear function 
over the integer hull $P_I\coloneqq \conv (P\cap \Z^d)$ of the polyhedron $ P \coloneqq \{ x \in \R^n : Ax \le b \} $.
Without any further assumption on $ A$ and $b $, this describes a general integer program and is hence $ \NPcomp $-hard to
solve.
A well-known special case in which the problem becomes polynomial-time solvable is when $A$ is \emph{totally unimodular}, 
i.e., the largest absolute value of the determinant of any square submatrix of $A$ is equal to 1. It is a well-known open question in the integer programming community whether integer programs can still be solved efficiently if they are described by an integer constraint matrix $A$ such that the absolute value of any determinant of a square submatrix of $A$ is bounded by some constant $k$.
Recently,~\cite{ArtmannWZ2017} answered this question in the affirmative for $k=2$, by showing that integer programs are tractable if the constraint matrix $A$ is \emph{bimodular}, that is, $A$ is an integer matrix of full column rank such that all determinants of $n\times n$ submatrices of $A$ lie within $\{-2,-1,0,1,2\}$.

In the totally unimodular case, the polynomial-time solvability can be easily explained by observing that $ P $ and its
integer hull $ P_I$ coincide, and hence the problem reduces to solving a linear program.
In contrast, the argumentation in~\cite{ArtmannWZ2017} for the bimodular case is much more involved and gives no
evidence of whether $ P_I $ has a simple polyhedral representation as well, compared to $P$. In this section we show that bimodular integer programs, i.e., integer programs with bimodular constraint matrices, can lead to polyhedra $P_I$ that cannot be described by a small MILEF. This result will follow by showing that the dominant of the odd cut polytope can be captured by a bimodular integer program.

To this end, let $ D = (V, A) $ be the complete digraph on $ n $ vertices and let us consider the polyhedron
\ifbool{SODAfinal}{
\begin{alignat*}{2}
 P \coloneqq \Big\{ & (x, y, z) \in \R^A_{\geq 0} \times \R^V \times \R : & \\
		    & x_{(v,w)} \geq y_w - y_v \quad \forall (v, w) \in A, \\
		    & \sum \nolimits_{v \in V} y_v = 2 z + 1 & \Big\}.
\end{alignat*}

}{
\begin{equation}\label{eq:bimodularOddCutRelax}
     P \coloneqq \left\{ (x, y, z) \in \R^A_{\geq 0} \times \R^V \times \R :
     y_w - y_v \le x_{(v, w)} \ \forall (v, w) \in A, \, \sum \nolimits_{v \in V} y_v = 2 z + 1 \right\}.
\end{equation}
}
First, note that $ P $ is described by a system of linear inequalities with a bimodular coefficient matrix (and an
integer right-hand side).
To see this, observe first that the constraint matrix has full column rank due to the non-negativity constraints. Moreover, notice that $P$ is described by inequalities forming identity matrices and a vertex-arc incidence matrix,
which are totally unimodular, plus an additional row (related to $z$) containing an entry of value 2 in an otherwise
empty column (that of variable $z$).
Thus, by developing over this last column, we see that any determinant of an $ \ell \times \ell $ submatrix with $ \ell
= |A| + |V| + 1 $ is bounded by $2$ in absolute value.
Second, note that the polyhedron $P$ is conic, i.e., there is a vertex for which all constraints are tight, because the
point $(x,y,z)=(0,0,-\frac{1}{2})$ satisfies all linear constraints with equality.
Third, it is easy to see that $ P_I \coloneqq \conv ( P \cap \R^A \times \Z^V \times \Z ) $ can be affinely projected
onto $ \oddcutdom(n) $ (a formal proof is provided in Section~\ref{secUpperBoundCut}).
Optimizing over the integer points of a conic polyhedron $P$ that is described by a bimodular constraint matrix is a \emph{conic bimodular integer program}. By the above discussion and Corollary~\ref{cor:lowerBoundOddCutDomMILEF}, we thus obtain:
\begin{theorem}
    \label{thm:bimodularHard}
    There exists a constant $ \gamma > 0 $ such that the following holds.
    For every $ n $ there is a conic bimodular integer program with $O(n^2)$ variables, such that any MILEF of size at
    most $ 2^{\gamma n} $ for the convex hull of its feasible points requires $\Omega(n / \log n)$ integrality
    constraints.
\end{theorem}
The importance of the fact that our hardness result even holds for \emph{conic} bimodular integer programs is motivated by a result from
Veselov and Chirkov~\cite{veselov_2009_integer}, which implies that it suffices to find an efficient algorithm for conic
bimodular integer programming, to solve any bimodular integer program efficiently.
Thus, a natural approach to solve bimodular integer programs would have been to try to find a compact LEF or MILEF, with few integer variables, that describes the 
feasible solutions to conic bimodular integer programs, thus avoiding the partially involved combinatorial techniques used in~\cite{ArtmannWZ2017}, which is so far the only method to efficiently solve bimodular integer programs. Also, one could have hoped that an approach based on extended formulations may be amenable to extensions beyond the bimodular case.
Theorem~\ref{thm:bimodularHard} shows that this approach cannot succeed.
Still, there is hope that one may be able to design combinatorial approaches that can solve natural generalizations of bimodular integer programs. A step in this direction was done in~\cite{nagele_2018_submodular}.

 \subsection{Large families of $ 0/1 $-polytopes}
As one of the first results establishing non-trivial lower bounds on size of LEFs, it is shown in~\cite{Rothvoss2013}
that for every constant $ \gamma > 0 $ the following holds: If $ \calP $ is any family of $ 0/1 $-polytopes in
$ \R^d $ with $ |\calP| \ge 2^{2^{\gamma d}} $, then there exists a polytope $ P \in \calP $ with $ \xc(P) \ge
2^{\Omega(d)} $.
It is also observed in~\cite{Rothvoss2013} that by the well-known fact that there are doubly-exponentially many
matroids\footnote{
    A matroid is a tuple $M=(N,\mathcal{I})$, where $N$ is a finite ground set and $\mathcal{I}\subseteq 2^N$ is a
    non-empty family of subsets of $N$ satisfying:
    \begin{enumerate*}[label=(\roman*)]
        \item if $I\in \mathcal{I}$ and $J\subseteq I$, then $J\in \mathcal{I}$, and
        \item if $I,J\in \mathcal{I}$ with $|I|< |J|$, then there is an element $e\in J\setminus I$ such that $I\cup
            \{e\}\in \mathcal{I}$.
    \end{enumerate*}
    The matroid polytope $P_M\subseteq [0,1]^N$ that corresponds to $M$ is the convex hull of all characteristic vectors
    of sets in $\mathcal{I}$.
}
on a ground set of cardinality $n$, there is for each $n\in \mathbb{Z}_{\geq 1}$ a matroid on a ground set of size $n$
whose corresponding matroid polytope has exponential (in $n$) extension complexity.
In this section, we extend both results to the mixed-integer setting.

To this end, we make use of a recent generalization of the result in~\cite{Rothvoss2013}.
For two non-empty compact sets $ A, B \subseteq \R^d $ recall that their \emph{Hausdorff distance} with respect to the
Euclidean norm is defined via
\ifbool{SODAfinal}{
$d_H(A, B) \coloneqq$
\[
 \max \left\{ \sup_{a \in A} \inf_{b \in B} \|a - b\|_2, \, \sup_{b \in B} \inf_{a \in A} \|a -
    b\|_2 \right\}.
\]
}{
\[
    d_H(A, B) \coloneqq \max \left\{ \sup_{a \in A} \inf_{b \in B} \|a - b\|_2, \, \sup_{b \in B} \inf_{a \in A} \|a -
    b\|_2 \right\}.
\]
}
\begin{theorem}[{\cite[Thm.~1]{AverkovKW2017}}]
    \label{thm:Averkovetal}
    Let $ \calP $ be a family of polytopes in $ [0,1]^d $ of dimensions at least one with $ 2 \le |\calP| < \infty $.
    Let $ \Delta > 0 $ be such that $ d_H(P, P') \ge \Delta $ holds for every two distinct polytopes $ P, P' \in \calP
    $.
    Then there exists a polytope $ P \in \calP $ with
    \[
        \xc(P) \ge \sqrt{\frac{\log_2 {|\calP|}}{8d(1 + \log_2(2\sqrt{d} / \Delta) + \log_2 \log_2 |\calP|)}}\enspace.
    \]
\end{theorem}
Note that the family~$ \calP $ in Theorem~\ref{thm:Averkovetal} is not restricted to only contain $ 0/1 $-polytopes.
Next, we show that every large enough family of polytopes in $ [0,1]^d $ even contains polytopes that do not admit small
approximate LEFs.
\ifbool{SODAfinal}{
To this end, we make use of the following lemma, whose proof is deferred to the extended version of this paper.
}{
To this end, we make use of the following lemma whose proof is given in Appendix~\ref{appendixRdist}.
}
\begin{lemma}
    \label{lemHausdorffBoundedByRdist}
    Let $ A, B \subseteq [0,1]^d $ be convex sets with $ \emptyset \ne A \subseteq B $.
    Then $ d_H(A, B) \le \sqrt{d} \cdot \frac{\rdist(A, B)}{1 + \rdist(A, B)} $.
\end{lemma}
\begin{proposition}
    For every constant $ \gamma > 0 $ there exists a constant $ \gamma' > 0 $ such that the following holds.
    For every family $ \calP $ of $ 0/1 $-polytopes in $ \R^d $ with $ |\calP| \ge 2^{2^{\gamma d}} $ there exists
    a polytope $ P \in \calP $ that admits no $ \frac{1}{4 d} $-LEF of size at most $ 2^{\gamma' d} $.
\end{proposition}
\begin{proof}
    We may assume that $ |\calP| \ge 2 $ and that $ \calP $ only contains polytopes of dimensions at least one.
    Suppose that every $ P \in \calP $ admits a $ \delta \coloneqq \frac{1}{4 d} $-LEF of size at most $ M $.
    Thus, for every $ P \in \calP $ there is a convex set $ B_P \subseteq \R^d $ with $ P \subseteq B_P $, $ \rdist(P,
    B_P) \le \delta $, and $ \xc(B_P) \le M $.
    Clearly, the set $ C_P \coloneqq B_P \cap [0,1]^d $ satisfies $ P \subseteq C_P $, $ \rdist(P, C_P) \le \delta $, as well
    as $ \xc(C_P) \le M + 2d $.
    By Lemma~\ref{lemHausdorffBoundedByRdist}, we have $ d_H(P, C_P) \le \delta \sqrt{d} $ for every $ P
    \in \calP $, and hence for every two distinct polytopes $ P, P' \in \calP $ we obtain
    \ifbool{SODAfinal}{
    \begin{align*}
     \tfrac{1}{\sqrt{d}} & \le d_H(P, P') \\
			& \le d_H(P, C_P) + d_H(C_P, C_{P'}) + d_H(C_{P'}, P') \\
			& \le d_H(C_P, C_{P'}) + 2 \delta \sqrt{d},
    \end{align*}

    }{
    \[
        \tfrac{1}{\sqrt{d}} \le d_H(P, P') \le d_H(P, C_P) + d_H(C_P, C_{P'}) + d_H(C_{P'}, P') \le
        d_H(C_P, C_{P'}) + 2 \delta \sqrt{d},
    \]
    }
    where the first inequality follows from the fact that for any two distinct $0/1$-polytopes in $\mathbb{R}^d$, we
    have that $\sfrac{1}{\sqrt{d}}$ is a lower bound on their Hausdorff distance,\footnote{
        This can be deduced by observing that the Hausdorff distance of any vertex of the hypercube $[0,1]^d$ to the
        convex hull of all other vertices is $\sfrac{1}{\sqrt{d}}$.
    }
    and the second inequality follows by the triangle inequality for the Hausdorff distance.
    Hence, this implies $ d_H(C_P, C_{P'}) \ge \frac{1}{2 \sqrt{d}} $.
    Applying Theorem~\ref{thm:Averkovetal} to the family~$ \{ C_P : P \in \calP \} $, we obtain that there exists a $ P
    \in \calP $ such that
    \[
        \xc(C_P) \ge \sqrt{\frac{2^{\gamma d}}{8d(1 + \log_2(4d) + \gamma d)}} \ge 2^{\tilde{\gamma} d}
    \]
    for some $ \tilde{\gamma} > 0 $ only depending on $ \gamma $.
    This shows $ M \ge \xc(C_P) - 2d \ge 2^{\tilde{\gamma} d} - 2d $, which yields the claim.
\end{proof}
The above statement together with Corollary~\ref{corollary:milefHardness} implies the following result.
\begin{proposition}
    For every constant $ \gamma > 0 $ there is a constant $ \alpha > 0 $ such that the following holds.
    Let $ \calP $ be any family of $ 0/1 $-polytopes in $ \R^d $ with $ |\calP| \ge 2^{2^{\gamma d}} $.
    Then there exists a polytope $ P \in \calP $ such that every $ \frac{1}{5d} $-MILEF of $ P $ of size at most $
    2^{\alpha d} $ has $ \Omega(\sfrac{n}{\log n}) $ integer variables.
\end{proposition}
Using the fact that there are doubly-exponentially many matroids (see~\cite{Dukes2003}), we thus obtain.
\begin{theorem}
    There is a constant $c>0$ such that the following holds.  Let $n\in \mathbb{Z}_{\geq 1}$, and let $m,k >0$ such that
    there exists a $ \frac{1}{5d} $-MILEF with complexity $(m,k)$ with $m\leq 2^{c\cdot n}$ for any matroid polytope of
    any matroid on a ground set of cardinality $ n $.
    Then $ k=\Omega(\sfrac{n}{\log n})$.
\end{theorem}
 \section{Upper bounds}
\label{secUpperBounds}
In this section, we provide MILEFs for polyhedra considered in Section~\ref{secApplications} that complement some bounds
on the number of integer variables obtained in that section.
To this end, we will consider different polytopes that are convex hulls of characteristic vectors of certain edge
subsets of the complete graph on $ n $ vertices, which we denoted by $ K_n = (V, E) $.
For all these polytopes there exist polynomial-size textbook MILEFs that use $ \Theta(n^2) $ integer variables (usually
they consist of a binary variable for every edge).
However, in what follows we present some (rather non-standard) MILEFs that only use $ O(n \log n) $ or even $ O(n) $
integer variables, respectively.

\subsection{Matching polytope and $ V $-join polytope}
\label{secUpperBoundMatching}
We start by considering the $ V $-join polytope $ \vjoin(n) \subseteq \R^E $ of $ K_n $.
Recall that a $ V $-join is an edge subset $ F \subseteq E $ such that every vertex in $ (V, F) $ has odd degree.
To construct a polynomial-size MILEF for $ \vjoin(n) $ with only $ n $ integer variables, let us fix any orientation $
\mathcal{O} $ of the edges in $ E $ and denote by $ \delta^+(v) \subseteq E $ and $ \delta^-(v) \subseteq E $ the
sets of edges that enter and leave $ v $ according to $ \mathcal{O} $, respectively.
Furthermore, let us write $ \delta(v) \coloneqq \{ e \in E : v \in e \} = \delta^+(v) \cup \delta^-(v) $.
Finally, for any edge set $ F \subseteq E $ and any vector $ x \in \R^E $ we use the notation $ x(F) \coloneqq \sum_{e \in
F} x_e $.
\begin{proposition}\label{prop:MILEFupperBoundVjoin}
    For every $ n $ even, we have
    \ifbool{SODAfinal}{
    \begin{alignat*}{2}
     \vjoin(n) = \conv \Big \{ & x\in [0,1]^E : & \\
				& \exists z\in \Z^V \text{ so that } \, \forall v \in V,  &  \\
			      & x(\delta^+(v)) - x(\delta^-(v)) = 2z_v + 1 & \Big \}. 
    \end{alignat*}
    }{
    \begin{equation*}
        \vjoin(n) = \conv \left \{ x \in [0,1]^E : \exists z \in \Z^V \text{ with } x(\delta^+(v)) - x(\delta^-(v)) =
        2z_v + 1 \text{ for all } v \in V \right \}.
    \end{equation*}
    }
    In particular, $ \vjoin(n) $ admits a MILEF of size $ O(n^2) $ with $ n $ integer variables.
\end{proposition}
\begin{proof}
    Let $ Q \subseteq \R^E $ denote the polytope on the right-hand side.
    To show $ \vjoin(n) \subseteq Q $, it suffices to show that every vertex of $ \vjoin(n) $ is contained in $ Q $.
    To this end, let $ x \in \R^E $ be a vertex of $ \vjoin(n) $.
    Since $ x $ is the characteristic vector of a $ V $-join, for every $ v \in V $ we have that 
    \[
        x(\delta(v)) = x(\delta^+(v)) + x(\delta^-(v))
    \]
    is odd, and so is $ x(\delta^+(v)) - x(\delta^-(v)) $.
    Thus, for every $ v \in V $ there exists an integer $ z_v \in \Z $ that satisfies $ x(\delta^+(v)) - x(\delta^-(v)) = 2z_v
    + 1 $, and hence $ x \in Q $.

    It remains to show $ Q \subseteq \vjoin(n) $, for which it again suffices to show that every vertex of $ Q $ is
    contained in $ \vjoin(n) $.
    To this end, let $ x $ be a vertex of $ Q $.
    Observe that there exists a vector $ z \in \Z^V $ such that $ x $ is a vertex of the polytope
    \ifbool{SODAfinal}{
    \begin{alignat*}{2}
     P_z \coloneqq \Big\{ & \tilde{x} \in [0,1]^E : \quad \forall v\in V, & \\
			  & \tilde{x}(\delta^+(v)) - \tilde{x}(\delta^-(v)) = 2z_v + 1 & \Big\}.
    \end{alignat*}
    }{
    \[
        P_z \coloneqq \left \{ \tilde{x} \in [0,1]^E : \tilde{x}(\delta^+(v)) - \tilde{x}(\delta^-(v)) = 2z_v + 1 \text{ for
        all } v \in V \right \}.
    \]
    }
    Note that $ P_z $ is defined by a totally unimodular matrix (the non-trivial constraints are described by a
    node-arc incidence matrix of the directed graph defined by the orientation $ \mathcal{O} $).
    Thus, since $ z $ is integral, we obtain that $ x \in \{0,1\}^E $.
    Furthermore, for every $ v \in V $ we clearly have that $ x(\delta^+(v)) - x(\delta^-(v)) $ is odd, and so is $
    x(\delta^+(v)) + x(\delta^-(v)) = x(\delta(v)) $.
    This shows that $ x $ is a characteristic vector of a $ V $-join and hence $ x \in \vjoin(n) $.
\end{proof}
As an immediate corollary of Proposition~\ref{prop:MILEFupperBoundVjoin} we obtain the following.
\begin{corollary}
    For every $ n $ even, $ \vjoindom(n) $ admits a MILEF of size $ O(n^2) $ with $ n $ integer variables.
\end{corollary}
This shows that the lower bound provided in Corollary~\ref{corLowerBoundMILEFVjoin} is tight up to a factor of $ O(\log
n) $.

Since the perfect matching polytope $ \pmatch(n) $ of $ K_n $ is a face of $ \vjoin(n) $, and since the matching polytope
$ \match(n) $ of $ K_n $ is equal to $ \{ x \in \R^E_{\ge 0} : x \le y \text{ for some } y \in \pmatch(n) \} $, the
above observation shows that $ \match(n) $ also admits a MILEF of size $ O(n^2) $ with $ n $ integer variables.

Below, we provide an alternative, even simpler MILEF of the same complexity for general graphs.
To this end, let $ G = (V, E) $ be any undirected graph, and fix any orientation $ \mathcal{O} $ of the edges of $ G $.
\begin{proposition}
    \label{propUpperBoundMatching}
    If $P\subseteq \R^E$ is the matching polytope of graph $G$, then 
    \ifbool{SODAfinal}{
    \begin{alignat*}{2}
     P = \conv \Big\{ & x \in \R^E_{\ge 0} : \quad \forall v\in V, & \\
		      & x(\delta(v))\leq 1 \text{ and } \, x(\delta^+(v))\in\Z & \Big \}.
    \end{alignat*}
    }{
    \[
        P = \conv \left \{ x \in \R^E_{\ge 0} : x(\delta(v)) \le 1 \text{ and } x(\delta^+(v)) \in \Z \text{ for every }
        v \in V \right \}.
    \]
    }
    In particular, $ P $ admits a MILEF of size $ O(n^2) $ with $ n $ integer variables.
\end{proposition}
\begin{proof}
    Let $ Q $ denote the polytope on the right-hand side.
    It is clear that $ P $ is contained in $ Q $.
    To show $ Q \subseteq P $, let $ x \in \R^E_{\ge 0} $ that satisfies $ x(\delta(v)) \le 1 $ and $ x(\delta^+(v)) \in
    \Z $ for every $ v \in V $.
    Let $ F \subseteq E $ be the support of $ x $.
    We claim that $ G' \coloneqq (V, F) $ is a bipartite subgraph of $ G $.
    To see this, first observe that $ x(\delta^+(v)) \in \{0,1\} $ for every $ v \in V $.
    Suppose that $ e = \{v,w\} \in F $, and assume that $ e \in \delta^+(v) \cap \delta^-(w) $.
    Since $ e \in F $, we have $ 0 < x_e \le x(\delta^+(v)) $, which implies $ x(\delta^+(v)) = 1 $.
    Furthermore, we have
    \ifbool{SODAfinal}{
    \begin{align*}
     x(\delta^+(w)) &\le x(\delta^+(w)) + 1 - x(\delta(w)) \\
     &= 1 - x(\delta^-(w)) \le 1 - x_e < 1,
    \end{align*}
    }{
    \[
        x(\delta^+(w)) \le x(\delta^+(w)) + 1 - x(\delta(w)) = 1 - x(\delta^-(w)) \le 1 - x_e < 1,
    \]
    }
    and hence $ x(\delta^+(w)) = 0 $.
    Thus, any edge in $ F $ is incident to a node $ v $ with $ \delta^+(v) = 1 $ and a node $ w $ with $ \delta^-(w) = 0
    $, showing that $ G' $ is bipartite.

    Since $ G' $ is bipartite and since $ x(\delta(v)) \le 1 $ holds for every $ v \in V $, the restriction of $ x $ to
    $ \R^F $ is contained in the matching polytope of $ G' $.
    Embedding the matching polytope of $ G' $ into $ \R^E $, we obtain that it is a face of $ P $ and hence $ x $ is
    contained in $ P $.
\end{proof}
For the case of the complete graph, this shows that the lower bound obtained in
Corollary~\ref{corLowerBoundMatchingRdist} is tight up to a factor of $ O(\log n) $.

\subsection{Cut polytope and odd-cut polytope}
\label{secUpperBoundCut}
Next, let us consider the cut polytope $ \cut(n) $ and the odd-cut polytope $ \oddcut(n) $ of $ K_n $ (for the latter we
assume that $ n $ is even).
Recall that a cut is an edge subset $ F \subseteq E $ that can be written as $ F = \left\{ \{v, w\} \in E : v \in S, \, w \notin S \right\} $ for
some $ S \subseteq V $, and it is called an odd cut if $ |S| $ is odd. (We remind the reader that we allow $S=\emptyset$ and $S=V$.)
Let us first start with two simple MILEFs for $ \cut(n) $ and $ \oddcut(n) $ that use $ O(n) $ integer variables.
\begin{proposition}\label{prop:MILEFlowerBoundCut}
    For every $ n $ we have
    \ifbool{SODAfinal}{
        $\cut(n) =$
        \begin{alignat}{2}
         \conv \Big\{  & x \in [0,1]^E : &  \nonumber \\
				& \exists y\in \{0,1\}^V \text{ s.t. } \forall \{v,w\} \in E, & \nonumber \\
				& \quad x_{\{v,w\}} \ge y_v - y_w, & \label{eqMILEFCutFirst} \\
				& \quad x_{\{v,w\}} \ge y_w - y_v,	& \\
				& \quad x_{\{v,w\}} \le y_v + y_w, \text{ and} & \\
				& \quad x_{\{v,w\}} \le 2 - y_v - y_w & \Big\}. \label{eqMILEFCutLast}
        \end{alignat}
    }{
        \begin{alignat}{10}
        \label{eqMILEFCutFirst}
        \cut(n) = \conv \Big\{ x \in [0,1]^E : \ & x_{\{v,w\}} \ge y_v - y_w,                                    & \\
                                                 & x_{\{v,w\}} \ge y_w - y_v,                                    & \\
                                                 & x_{\{v,w\}} \le y_v + y_w, \text{ and}                        & \\
        \label{eqMILEFCutLast}
                                                 & x_{\{v,w\}} \le 2 - y_v - y_w \ \text{ for all } \{v,w\} \in E, & \\
        \nonumber
                                                 & y \in \{0,1\}^V                                               & \Big\}.
    \end{alignat}
    }
    Furthermore, for every $ n $ even we have
    \ifbool{SODAfinal}{
    $\oddcut(n) =$
    \begin{alignat*}{2}
        \conv \Big\{ &x \in [0,1]^E : & \\
				  &\exists y \in \{0,1\}^V \text{ satisfying \eqref{eqMILEFCutFirst}--\eqref{eqMILEFCutLast}}, &\\
				  &\exists z\in \Z \text{ s.t. } \sum \nolimits_{v \in V} x_v = 2z + 1 & \Big\}.
    \end{alignat*}
    }{
    \begin{equation*}
        \oddcut(n) = \conv \left \{ x \in [0,1]^E : y \in \{0,1\}^V, \,
            (x,y) \text{ satisfy \eqref{eqMILEFCutFirst}--\eqref{eqMILEFCutLast}}, \,
            \sum \nolimits_{v \in V} x_v = 2z + 1, \,
            z \in \Z \right \}.
    \end{equation*}
    }
    In particular, both $ \cut(n) $ and $ \oddcut(n) $ admit MILEFs of size $ O(n^2) $ with $ O(n) $ integer variables.
\end{proposition}
\begin{proof}
    Let $ Q $ denote the polytope on the right-hand side of the first claim.
    From the definition of a cut, it is clear that $ \cut(n) $ is contained in $ Q $.
    Let $ x \in [0,1]^E $ and $ y \in \{0,1\}^V $ that satisfy \eqref{eqMILEFCutFirst}--\eqref{eqMILEFCutLast}.
    It is straightforward to check that the integrality of $ y $ forces $ x $ to be integral as well.
    Furthermore, it is easy to see that $ x $ is the characteristic vector of the cut defined by $ S \coloneqq \{ v \in V : y_v
    = 1 \} $.
    Thus, $ x $ is contained in $ \cut(n) $, which shows $ Q \subseteq \cut(n) $.

    The second claim (the description of $ \oddcut(n) $) follows from the above argumentation and the fact that $
    \sum_{v \in V} x_v = 2z + 1, \, z \in \Z $ is equivalent to requiring $ S $ to be odd.
\end{proof}

The above proposition immediately implies the following.
\begin{corollary}\label{cor:MILEFlowerBoundDomOCut}
    For every $ n $ even, the dominant of the odd cut polytope $ \oddcutdom(n) $ admits a MILEF of size $ O(n^2) $ with $
    O(n) $ integer variables.
\end{corollary}
Proposition~\ref{prop:MILEFlowerBoundCut} and Corollary~\ref{cor:MILEFlowerBoundDomOCut} show that the bounds obtained in Corollaries~\ref{corLowerBoundCut} and~\ref{corLowerBoundOddCutDom}, respectively, are tight up to a factor of $ O(\log n) $.
Recall that in our reasoning in Section~\ref{secApplicationBimodular}, we used another MILEF for $ \oddcutdom(n) $,
whose validity we want to prove next. We prove this through the proposition below, which shows that the polytope $P$ defined by~\eqref{eq:bimodularOddCutRelax} satisfies that there is an affine projection of $P\cap (\mathbb{R}^A \times \mathbb{Z}^V \times \mathbb{Z})$ whose convex hull is $\oddcutdom(n)$.
\begin{proposition}
    For every $ n $ even, let $ D = (V, A) $ be the complete digraph on $ n $ vertices.
    Then 
    \ifbool{SODAfinal}{
    $\oddcutdom(n) =$
    \begin{alignat*}{2}
     \conv \Big\{ & x\in \R^E : \, \exists \bar{x} \in \R^A_{\ge 0}, \, y\in\Z^V, \, z\in\Z \text{ s.t.} &\\
		  & x_{\{v,w\}} = \bar{x}_{(v,w)} + \bar{x}_{(w,v)} \quad \forall \{v,w\}\in E &\\
		  & \bar{x}_{(v,w)} \ge y_w - y_v \quad \forall (v,w)\in A &\\
		  & \sum \nolimits_{v\in V} y_v = 2z +1 &\Big\}.
    \end{alignat*}

    }{
    \begin{alignat*}{10}
        \oddcutdom(n) = \conv \Big \{ x \in \R^E : \
        & x_{\{v,w\}} = \bar{x}_{(v,w)} + \bar{x}_{(w,v)} & \text{ for all } \{v,w\} \in E, & \\
        & y_w - y_v \le \bar{x}_{(v, w)} & \text{ for all } (v,w) \in A, & \\
        & \sum \nolimits_{v \in V} y_v = 2 z + 1, \\
        & \bar{x} \in \R^A_{\ge 0}, \, y \in \Z^V, \, z \in \Z & & \Big \}.
    \end{alignat*}
    }
\end{proposition}
\begin{proof}
    Let $ Q $ denote the polyhedron on the right-hand side.
    It is straightforward to check that every characteristic vector of an odd cut is contained in $ Q $.
    As $ Q $ is clearly equal to its dominant, this shows $ \oddcutdom(n) \subseteq Q $.

    To see the reverse inclusion, let us fix $ y \in \Z^V_{\ge 0} $ such that $ \sum_{v \in V} y_v $ is odd.
    It remains to show that the projection onto $\R^E$ of the polyhedron 
    \ifbool{SODAfinal}{
    \begin{alignat*}{2}
     P_y \coloneqq \Big\{ &(x,\bar{x}) \in \R^E \times \R^A_{\ge 0} : \, \forall (v,w)\in A, & \\
			  & \quad x_{\{v,w\}} = \bar{x}_{(v,w)}+\bar{x}_{(w,v)}, & \\
			  & \quad \bar{x}_{(v,w)} \geq y_w - y_v & \Big\}
    \end{alignat*}
    }{
    \[
        P_y \coloneqq \left \{
            (x,\bar{x}) \in \R^E \times \R^A_{\ge 0}
            :
            x_{\{v,w\}} = \bar{x}_{(v,w)} + \bar{x}_{(w,v)}
            \text{ and }
            y_w - y_v \le \bar{x}_{(v, w)} \text{ for all } (v,w) \in A
        \right \}
    \]
    }
     is contained in $\oddcutdom(n)$. To this end, let $ (x, \bar{x}) \in P_y $.
    Let $ \delta $ be the smallest integer such that $ |\{v \in V : y_v = \delta \}| $ is odd.
    Note that such a $ \delta $ exists since $ \sum_{v \in V} y_v $ is odd.
    By the definition of $ \delta $, the set $ S \coloneqq \{ v \in V : y_v \le \delta \} $ has odd cardinality.
    For any $ v \in S $ and any $ w \in V \setminus S $ we have
    \ifbool{SODAfinal}{
    \begin{align*}
     x_{\{v,w\}} &= \bar{x}_{(v,w)} + \bar{x}_{(w,v)} \ge \bar{x}_{(v,w)} \\
     &\ge y_w - y_v \ge (\delta + 1) - \delta = 1.
    \end{align*}

    }{
    \[
        x_{\{v,w\}} = \bar{x}_{(v,w)} + \bar{x}_{(w,v)} \ge \bar{x}_{(v,w)} \ge y_w - y_v \ge (\delta + 1) - \delta = 1.
    \]
    }
    Thus, $ x $ is entry-wise greater than or equal to the characteristic vector of the odd cut induced by $ S $, and hence $
    x \in \oddcutdom(n) $.
\end{proof}

\ifbool{SODAfinal}{}{
\subsection{Traveling salesman polytope}\label{secUpperBoundTSP}
Finally, we argue that there is a polynomial-size MILEF for the traveling salesman polytope $ \tsp(n) $ of $ K_n $ that
only uses $ O(n \log n) $ integer variables.
Let $ \ell \coloneqq \lceil \log_2 n \rceil $ and let us fix any set $ S \subseteq \{0,1\}^\ell $ with cardinality $ n $.
Furthermore, pick any bijective map $ f : S \to V $.
Finally, fix any Hamiltonian cycle $ T \subseteq E $, and consider the polytope
\ifbool{SODAfinal}{
\begin{alignat}{2}
    Q \coloneqq \conv \Big \{ & (y, y', z) \in S \times S \times \{0,1\} : & \label{eqMILEFTspAuxPolytope} \\
			      & y \ne y', & \nonumber \\
			      & z = \left| T \cap \left \{ \{f(y), f(y')\} \right \} \right| & \Big\}. \nonumber
\end{alignat}
}{
\begin{equation}
    \label{eqMILEFTspAuxPolytope}
    Q \coloneqq \conv \left \{ (y_1, y_2, z) \in S \times S \times \{0,1\} :
    y_1 \ne y_2, z = \left| T \cap \left \{ \{f(y_1), f(y_2)\} \right \} \right| \right \}.
\end{equation}
}
We will use $ Q \subseteq \R^\ell \times \R^\ell \times \R $ to construct a MILEF for $ \tsp(n) $ as described in the proposition below. To bound the number of constraints used for this MILEF, we will later show that $Q$ has small extension complexity.
\begin{proposition}
    \label{propMILEFTsp}
    For every $ n $ we have 
    \ifbool{SODAfinal}{
     $\tsp(n) =$ 
     \begin{alignat*}{2}
      \conv \Big\{ & x\in [0,1]^E : \, \exists y:V\rightarrow \{0,1\}^\ell \text{ s.t. } & \\
		  & (y_v, y_w, x_{\{v,w\}}) \in Q \quad \forall \{v,w\} \in E & \Big\}.
     \end{alignat*}
    }{
    \begin{align*}
        \tsp(n) = \conv \{ x \in [0,1]^E : \
            & \exists y_v\in \{0,1\}^\ell \text{ for $v\in V$ such that }\\
            & (y_v, y_w, x_{\{v,w\}}) \in Q \text{ for all } \{v,w\} \in E \}.
    \end{align*}
    }
\end{proposition}
\begin{proof}
    Let $ K $ denote the polytope on the right-hand side of the claim.
    Let $ x \in \{0,1\}^E $ be the characteristic vector of a Hamiltonian cycle $ C \subseteq E $.
    Then there exists a bijective map $ g : V \to V $ such that $ \{v,w\} \in C \iff \{g(v), g(w)\} \in T $.
    For every $ v \in V $ choose $ y_v \in S \subseteq \{0,1\}^\ell $ such that $ f(y_v) = g(v) $.
    Now for every $ \{v,w\} \in E $ we clearly have $ y_v \ne y_w $ as well as
    \ifbool{SODAfinal}{
    \begin{align*}
        x_{\{v,w\}} = 1 &\iff \{v,w\} \in C \\
        &\iff \{g(v), g(w)\} \in T \\
        &\iff \{f(y_v), f(y_w)\} \in T,
    \end{align*}
    }{
    \[
        x_{\{v,w\}} = 1 \iff \{v,w\} \in C \iff \{g(v), g(w)\} \in T \iff \{f(y_v), f(y_w)\} \in T,
    \]
    }
    which means $ x_{\{v,w\}} = \left| T \cap \left \{ \{f(y_v), f(y_w)\} \right \} \right| $.
    Thus, we have $ (y_v, y_w, x_{\{v,w\}}) \in Q $ for every edge $ \{v,w\} \in E $ and hence $ x \in K $.
    This shows $ \tsp(n) \subseteq K $.

    For the reverse inclusion, let $ x \in [0,1]^E $ and consider for $y_v \in \{0,1\}^{\ell}$ for $v\in V$ such that we have $(y_v, y_w, x_{\{v,w\}}) \in Q$ for every edge $ \{v,w\} \in E $.
    Since every vertex $ v $ is incident to some edge, the definition of $ Q $ requires that $ y_v \in S $.
    Furthermore, since every two vertices are adjacent, all $ y_v $ are pairwise distinct.
    Consider the set
    \ifbool{SODAfinal}{
    \begin{alignat*}{2}
     X \coloneqq \Big\{ & (y, y', z) \in S \times S \times \{0,1\} : & \\
			& z = \left| T \cap \left \{ \{f(y), f(y')\} \right \} \right| &\Big\}.
    \end{alignat*}
    }{
    \[
        X \coloneqq \left\{ (y, y', z) \in S \times S \times \{0,1\} :
            z = \left| T \cap \left \{ \{f(y), f(y')\} \right \} \right| \right\}.
    \]
    }
    Fix any edge $ \{v,w\} \in E $ and note that we have $ (y_v, y_w, x_{\{v,w\}}) \in \conv(X) $.
    Thus, there exist some points $ (y_1,y_1',z_1), \dotsc, (y_k, y_k', z_k) \in S $ and coefficients $ \lambda_1,\dotsc,\lambda_k \ge 0 $ with $
    \sum_{i=1}^k \lambda_i = 1 $ such that
    \[
        (y_v, y_w, x_{\{v,w\}}) = \sum_{i=1}^k \lambda_i \cdot (y_i, y_i', z_i)\enspace.
    \]
    Since $ y_v \in \{0,1\}^\ell $ and $ y_1,\dotsc,y_k \in \{0,1\}^\ell $, this implies that $ y_1 = \dotsb = y_k = y_v
    $.
    Analogously, we must also have $ y_1' = \dotsb = y_k' = y_w $.
    By the definition of $ X $, we further have
    \[
        z_i = \left| T \cap \left \{ \{f(y_i), f(y_i')\} \right \} \right| = \left| T \cap \left \{ \{f(y_v), f(y_w)\} \right \} \right|
    \]
    and hence
    \[
        x_{\{v,w\}} = \sum_{i=1}^k \lambda_i \cdot z_i = \left| T \cap \left \{ \{f(y_v), f(y_w)\} \right \} \right|.
    \]
    In other words, we have $ x_{\{v,w\}} \in \{0,1\} $ with $ x_{\{v,w\}} = 1 \iff \{f(y_v), f(y_w)\} \in T $.
    This means that $ x $ is the characteristic vector of the Hamiltonian cycle with edge set $ \left \{ \{v,w\} :
    \{f(y_v), f(y_w)\} \in T \right \} $ (recall that the $ y_v $'s are pairwise distinct).
    Thus, we obtain $ x \in \tsp(n) $ and hence $ K \subseteq \tsp(n) $.
\end{proof}
\begin{corollary}
    For every $ n $, $ \tsp(n) $ admits a MILEF of size $ O(n^4) $ with $ O(n \log n) $ integer variables.
\end{corollary}
\begin{proof}
    By Proposition~\ref{propMILEFTsp} it suffices to show that the polytope $ Q $ defined 
    in~\eqref{eqMILEFTspAuxPolytope} can be described by an extended formulation of size $ O(n^2) $.
    To see this, observe that $ Q $ has $ k \coloneqq |S|(|S| - 1) = O(n^2) $ vertices.
    Since $ Q $ is the set of all convex combinations of its vertices, we have that $ Q $ is the projection of the
    simplex $ \{ x \in \R^k_{\ge 0} : \sum_{i=1}^k x_i = 1 \} $ under the linear map defined by a matrix whose columns
    are the vertices of $ Q $.
    Thus, $ Q $ indeed has an extended formulation of size $ k $.
\end{proof}
This shows that the lower bound obtained in Corollary~\ref{corLowerBoundTSP} is tight up to a factor of $ O(\log^2 n) $.
We are not aware of any polynomial-size MILEF for $ \tsp(n) $ that uses $ o(n \log n) $ integer variables.
}

 \ifbool{SODAfinal}{}{\section{Towards tight bounds}
\label{secTowards}
In this work, we obtained lower bounds on the number of integer variables required in sub-exponential size MILEFs for a
variety of polyhedra by relying on lower bounds on sizes of approximate extended formulations for such polyhedra.
We close our paper by highlighting some gaps left by our techniques.

\paragraph{Stable set polytopes}
For the case of stable set polytopes, we show a lower bound of $ \Omega(\sfrac{\sqrt{n}}{\log n}) $ (for certain
graphs), while we are not aware of any MILEF of sub-exponential size that uses $ o(n) $ integer variables.
In fact, we believe that there exist graphs for which $ \Omega(n) $ integer variables are needed.
This large gap can be explained by our current approach, which simply uses the lower bound for either the cut polytope or matching polytope in a
black-box way by considering stable set polytopes of graphs on $ n $ vertices that have faces that can be affinely
projected onto $ \cut(n') $ or $\match(n')$, respectively, where $ n' = O(\sqrt{n}) $.
A more promising family of graphs to study is the one considered in the recent work of G{\"o}{\"o}s, Jain \&
Watson~\cite{GoosJW2016} who exhibit $ n $-vertex graphs whose stable set polytopes have extension complexities of $2^{\Omega(n / \log n)}$.
Because their work, however, only refers to exact rather than approximate extended formulations, it would require further analysis 
to lift their results to the mixed-integer setting through our techniques.

Despite the fact that we do not believe that all stable set polytopes admit polynomial-size MILEFs with $ o(n) $
integer variables, another motivation for improving the lower bound is the following.
In~\cite[Prop.~2]{KaibelW2015b} it is mentioned that if a family of polytopes $ P $ with vertices in $ \{0,1\}^d $
admits a polynomial-time algorithm to decide whether a point in $ \{0,1\}^d $ belongs to $ P $, then $ P $ can be
described by a MILEF whose size is polynomial in $ d $ and that uses only $ d $ integer variables.
We are not aware of any family of polytopes that shows that the bound on the number of integer variables is
asymptotically tight, but we believe that stable set polytopes are good candidates.

\paragraph{Traveling salesman polytopes}
We proved that every MILEF of sub-exponential size for $ \tsp(n) $ requires at least $ \Omega(\sfrac{n}{\log n}) $
integer variables, while there exists a polynomial-size MILEF with only $ O(n \log n) $ integer variables.
While it is likely that the lower bound can be improved to $ \Omega(n) $, it is not clear to us whether $ \tsp(n) $
admits a polynomial-size MILEF with $ O(n) $ integer variables.

\paragraph{Closing the logarithmic gap: original-space formulations}
Even though we get nearly tight lower bounds on the number of integer variables required in sub-exponential size
MILEFs for the matching polytope, the cut polytope, and the (dominant of the) odd-cut polytope, there is still a gap
remaining.
More precisely, for a graph on $n$ vertices, we show a lower bound of $\Omega(\sfrac{n}{\log n})$ for each of the
above polytopes, whereas there are polynomial-size descriptions using only $O(n)$ many integer variables.
This leaves a logarithmic gap.
We believe that the lower bounds are not tight and $\Omega(n)$ integer variables are needed.

\medskip

Whereas we do not know how to get rid of the $\log n$-factor in general, we can show a stronger lower bound through a different technique for a restricted class of MILEFs for the matching polytope; namely, MILEFs that live in the \emph{original space}, i.e., the same space as the matching polytope. In other words, MILEFs in original space are not allowed to use additional variables. More formally, we say that a MILEF $ (Q, \sigma, \pi)$ for a polyhedron $P$ with $ k $ integer
constraints is \emph{in original space}, if $\pi$ is the identity, i.e., we have the following (see \eqref{eq:defMILEF}):
\[
    P = Q_\sigma \coloneqq \conv(Q\cap \sigma^{-1}(\mathbb{Z}^k)).
\]
For such MILEFs for the matching polytope we show a lower bound of $k=\Omega(n)$. However, we highlight that we derive this linear lower bound only for the matching polytope and in original space, and it is open whether such a technique may extend to general MILEFs and beyond the matching polytope.

Notice that a lower bound of $\Omega(n)$ for the number of integer constraints for small MILEFs of the matching polytope
in the original space is tight (up to a constant factor), because the MILEF given in
Proposition~\ref{propUpperBoundMatching} is in original space.
\begin{theorem}\label{thmLowerBoundMatching2}
    There exists a constant $ \gamma > 0 $ such that any MILEF of $ \match(n) $ of size at most $ 2^{\gamma n} $ in
    original space has $ \Omega(n) $ integer constraints.
\end{theorem}
\begin{proof}
    Recall that there exists a constant $ \gamma > 0 $ such that the extension complexity of $ \match(n) $ is at least $
    2^{\gamma n} $, for every $ n \ge 2 $.
    It suffices to prove that, for any MILEF of $ \match(n) $ in original space of complexity $ (m, k) $, the inequality
    \begin{equation}
        \label{eqProofOriginalSpace}
        m \cdot 2^{3 \gamma k} \ge 2^{\gamma n}
    \end{equation}
    must hold.
    By the definition of $ \gamma $, the inequality clearly holds whenever $ k = 0 $.

    We assume that $ k \ge 1 $ and proceed by induction over $ n \ge 2 $.
    Since any MILEF of $ \match(n) $ of complexity $ (m, k) $ must satisfy $ m \ge 1 $,
    inequality~\eqref{eqProofOriginalSpace} is clearly satisfied whenever $ 3 k \ge n $ and hence it holds if $ n \in
    \{2,3\} $.
    Now, let $ n \ge 4 $ and assume that $ \match(n) $ admits a MILEF of complexity $ (m, k) $ with $ k \ge 1 $.
    That is, denoting by $ G = (V, E) $ the complete undirected graph on $ n $ vertices, there exist matrices $ A \in
    \R^{[m] \times E} $, $ C \in \R^{[k] \times E} $ and vectors $ b \in \R^{[m]} $, $ d \in \Z^k $ such that
    \[
        \match(n) = \conv \left \{ x \in \R^E : Ax \leq b, \, Cx + d \in \Z^k \right\} .
    \]
    We start with some simplifications that can be done over the integrality constraints $Cx+d \in \Z^k$ without loss of
    generality.
    As the vector $ \zerovec $ is contained in $ \match(n) $, the vector $d$ must be integral; thus we can assume it to
    be zero, because $ Cx $ is integral if and only $ Cx + d $ is integral.
    Next, as the characteristic vector $\chi^e$ of each single edge $e$ is contained in $ \match(n) $, we learn that $C$
    is an integral matrix.
    Finally, we remark that we can add to a row of $ C $ an integer multiple of another row,
    and this operation will not change the (non-)integrality of a vector $Cx$.

    Fix an edge $e$ such that the corresponding column in $C$ is not zero.
    By performing integral row operations, as described above, we can assume that there is a single non-zero entry in this column.
    Let $r_e$ be the row corresponding to this non-zero entry, and let $ \bar{C} \in \Z^{[k-1] \times E} $ be the
    collection of all the other rows; hence, $ \bar{C} \chi^e = 0 $.
    We obtain
    \ifbool{SODAfinal}{
    \begin{alignat}{2}
     \match(n) = \conv \Big\{ & x\in \R^E : \, Ax \leq b, & \label{eq:original1} \\
			      & \bar{C} x \in \Z^{k-1}, \, r_e\t x \in \Z & \nonumber \Big\}. 
    \end{alignat}
    }{
    \begin{equation} \label{eq:original1}
        \match(n)
        = \conv \left \{ x \in \R^E : Ax \leq b, \, \bar{C} x \in \Z^{k-1}, \, r_e\t x \in \Z \right \}.
    \end{equation}
    }

    Now, let $F\subseteq E$ contain $e$ and all edges adjacent to it, and let $G'=(V,E\setminus F)$. 
    Let $a_e$ be the column of matrix $A$ corresponding to edge $e$.
    Let $ P' $ be the matching polytope of $ G' $.
    We claim that, if we identify $ P' $ with the face of $\match(n)$ defined by setting $ x_f=0$ for all $f\in F$, we have the
    following identity:
    \ifbool{SODAfinal}{
    \begin{alignat}{2}
     P' = \conv \Big\{ & x \in \R^E : \quad x_f = 0 \quad \forall f \in F, & \label{eq:original2} \\
		      & Ax \leq \min \{b, b - a_e\}, &\nonumber \\
		      & \bar{C} x \in \Z^{k-1} & \Big\}, \nonumber
    \end{alignat}
    }{
    \begin{equation}\label{eq:original2}
        P' = \conv \left \{ x \in \R^E : x_f = 0 \ \forall f \in F, \, Ax \leq \min \{b, b - a_e\}, \, \bar{C} x \in \Z^{k-1} \right \} ,
    \end{equation}
    }
    where the $\min$ operator in $\min\{b, b-a_e\}$ is taken component-wise.
    Note that the proof is complete once we show this, because $ P' $ is linearly isomorphic to $ \match(n - 2) $, and
    hence $ \match(n - 2) $ admits a MILEF of complexity $ (m, k - 1) $.
    By the induction hypothesis, this implies $ m 2^{3 \gamma (k - 1)} \ge 2^{\gamma (n - 2)} $, which yields
    inequality~\eqref{eqProofOriginalSpace}.

    To show that the inclusion ``$\subseteq$'' in~\eqref{eq:original2} holds, consider any vertex  $ x \in P' $.
    Notice that $ x $ is the characteristic vector of a matching in $G'$,
    which augments to a matching in $ G $ when we add edge $ e $.
    Therefore, both $x$ and $y \coloneqq x + \chi^e$ must be in $\match(n) $.
    By using~\eqref{eq:original1}, we deduce that the inequalities $ Ax \leq b $ and $ Ay = Ax + a_e \leq b$
    hold, and so $ Ax \leq \min \{b, b - a_e \} $ holds as well.
    The other conditions on the right-hand side of \eqref{eq:original2} are clearly satisfied for $x$. 

    For the opposite inclusion, let $ x \in \R^E $ satisfying $ x_f = 0 $ for all $ f \in F $, $ Ax \leq \min \{b, b -
    a_e\} $, and $  \bar{C} x \in \Z^{k-1} $.
    Clearly, there must be some $ \lambda \in [0,1] $ such that $ y \coloneqq x + \lambda \chi^e $ satisfies $ r_e\t y
    \in \Z $.
    Furthermore, we have $ \bar{C} y = \bar{C} x \in \Z^{k-1} $ because the column of $ \bar{C} $ that corresponds to $
    e $ is an all-zeros column.
    Moreover, we have the inequality $ Ay = Ax + \lambda a_e \leq \min\{b,b-a_e\} + \lambda a_e \leq b $.
    Thus, the vector $ y $ satisfies all constraints of the formulation in~\eqref{eq:original1}, and so it is contained
    in $ \match(n) $.
    Since $ \match(n) $ is down-closed, $ x $ is also contained in $ \match(n) $.
    Finally, recall that $ x $ satisfies $ x_f = 0 $ for all $ f \in F $ and hence $ x \in P' $.
\end{proof}
 }

\bibliographystyle{plain}

\ifbool{SODAfinal}{}{
\appendix
\section{Relative distance: proofs}
\label{appendixRdist}
In this part, we provide the proofs of Lemmas~\ref{lem:rd2}, \ref{lemma:rdistMaxGap}, \ref{lemma:rdistMinGap}, and
\ref{lemHausdorffBoundedByRdist}.
\begin{proof}[Proof of Lemma~\ref{lem:rd2}]
    In order to prove~\ref{item:rdistSetDef}, let us define
    \[
        f(A, B) \coloneqq \inf \left \{ \lambda \ge 0 : B\subseteq (1+\lambda)A - \lambda A \right \}.
    \]
    It is straightforward to check that $f(A,B)=\rd(A,B)$ in the cases that $A=B=\emptyset$ (value 0) 
    or $A=\emptyset\neq B$ (value $\infty$).
    Thus, in what follows we may assume that both sets are non-empty.
    Since both $ \rdist(A, B) $ and $ f(A, B) $ are non-negative, it suffices to show that $ f(A, B) < \lambda $ implies
    $\rd(A,B)\leq \lambda$, and that $ f(A, B) > \lambda $ implies $\rd(A,B)\geq \lambda$, for any $\lambda>0$.

    Suppose first that $ f(A, B) < \lambda$ holds for some $ \lambda > 0 $.
    Clearly, this implies $ B \subseteq (1+\lambda)A-\lambda A$ and hence for any linear map $ \pi : \R^d\rightarrow \R
    $ we obtain
    \[
        \pi(B) \subseteq \pi \left( (1 + \lambda) A - \lambda A \right) = (1 + \lambda) \pi(A) - \lambda \pi(A).
    \]
    Thus, for any point $b\in B$, there must be points $a,a'\in A$ such that 
    $\pi(b)=(1+\lambda)\pi(a) - \lambda \pi(a')$, or equivalently, $\pi(b) - \pi(a) = \lambda (\pi(a) - \pi(a'))$.
    Recalling that we treat the fraction $\sfrac{0}{0}$ as $0$, we obtain the inequality
    \[
     \lambda \geq \frac{|\pi(b) - \pi(a)|}{|\pi(a) - \pi(a')|} \geq \frac{\inf_{a\in A}|\pi(b) - \pi(a)|}{\diam(\pi(A))}.
    \]
    As this inequality holds for every linear map $ \pi : \R^d \to \R $ and every point $ b \in B $, we obtain
    \[
     \lambda \geq \sup_{\pi:\R^d\rightarrow \R} \frac{\sup_{b\in B} \inf_{a\in A} |\pi(b) - \pi(a)|}{\diam(\pi(A))}
     = \sup_{\pi:\R^d\rightarrow \R} \frac{d_H(A, B)}{\diam(\pi(A))} = \rdist(A,B).
    \]

    Conversely, suppose that $ f(A, B) > \lambda > 0$.
    Clearly, this implies $B\not\subseteq (1+\lambda)A-\lambda A\eqqcolon A'$.
    Let $b\in B\setminus A'$. Moreover, notice that $A'$ is convex, which follows by convexity of $A$.
We now invoke a classic convex separation theorem, see~\cite[Thm.~11.3]{Rockafellar2015}, to \emph{properly} separate $b$ from $A'$. 
More precisely, using that both $A'$ and $\{b\}$ are convex sets whose relative interiors do not intersect---which holds trivially because the relative interior of $\{b\}$ is the empty set---
one can find a hyperplane $H$ that \emph{properly} separates $\{b\}$ from $A'$, which means that
\begin{enumerate*}[label=(\roman*)]\item $b$ is contained in one of the two closed halfspaces defined by $H$, 
\item $A'$ is contained in the other closed halfspace defined by $H$, and 
\item not both $A'$ and $\{b\}$ are fully contained in $H$\end{enumerate*}.
By shifting $H$ to go through $b$, one can assume that $A'$ is not fully contained in $H$.
    This implies that there is a linear map $\pi:\R^d\to \R$ such that
    \[
     \pi(b) \geq \sup_{x\in A'} \pi(x) \quad \text{ and } \quad \pi(b) > \inf_{x\in A'} \pi(x).
    \]
    If $\sup_{a\in A} \pi(a)=\inf_{a\in A} \pi(a)$, then $ \pi(A) = \{\gamma\} $ for some $ \gamma \in \R $, and hence
    $ \diam(\pi(A)) = 0 $.
    Furthermore, we have $ \gamma = \inf_{x\in A'} \pi(x) < \pi(b) $, which implies $ d_H(\pi(A), \pi(B)) > 0 $ and we
    obtain $ \rdist(A, B) = \infty \ge \lambda $.

    Otherwise, if $\sup_{a\in A} \pi(a) > \inf_{a\in A} \pi(a)$, then we have
    \[
        \pi(b) \geq \sup_{x\in A'} \pi(x)
        = (1+\lambda)\sup_{a\in A} \pi(a) - \lambda \inf_{a\in A} \pi(a) = \sup_{a\in A} \pi(a) + \lambda \diam(\pi(A)),
    \]
    where $\diam(\pi(A))> 0$ follows from the assumption $\sup_{a\in A}\pi(a) > \inf_{a\in A} \pi(a)$.
    As $\pi(b)$ is finite, then so are the quantities $\sup_{a\in A} \pi(a)$ and $\diam(\pi(A))$. Finally,
    \[
        \lambda \leq \frac{\pi(b) - \sup_{a\in A} \pi(a)}{\diam(\pi(A))} \leq \frac{d_H(\pi(A), \pi(B))}{\diam(\pi(A))}
        \leq \rd(A,B).
    \]

    \medskip 

    Claim \ref{item:rdistProjSmaller} follows directly from~\ref{item:rdistSetDef} and the fact that every affine map $
    \pi $ satisfies $ \pi((1 + \lambda) A - \lambda A) = (1 + \lambda) \pi(A) + \lambda \pi(A) $.

    \medskip 

    In order to show~\ref{item:rdistTriangleIneq}, let $ R $ denote the right-hand side of the claimed inequality.
    Since $ \rd(A,C) $ and $ R $ are non-negative, it suffices to show that $ \rd(A,C) > \lambda > 0 $ implies $ R \ge
    \lambda $, for any $ \lambda > 0 $.
    By the definition of $ \rd(\cdot) $, note that $ \rd(A,C) > \lambda > 0 $ implies that there exists a linear map $
    \pi : \R^d \to \R $ such that
\begin{equation}\label{eq:rdistTriangleLamLow}
        \lambda \le \frac{d_H(\pi(A), \pi(C))}{\diam(\pi(A))}\enspace,
\end{equation}
    where $ 0 < \diam(\pi(A)) < \infty $ due to $\lambda > 0$.
    In particular, $ \pi(A) $ is a proper interval.
    If $ \pi(C) $ is unbounded, we must have $ \rd(A,B) = \infty $ (if also $ \pi(B) $ is unbounded) or $ \rd(B, C) =
    \infty $ (if $ \pi(B) $ is bounded).
    Thus, if $ \pi(C) $ is unbounded, we have $ R = \infty $ and the inequality holds.

    It remains to consider the case that $ \pi(A) $, $ \pi(B) $, and $ \pi(C) $ are proper intervals. (Notice that these intervals need not be closed.)
    In this case, there exist numbers $ c \leq b \leq a < a' \leq b' \leq c'$ describing the closures of these intervals: $\cl(\pi(A))=[a,a']$, $\cl(\pi(B))=[b,b']$, and $\cl(\pi(C))=[c,c']$.
    Using this notation and setting $ x \coloneqq \max \{ b - c, c' - b' \} $, $ y \coloneqq \max \{a - b, b' - a'\} $, and $ D \coloneqq a'
    - a $, we have
    \[
        \rd(A, B)
        \ge \frac{d_H(\pi(A), \pi(B))}{\diam(\pi(A))}
        = \frac{\max \{ a - b, b' - a' \}}{D}
        = \frac{y}{D}
    \]
    as well as
    \begin{align*}
        \rd(B, C)
        \ge \frac{d_H(\pi(B), \pi(C))}{\diam(\pi(B))}
        = \frac{\max \{ b - c, c' - b' \}}{b' - b}
        = \frac{x}{(b' - a') + D  + (a - b)} \ge \frac{x}{2y + D}.
    \end{align*}
    Thus, we obtain
    \begin{align*}
        R & \ge \frac{y}{D} + \frac{x}{2y + D} + 2 \cdot \frac{y}{D} \cdot \frac{x}{2y + D} = \frac{x + y}{D} \\
        & \ge \frac{\max \{(a - b) + (b - c), (c' - b') + (b' - a')\}}{D} = \frac{d_H(\pi(A), \pi(C))}{\diam(\pi(A))} \\
        & \ge \lambda,
    \end{align*}
    as claimed, where the last inequality follows by~\eqref{eq:rdistTriangleLamLow}.

    \medskip 

    To prove~\ref{item:lemrd2HullUnion}, first notice that the claimed inequality holds trivially if all sets are empty,
    or if $A_i=\emptyset\neq B_i$ for some $i\in [t]$.  We can also ignore any pair of empty sets $A_i=B_i=\emptyset$,
    as its removal does not modify the terms in the inequality.  Thus, we assume in what follows that all sets are
    non-empty.
    It suffices to show that $\max_{i\in [t]} \rd(A_i, B_i) < \lambda$ implies $\rd(\conv(\cup_{i\in [t]} A_i),
    \conv(\cup_{i\in[t]} B_i)) \leq \lambda$, for any $ \lambda > 0 $.

    Suppose that $\max_{i\in [t]} \rd(A_i, B_i) < \lambda$ holds for some $ \lambda > 0 $.
    By~\ref{item:rdistSetDef}, this implies that we have $B_i\subseteq (1+\lambda)A_i - \lambda A_i$ for each $i\in
    [t]$.
    Let $ b \in \conv(\cup_{i\in[t]} B_i) $ and write it as $ b = \sum_{i \in [t]} \mu_i b_i $ for some $
    \mu_1,\dotsc,\mu_t \ge 0 $ with $ \sum_{i\in [t]} \mu_i =1$ and $ b_i \in B_i $ for $ i \in [t] $.
    For every $ i \in [t] $, since $B_i\subseteq (1+\lambda)A_i - \lambda A_i$, there exist $ a_i, a_i' \in A_i $ with $
    b_i = (1 + \lambda)a_i - \lambda a_i' $.
    We obtain
    \begin{align*}
        b & = \sum_{i \in [t]} \mu_i \left( (1 + \lambda a_i) - \lambda a_i' \right) \\
        & = (1 + \lambda) \sum_{i \in [t]} \mu_i a_i - \lambda \sum_{i \in [t]} \mu_i a_i'
        \in (1 + \lambda) \conv(\cup_{i\in [t]} A_i) - \lambda \conv(\cup_{i\in [t]} A_i),
    \end{align*}
    which shows $ B \subseteq (1 + \lambda) \conv(\cup_{i\in [t]} A_i) - \lambda \conv(\cup_{i\in [t]} A_i) $.
    By~\ref{item:rdistSetDef}, this implies that the relative distance of $ \conv(\cup_{i\in [t]} A_i) $ and $
    \conv(\cup_{i\in[t]} B_i) $ is at most $ \lambda $, which completes the proof.
\end{proof}
\begin{proof}[Proof of Lemma~\ref{lemma:rdistMaxGap}]
We shall use the alternative definition of relative distance provided in Lemma \ref{lem:rd2} \ref{item:rdistSetDef}.
Since both $ \rd(A, B) $ and $ \LPgapMax(A, B) $ are non-negative, it suffices to show that $\LPgapMax(A,B)<\lambda$
implies $\rd(A,B)\leq \lambda$, and that $\rd(A,B)< \lambda$ implies $\LPgapMax(A,B)\leq \lambda$, for any $\lambda>0$.

Suppose first that $\LPgapMax(A,B)<\lambda$.
As $A$ is down-closed and $B$ is contained in $\R^d_{\geq 0}$, this implies that $B\subseteq (1+\lambda)A$. 
And as $0\in A$, it is clear that $(1+\lambda)A \subseteq (1+\lambda)A -\lambda A$. 
Therefore, we have the inclusion $B\subseteq (1+\lambda)A -\lambda A$, and the inequality $\rd(A,B)\leq \lambda$.

Conversely, if $\rd(A,B)<\lambda$, we have the inclusion $B\subseteq (1+\lambda)A -\lambda A$. 
As $A$ is down-closed, so is the set $(1+\lambda)A$, and by definition this means that 
$((1+\lambda)A - \lambda A) \cap \R^d_{\geq 0}\subseteq (1+\lambda)A$. 
Thus, since $B$ is contained in $\R^d_{\geq 0}$, we have $B\subseteq (1+\lambda)A$. 
This implies that $\LPgapMax(A,B)\leq \lambda$.
\end{proof}
\begin{proof}[Proof of Lemma~\ref{lemma:rdistMinGap}]
    If $ d' = 1 $, then both $ A $ and $ B $ are proper line segments whose endpoints are $ 0/1 $-points.
    Since no such line segment contains a third $ 0/1 $-point and since $ A \subseteq B $, we obtain $ A = B $.
    Thus, from now on we can assume that $ d' \ge 2 $ holds.

    To prove the two inequalities, we first argue that we may assume that $ A $ and $B$ are full-dimensional.
    To see this, let $H\subset \R^d$ be the affine hull of $A$ (and $B$). 
    If $H\neq \R^d$, then there exists a set $ I \subseteq [d] $ such that $ H = \{ (x_1,\dotsc,x_d) \in \R^d : x_i =
    1 \ \forall i \in I \} $, because $ A $ is up-closed.
    In this case, let $ \pi : \R^d \to \R^{d'} $ denote the projection onto the coordinates in $ [d]\setminus I $.
    It is straightforward to verify that $ \rd(A, B) = \rd(\pi(A), \pi(B)) $ and $ \LPgapMin(A, B) = \LPgapMin(\pi(A),
    \pi(B)) $ hold.
    Assuming that the inequalities hold for full-dimensional sets, we directly obtain the claimed inequalities since $ \dim(\pi(A)) = \dim(\pi(B)) = d' $.
    Thus, we may assume that $ d' = \dim(A) = \dim(B) = d $ holds.

    To show~\ref{itemRdistMinGapLowerBoundRdist}, since $ \rd(A,B) \ge 0 $, it suffices to show that $ \LPgapMin(A, B) > \eps $
    implies $ \rd(A, B) \ge \frac{1}{d - 1} \cdot \frac{\eps}{1 + \eps} $, for any $ \eps \ge 0 $.
    Assume that $ \LPgapMin(A, B) > \eps $ holds for some $ \eps \ge 0 $.
    This implies that there exists a direction $ c = (c_1,\dotsc,c_d) \in \R^d_{\ge 0} $ such that
    \begin{equation}
        \label{eqProofLemmardistMinGap1}
        \min_{a \in A} c\t a > (1 + \eps) \inf_{b \in B} c\t b.
    \end{equation}
    First, observe that~\eqref{eqProofLemmardistMinGap1} implies $ \alpha \coloneqq \min_{a \in A} c\t a > 0 $.
    Second, we argue that we may assume that $ \alpha \ge \frac{\|c\|_1}{d} $ holds.
    To this end, let $ V \subseteq \{0,1\}^d $ denote the vertex set of $ A $.
    We clearly have $ \alpha = \min_{v \in V} c\t v $.
    For every $ i \in [d] $, replace $ c_i $ by the smallest nonnegative number such that the value of $ \min_{v \in V}
    c\t v $ does not change.
    With this modification, we clearly have that~\eqref{eqProofLemmardistMinGap1} is still valid.
    Furthermore, for every $ i \in [d] $ with $ c_i > 0 $ there must exist a point $ v = (v_1,\dotsc,v_d) \in V $ with $
    c\t v = \alpha $ and $ v_i = 1 $, and hence $ \alpha \ge c_i $.
    This implies $ d \cdot \alpha \ge \|c\|_1 $, as claimed.

    Third, since $ A $ is up-closed, it contains the all-ones vector and hence $ \max_{a \in A} c\t a = \|c\|_1 $.
    Denoting by $ \pi : \R^d \to \R $ the linear projection $ x \mapsto c\t x $, we thus obtain $ \pi(A) = [\alpha,
    \|c\|_1] $.
    By inequality~\eqref{eqProofLemmardistMinGap1}, we also have $ \pi(B) = [\beta, \|c\|_1] $, where $ \beta \coloneqq \inf_{b
    \in B} c\t b \le \frac{\alpha}{1 + \eps} $.
    Finally, we establish
    \[
        \rd(A, B)
        \ge \frac{\alpha - \beta}{\|c\|_1 - \alpha}
        \ge \frac{\alpha - \frac{1}{1 + \eps}\alpha}{d\alpha - \alpha}
        = \frac{1}{d-1} \cdot \frac{\eps}{1 + \eps}.
    \]

    To show~\ref{itemRdistMinGapLowerBoundMinGap}, since $ \LPgapMin(A, B) \ge 0 $, it suffices to show that $ \rd(A, B) >
    \lambda $ implies $ \LPgapMin(A, B) \ge \frac{\lambda}{d - 1 - \lambda} $, for any $ \lambda \ge 0 $.
    Assume that $ \rd(A, B) > \lambda $ holds for some $ \lambda \ge 0 $.
    By Lemma~\ref{lem:rd2}~\ref{item:rdistSetDef}, this implies $ B \not \subseteq (1 + \lambda) A - \lambda A $.
    Denoting by $ \onesvec \in A $ the all-ones vector, this in particular means that there exists a $ \bar{b} \in B $ such
    that $ \bar{b} + \lambda \onesvec \notin (1 + \lambda) A $.
    Equivalently, we obtain $ \frac{1}{1 + \lambda} \bar{b} + \frac{\lambda}{1 + \lambda} \onesvec \notin A $.
    Since $ \frac{1}{1 + \lambda} \bar{b} + \frac{\lambda}{1 + \lambda}\onesvec \in [0,1]^d $ and $ A $ is up-closed, we obtain $
    \frac{1}{1 + \lambda} \bar{b} + \frac{\lambda}{1 + \lambda} \onesvec \notin A + \R^d_{\ge 0} $.
    Since $ A + \R^d_{\ge 0} $ is an up-closed polyhedron, there exists a vector $ c \in \R^d_{\ge 0} $ such that
    \[
        c\t \left( \frac{1}{1 + \lambda} \bar{b} + \frac{\lambda}{1 + \lambda} \onesvec \right)
        < \min_{a \in A + \R^d_{\ge 0}} c\t a
        = \min_{a \in A} c\t a
        =: \alpha
    \]
    holds, which is equivalent to
    \[
        c\t \bar{b} < (1 + \lambda) \alpha - \lambda \|c\|_1.
    \]
    Furthermore, note that since $ A $ is an up-closed full-dimensional $ 0/1 $-polytope, it must contain all $ 0/1
    $-points with a support of size $ d - 1 $ and hence $ \frac{d - 1}{d} \onesvec \in A $.
    This clearly implies $ \alpha \le \frac{d - 1}{d} \|c\|_1 $.
    Hence, using the previous inequality, we obtain
    \[
        \inf_{b \in B} c\t b
        \le c\t \bar{b}
        < (1 + \lambda) \alpha - \lambda \|c\|_1
        \le (1 + \lambda) \alpha - \lambda \frac{d}{d - 1} \alpha
        = \frac{d - 1 - \lambda}{d - 1} \alpha.
    \]
    Since $ \inf_{b \in B} c\t b \ge 0 $ and $ \alpha \ge 0 $, we must have $ d - 1 - \lambda > 0 $ and hence
    \[
        \alpha \ge \left( 1 + \frac{\lambda}{d - 1 - \lambda} \right) \inf_{b \in B} c\t b,
    \]
    which shows $ \LPgapMin(A, B) \ge \frac{\lambda}{d - 1 - \lambda} $.
\end{proof}
\begin{proof}[Proof of Lemma~\ref{lemHausdorffBoundedByRdist}]
    Since $ \rdist(A, B) \le 1 $, it suffices to show that $ \rdist(A, B) < \lambda $ for some $ 0 \le \lambda \le 1 $
    implies $ d_H(A, B) \le \sqrt{d} \cdot \frac{\lambda}{1 + \lambda} $.
    Note that, by Lemma~\ref{lem:rd2} \ref{item:rdistSetDef}, $ \rdist(A, B) < \lambda $ implies that for every $ b \in
    B $ there exists an $ a \in A $ such that $ a' \coloneqq \frac{1}{1 + \lambda} b + \frac{\lambda}{1 + \lambda} a \in A $,
    and thus
    \[
        \|b - a'\|_2 = \frac{\lambda}{1 + \lambda} \|b - a\|_2 \le \frac{\lambda}{1 + \lambda} \sqrt{d}\enspace,
    \]
    where the inequality follows from the fact that $ a,b \in [0,1]^d $.
    Thus, for every $ b \in B $ there exists a point $ a' \in A $ with $ \|b - a'\|_2 \le \frac{\lambda}{1 + \lambda}
    \sqrt{d} $, which yields the claim.
\end{proof}
 }

\end{document}